\documentclass[dvipsnames]{lmcs} 
\pdfoutput=1

\usepackage{lastpage}
\lmcsdoi{16}{4}{8}
\lmcsheading{}{\pageref{LastPage}}{}{}%
{Dec.~11,~2019}{Nov.~05,~2020}{}

\usepackage[utf8]{inputenc}

\keywords{multiplayer non-zero-sum games played on graphs; quantitative reachability objectives; subgame perfect equilibria; constrained existence problem}

\usepackage{amsmath,amssymb}

\usepackage{tikz}
\usetikzlibrary{automata,arrows,shapes}
\usetikzlibrary{decorations.pathreplacing}
\usepackage[]{algorithm2e}
\usepackage[]{xcolor}

\DeclareMathOperator{\Cost}{Cost}
\DeclareMathOperator{\Plays}{Plays}
\DeclareMathOperator{\Hist}{Hist}
\DeclareMathOperator{\Succ}{Succ}

\DeclareMathOperator{\extCost}{Cost}
\DeclareMathOperator{\extCostib}{Cost}
\newcommand\extCosti[1]{\extCostib_{#1}}

\def\extended{X}
\DeclareMathOperator{\extGame}{\mathcal{X}}
\def\extV{\ensuremath{V^\extended}}
\def\extE{\ensuremath{E^\extended}}
\def\extVi#1{\ensuremath{V^\extended_{#1}}}
\def\extFi#1{\ensuremath{F^\extended_{#1}}}
\def\extG{X}

\def\regionGraphI#1{\ensuremath{X^{#1}}}
\def\regionGraphGeqI#1{\ensuremath{X^{\geq #1}}}

\def\rest#1#2{\ensuremath{#1_{\restriction#2}}}
\def\outcome#1#2{\ensuremath{\langle #1 \rangle_{#2}}}

\newcommand{\C}[1]{\ensuremath{\mathbb{C}(\lambda^{#1})}}

\newcommand{\N}[1]{\ensuremath{#1 \in \mathbb{N}}}

\newcommand{\mFR}[2]{\ensuremath{\mathrm{mR}(\lambda^{#1}_{#2})}}
\def\maxFR{\ensuremath{\mathrm{mR}(\lambda)}}
\newcommand{\dom}[2]{\ensuremath{\mathrm{Dom}(\lambda^{#1}_{#2})}}

\usepackage{booktabs}

\begin{document}

\title[The Complexity of SPEs in Quantitative Reachability Games]{The Complexity of Subgame Perfect Equilibria in Quantitative Reachability Games}
\titlecomment{{\lsuper*}Work partially supported by the PDR project \emph{Subgame perfection in graph games} (F.R.S.-FNRS), the ARC project \emph{Non-Zero Sum Game Graphs: Applications to Reactive Synthesis and Beyond} (F\'ed\'eration Wallonie-Bruxelles), the EOS project \emph{Verifying Learning Artificial Intelligence Systems} (F.R.S.-FNRS \& FWO), and the COST Action 16228 \emph{GAMENET} (European Cooperation in Science and Technology).}

\author[T.~Brihaye]{Thomas Brihaye\rsuper{a}}
\author[V.~Bruy\`ere]{V\'eronique Bruy\`ere\rsuper{a}}
\author[A.~Goeminne]{Aline Goeminne\rsuper{{a,b}}}
\author[J.-F.~Raskin]{\texorpdfstring{\\}{}Jean-Fran\c{c}ois Raskin\rsuper{b}}
\author[M.~van den Bogaard]{Marie van den Bogaard\rsuper{b}\texorpdfstring{\vspace{-2em}}{}}

\address{\lsuper{a}Universit\'e de Mons (UMONS), Belgium}
\email{\{thomas.brihaye,veronique.bruyere,aline.goeminne\}@umons.ac.be}
\address{\lsuper{b}Universit\'e libre de Bruxelles (ULB), Belgium}
\email{\{jraskin,marie.van.den.bogaard\}@ulb.ac.be}





\begin{abstract}
  \noindent We study multiplayer quantitative reachability games played on a finite directed graph, where the objective of each player is to reach his target set of vertices as quickly as possible. Instead of the well-known notion of Nash equilibrium (NE), we focus on the notion of subgame perfect equilibrium (SPE), a refinement of NE well-suited in the framework of games played on graphs. It is known that there always exists an SPE in quantitative reachability games and that the constrained existence problem is decidable. We here prove that this problem is PSPACE-complete. To obtain this result, we propose a new algorithm that iteratively builds a set of constraints characterizing the set of SPE outcomes in quantitative reachability games. This set of constraints is obtained by iterating an operator that reinforces the constraints up to obtaining a fixpoint. With this fixpoint, the set of SPE outcomes can be represented by a finite graph of size at most exponential. A careful inspection of the computation allows us to establish PSPACE membership.
\end{abstract}

\maketitle

\section*{Introduction}\label{S:one}

  While two-player zero-sum games played on graphs are the most studied model to formalize and solve the reactive synthesis problem~\cite{PnueliR89}, recent work has considered non-zero-sum extensions of this mathematical framework, see e.g.~\cite{KHJ06,FismanKL10,BRS14,KupfermanPV14,BRS-concur15,brenguier_et_al:LIPIcs:2016:6877,DBLP:conf/icalp/ConduracheFGR16,brenguier_et_al:LIPIcs:2017:7806,DBLP:conf/csl/BassetJPRB18}, see also the surveys~\cite{GU08,BrenguierCHPRRS16,Bruyere17}. In the zero-sum game approach, the system and the environment are considered as \emph{monolithic} and fully \emph{adversarial} entities. Unfortunately, both assumptions may turn to be too strong. First, the reactive system may be composed of several components that execute concurrently and have their own purpose. So, it is natural to model such systems with \emph{multiplayer} games with each player having his own objective. Second, the environment usually has its own objective too, and this objective is usually not the negation of the objective of the reactive system as postulated in the zero-sum case. Therefore, there are instances of the reactive synthesis problem for which no solution exists in the zero-sum setting, i.e.\ no winning strategy for the system against a completely antagonistic environment, while there exists a strategy for the system which enforces the desired properties against all \emph{rational} behaviors of the environment pursuing its own objective.

While the central solution concept in zero-sum games is the notion of \emph{winning strategy}, it is well known that this solution concept is not sufficient to reason about non-zero-sum games. In non-zero-sum games, notions of equilibria are used to reason about the rational behavior of players. The celebrated notion of \emph{Nash equilibrium} (NE)~\cite{nash50} is one of the most studied. A profile of strategies is an NE if no player has an incentive to deviate, i.e.\ change his strategy and obtain a better reward, when this player knows that the other players will be playing their respective strategies in the profile. A well-known weakness of NEs in sequential games, which include infinite duration games played on graphs, is that they are subject to \emph{non-credible threats}: decisions in subgames that are irrational and used to threaten the other players and oblige them to follow a given behavior. To avoid this problem, the concept of \emph{subgame perfect equilibrium} (SPE) has been proposed, see e.g.~\cite{osbornebook}.
In the framework of systems subject to bugs, this means that the players will keep playing rationally in all situations, i.e.\ when no bug occurs and after any bug's occurrence.
SPEs are NEs with the additional property that they are also NEs in all subgames of the original game. While it is now quite well understood how to handle NEs algorithmically in games played on graphs~\cite{Ummels08,UmmelsW11,BPS13,DBLP:conf/icalp/ConduracheFGR16}, this is not the case for SPEs.

\subparagraph{\bf Contributions} In this paper, we provide an algorithm to decide in \emph{polynomial space} the constrained existence problem for SPEs in \emph{quantitative reachability games}. A quantitative reachability game is played by $n$ players on a finite graph in which each player has his own reachability objective. The objective of each player is to reach his target set of vertices as quickly as possible. In a series of papers, it has been shown that SPEs always exist in quantitative reachability games~\cite{DBLP:journals/corr/abs-1205-6346}, and that the set of outcomes of SPEs in a quantitative reachability game is a regular language which is effectively constructible~\cite{BrihayeBMR15}. As a consequence of the latter result, the constrained existence problem for SPEs is decidable. The previously mentioned results use the property that for quantitative reachability games, SPEs coincide with \emph{weak SPEs} and \emph{very weak SPEs}~\cite{BrihayeBMR15}. Weak (resp.\ very weak) SPE must be resistant to unilateral deviations of one player that differ from the original one on a finite number of histories only (resp.\ on the initial vertex only).

Unfortunately, the proof in~\cite{BrihayeBMR15} that establishes the regularity of the set of possible outcomes of SPEs in quantitative reachability games exploits a \emph{well-quasi order} for proving termination and it cannot be used to obtain a good upper bound on the complexity for the algorithm. Here, we propose a new algorithm and we show that this set of outcomes can be represented using an automaton of size at most exponential. It follows that the constrained existence problem for SPEs can be decided in PSPACE\@. We also provide a matching lower-bound showing that this problem is PSPACE-complete.

Our new algorithm iteratively builds a set of constraints that exactly characterizes the set of SPEs in quantitative reachability games. This set of constraints is obtained by iterating an operator that reinforces the constraints up to obtaining a fixpoint. A careful inspection of the computation allows us to establish PSPACE membership.

\subparagraph{\bf Related work}

Algorithms to reason on NEs in graph games are studied in~\cite{Ummels08} for $\omega$-regular objectives and in~\cite{UmmelsW11,BPS13} for quantitative objectives. Algorithms to reason on SPEs are given in~\cite{Ummels06} for $\omega$-regular objectives. Quantitative reachability objectives are not $\omega$-regular objectives. Reasoning about NEs and SPEs for $\omega$-regular specifications can also be done using strategy logics~\cite{DBLP:journals/iandc/ChatterjeeHP10,mogavero_et_al:LIPIcs:2010:2859}.

Other notions of rationality and their use for reactive synthesis have been studied in the literature: rational synthesis in cooperative~\cite{FismanKL10} and adversarial~\cite{KupfermanPV14} setting, and their algorithmic complexity has been studied in~\cite{DBLP:conf/icalp/ConduracheFGR16}. Extensions with imperfect information have been investigated in~\cite{DBLP:conf/lics/FiliotGR18}. Synthesis rules based on the notion of admissible strategies have been studied in~\cite{berwanger07,BRS14,BRS-concur15,brenguier_et_al:LIPIcs:2016:6877,brenguier_et_al:LIPIcs:2017:7806,DBLP:conf/csl/BassetJPRB18}.

The restricted class of deviating strategies used in very weak SPEs is a well-known notion that for instance appears in~\cite{kuhn53} with the one-step deviation property. Weak SPEs and very weak SPEs are equivalent notions, but there are games for which there exists a weak SPE but no SPE~\cite{BrihayeBMR15,SV03}. Nevertheless, (very) weak SPEs and SPEs are equivalent for quantitative reachability games, an important property used in the proofs of~\cite{BrihayeBMR15} and of this paper. The equivalence between SPEs and very weak SPEs is also implicitly used as a proof technique in a continuous setting in~\cite{Fudenberg83} and in a lower-semicontinuous setting in~\cite{FleschKMSSV10}.

In~\cite{Bruyere0PR17}, general conditions are given that guarantee the existence of a weak SPE\@. It follows that there always exists a weak SPE for games where players use a prefix-independent payoff function. The computational complexity of the constraint existence problem for weak SPEs in games with $\omega$-regular objectives is studied in~\cite{DBLP:journals/corr/abs-1809-03888}. In~\cite{DBLP:conf/rp/BrihayeBGT19}, the authors focus on relevant SPEs in qualitative and quantitative reachability games where a relevant SPE is either an SPE with cost profile as low as possible, or an SPE which maximizes the social welfare (both maximizes the number of players who reach their target set and minimize the sum of the costs of these players), or is Pareto optimal in the set of all SPE cost profiles. Notice that the obtained complexity results are based on the labeling technique proposed in our paper and that they are the same with or without lower bounds on the cost profiles.

Fixpoint techniques are used several papers to establish the existence of (weak) SPEs in some classes of games like~\cite{FleschKMSSV10,BrihayeBMR15,Bruyere0PR17,DBLP:journals/corr/abs-1809-03888}. However they cannot be used in our context to get the PSPACE complexity result.

\subparagraph{\bf Structure of the paper}

In Section~\ref{section:preliminaries}, we recall the notions of $n$-player graph games and (very weak/weak) SPEs, we introduce the notion of extendend games and we state the studied constrained existence problem. In Section~\ref{section:charac}, we provide a way to characterize the set of plays that are outcomes of SPEs and give an algorithm to construct this set. This algorithm relies on the computation of a sequence of labeling functions ${(\lambda_k)}_{k\in \mathbb{N}}$ until reaching a fixpoint $\lambda^*$ such that the plays which are $\lambda^*$-consistent are exactly the plays which are outcomes of SPEs. In Section~\ref{section:counterGraph}, given a labeling function $\lambda$, we introduce the notion of counter graph in which infinite paths correspond to $\lambda$-consistent plays. We also show that such a counter graph has an exponential size. In Section~\ref{section:PSPACEc}, using counter graphs, we prove the PSPACE-easiness of the constrained existence problem. We then prove the PSPACE-hardness of this problem. A conclusion is provided in the last section.

\medskip
This article is an extended version of an article that appeared in the Proceedings of CONCUR 2019~\cite{BrihayeBGRB19}. All results are proved while most of those proofs were omitted or very briefly sketched in the CONCUR Proceedings.

\section{Preliminaries}%
\label{section:preliminaries}

In this section, we recall the notions of quantitative reachability game and subgame perfect equilibrium. We also state the problem studied in this paper and our main result.

\subsection{Quantitative reachability games}

An \emph{arena} is a tuple $G = (\Pi, V, {(V_i)}_{i\in \Pi}, E) $ where $\Pi = \{ 1, 2, \ldots, n \}$ is a finite set of $n$ players, $V$ is a finite set of vertices with $|V| \geq 2$, ${(V_i)}_{i\in\Pi}$ is a partition of $V$ between the players, and $E \subseteq V \times V$ is a set of edges such that for all $v \in V$ there exists $v'\in V$ such that $(v,v')\in E$. Without loss of generality, we suppose that $|\Pi| \leq |V|$.

A \emph{play} in $G$ is an infinite sequence of vertices $\rho = \rho_0 \rho_1 \ldots$ such that for all $k \in \mathbb{N} \footnote{Throughout this document we assume that $\mathbb{N}$ contains $0$, \emph{i.e.,} $\mathbb{N}= \{0,1,2,\ldots \}.$}$, $(\rho_k, \rho_{k+1}) \in E$.  A \emph{history} is a finite sequence $h = h_0h_1 \ldots h_k$ with $k \in\mathbb{N}$ defined similarly. The \emph{length} $|h|$ of $h$ is the number $k$ of its edges. We denote the set of plays by $\Plays$ and the set of histories by $\Hist$ (when it is necessary, we use notation $\Plays_G$ and $\Hist_G$ to recall the underlying arena $G$). Moreover, the set $\Hist_i$ is the set of histories such that their last vertex $v$ is a vertex of player $i$, i.e.\ $v \in V_i$.

Given a play $\rho = \rho_0 \rho_1 \ldots \in \Plays$ and $k \in \mathbb{N}$, the prefix $\rho_0 \rho_1 \ldots \rho_k$ of $\rho$ is denoted by $\rho_{\leq k}$ and its suffix $\rho_k \rho_{k+1} \ldots$ is denoted by $\rho_{\geq k}$.
A play $\rho$ is called a \emph{lasso} if it is of the form $\rho = h\ell^\omega$ with $h\ell \in \Hist$.
Notice that $\ell$ is not necessary a simple cycle. 

Given an arena $G$, we denote by $\Succ(v) = \{v' \mid (v, v') \in E \}$ the set of \emph{successors} of $v$, for $v \in V$, and by $\Succ^*$ the transitive closure of $\Succ$.

A \emph{quantitative game} $\mathcal{G} = (G,{(\Cost_i)}_{i\in\Pi})$ is an arena equipped with a cost function profile $\Cost = {(\Cost_i)}_{i\in\Pi}$ such that each function $\Cost_i : \Plays \rightarrow \mathbb{R} \cup \{+\infty\}$ assigns a cost to each play. In a quantitative game $\mathcal G$, an initial vertex $v_0\in V$ is often fixed, and we call $(\mathcal{G}, v_0)$ an \emph{initialized game}. A play (resp.\ a history) of $(\mathcal{G},v_0)$ is then a play (resp.\ a history) of $\mathcal{G}$ starting in $v_0$. The set of such plays (resp.\ histories) is denoted by $\Plays(v_0)$ (resp.\ $\Hist(v_0)$).  We also use notation $\Hist_i(v_0)$ when these histories end in a vertex $v \in V_i$.

In this article we are interested in \emph{quantitative reachability games} such that each player has a target set of vertices that he wants to reach. The cost to pay is equal to the number of edges to reach the target set, and each player aims at minimizing his cost.

\begin{defi}[Quantitative reachability game]%
\label{def:cost}
A \emph{quantitative reachability game} is a tuple $\mathcal{G} = (G, {(F_i)}_{i\in \Pi}, {(\Cost_i)}_{i\in\Pi})$ such that
\begin{itemize}
    \item $G$ is an arena;
    \item for each $i \in \Pi$, $F_i \subseteq V$ is the target set of player $i$;
    \item for each $i \in \Pi$ and each $\rho = \rho_0 \rho_1 \ldots \in \Plays$, $\Cost_i(\rho)$ is equal to the least index $k$ such that $\rho_k \in F_i$, and to $+\infty$ if no such index exists.
\end{itemize}
\end{defi}

\noindent
In the rest of this document, we simply call such a game a \emph{reachability game}. Notice that the cost function used for reachability games can be supposed to be continuous in the following sense~\cite{DBLP:journals/corr/abs-1205-6346}. With $V$ endowed with the discrete topology and $V^\omega$ with the product topology, a sequence of plays ${(\rho_n)}_{n\in\mathbb{N}}$ converges to $\rho$  if every prefix of $\rho$ is a prefix of all $\rho_n$ except, possibly, finitely many of them. A cost function $\Cost_i$ is \emph{continuous} if whenever $\lim_{n\rightarrow +\infty} \rho_n = \rho$, we have that $\lim_{n\rightarrow +\infty} \Cost_i(\rho_n) = \Cost_i(\rho)$. In reachability games, the function $\Cost_i$ can be transformed into a continuous one as follows:
\begin{equation} \label{eq:contFunct}\Cost'_i(\rho) = 1- \frac{1}{\Cost_i(\rho)+1} \text{ if } \Cost_i(\rho) < +\infty, \text{ and } \Cost'_i(\rho) = 1 \text{ otherwise.}
\end{equation}

Notice that for the problem we study both cost functions ${(\Cost_i)}_{i\in\Pi}$ or ${(\Cost'_i)}_{i\in \Pi}$ are equivalent. Therefore throughout this paper we use both functions indifferently.

Given a quantitative game $\mathcal G$, a \emph{strategy} of a player $i\in \Pi$ is a function $\sigma_i: \Hist_i \rightarrow V$. This function assigns to each history $hv$, with $v \in V_i$, a vertex $v'$ such that $(v,v') \in E$. In an initialized game $(\mathcal{G},v_0)$, $\sigma_i$ needs only to be defined for histories starting in $v_0$. A play $\rho=\rho_0\rho_1\ldots$ is \emph{consistent} with  $\sigma_i$ if for all $\rho_k \in V_i$ we have that $\sigma_i(\rho_0 \ldots \rho_k) = \rho_{k+1}$. A strategy $\sigma_i$ is \emph{positional} if it only depends on the last vertex of the history, \emph{i.e.}, $\sigma_i(hv) = \sigma_i(v)$ for all $hv \in \Hist_i$. It is \emph{finite-memory} if it can be encoded by a finite-state machine.

Given a quantitative game $\mathcal G$, a \emph{strategy profile} is a tuple $\sigma = {(\sigma_i)}_{i\in \Pi}$ of strategies, one for each player. It is called positional (resp.\ finite-memory) if for all $i \in \Pi$, $\sigma_i$ is positional (resp.\ finite-memory).  Given an initialized game $(\mathcal{G}, v_0)$ and a strategy profile $\sigma$, there exists an unique play from $v_0$ that is consistent with each strategy $\sigma_i$. We call this play the \emph{outcome} of $\sigma$ and it is denoted by $\outcome{\sigma}{v_0}$. Let $c = {(c_i)}_{i \in \Pi} \in {(\mathbb{N}\cup\{+\infty\})}^{|\Pi|}$, we say that $\sigma$ is a strategy profile \emph{with cost} $c$ or that $\outcome{\sigma}{v_0}$ \emph{has cost} $c$ if $c_i = \Cost_i(\outcome{\sigma}{v_0})$ for all $i \in \Pi$.

\subsection{Solution concepts and constraint problem}

In the multiplayer game setting, the solution concepts usually studied are \emph{equilibria} (see~\cite{GU08}). We here recall the concepts of Nash equilibrium and subgame perfect equilibrium.

Let $\sigma = {(\sigma_i)}_{i\in \Pi}$ be a strategy profile in an initialized game $(\mathcal{G},v_0)$. When we highlight the role of player~$i$, we denote $\sigma$ by $(\sigma_i, \sigma_{-i})$ where $\sigma_{-i}$ is the profile ${(\sigma_j)}_{j\in \Pi \setminus \{i\}}$. A strategy $\sigma'_i \neq \sigma_i$ is a \emph{deviating} strategy of player~$i$, and it is a \emph{profitable deviation} for him if $\Cost_i(\outcome{\sigma}{v_0}) > \Cost_i(\outcome{\sigma'_i, \sigma_{-i}}{v_0})$. Hence a deviating strategy is profitable for player~$i$ if it leads to a smaller cost for him when the other players stick to their own strategy.

The notion of Nash equilibrium is classical: a strategy profile $\sigma$ in an initialized game $(\mathcal{G},v_0)$ is a \emph{Nash equilibrium} (NE) if no player has an incentive to deviate unilaterally from his strategy, i.e.\ no player has a profitable deviation. Formally, $\sigma$ is an NE if for each $i \in \Pi$ and each deviating strategy $\sigma'_i$ of player $i$, we have $\Cost_i(\outcome{\sigma}{v_0}) \leq \Cost_i(\outcome{\sigma'_i, \sigma_{-i}}{v_0})$.

When considering games played on graphs, a useful refinement of NE is the concept of \emph{subgame perfect equilibrium} (SPE) which is a strategy profile being an NE in each subgame. It is well-known that contrarily to NEs, SPEs avoid non-credible threats~\cite{GU08}. Formally, given a quantitative game ${\mathcal G} = (G,\Cost)$, an initial vertex $v_0$, and a history $hv \in \Hist(v_0)$, the initialized game $(\rest{\mathcal{G}}{h},v)$ is called a \emph{subgame} of $(\mathcal{G},v_0)$ such that $\rest{\mathcal{G}}{h} = (G, \rest{\Cost}{h})$ and $\Cost_{i\restriction h}(\rho) = \Cost_i(h\rho)$ for all $i \in \Pi$ and $\rho \in V^{\omega}$. Notice that $(\mathcal{G},v_0)$ is subgame of itself. Moreover if $\sigma_i$ is a strategy for player~$i$ in $(\mathcal{G},v_0)$, then $\sigma_{i\restriction h}$ denotes the strategy in $(\rest{\mathcal{G}}{h},v)$ such that for all histories $h'\in \Hist_i(v)$, $\sigma_{i\restriction h}(h') = \sigma_i(hh')$. Similarly, from a strategy profile $\sigma$ in $(\mathcal{G},v_0)$, we derive the strategy profile $\rest{\sigma}{h}$ in $(\rest{\mathcal{G}}{h},v)$.

\begin{defi}[Subgame perfect equilibrium]
	Let $(\mathcal{G},v_0)$ be an initialized game. A strategy profile $\sigma$ is a \emph{subgame perfect equilibrium} in $(\mathcal{G},v_0)$ if for all $hv \in \Hist(v_0)$, $\rest{\sigma}{h}$ is an NE in $(\rest{\mathcal{G}}{h},v)$.
\end{defi}

It is proved in~\cite{DBLP:journals/corr/abs-1205-6346} that there always exists an SPE in reachability games.

\begin{thmC}[\cite{DBLP:journals/corr/abs-1205-6346}]%
\label{thm:existence_SPE}
There exists an SPE in each initialized reachability game.
\end{thmC}

In this paper, we are interested in solving the following \emph{constraint problem}.

\begin{defi}[Constraint problem]
Given ($\mathcal{G},v_0)$ an initialized reachability game and two threshold vectors $x, y \in {(\mathbb{N}\cup \{+\infty\})}^{|\Pi|}$, the constraint problem is to decide whether there exists an SPE in $(\mathcal{G},v_0)$ with cost $c$ such that $x \leq c \leq y$, that is, $x_i \leq c_i \leq y_i$ for all $i \in \Pi$.
\end{defi}

Our main result is the following one:

\begin{thm}%
\label{thm:main}
The constraint problem for initialized reachability games is PSPACE-complete.
\end{thm}

The remaining part of the paper is devoted to the proof of this result. Let us first illustrate the introduced concepts with an example.

\begin{exa}%
\label{ex:fusee}
A reachability game $\mathcal{G} = (G, {(F_i)}_{i\in \Pi}, {(\Cost_i)}_{i\in\Pi})$ with two players is depicted in Figure~\ref{fig:fusee}. The circle vertices are owned by player~$1$ whereas the square vertices are owned by player~$2$. The target sets of both players are respectively equal to $F_1 = \{v_2\}$ (grey vertex), $F_2 = \{v_2,v_5\}$ (double circled vertices).

\begin{figure}[h!]
    \centering
    \scalebox{0.65}{
        \begin{tikzpicture}

        \node[draw,circle, accepting] (v5) at (0,0){$v_5$};
        \node[draw,circle] (v4) at (2,0){$v_4$};
        \node[draw,rectangle] (v0) at (4,0){$v_0$};
        \node[draw,circle] (v1) at (6,0){$v_1$};

        \node[draw,circle] (v6) at (8,0){$v_6$};
        \node[draw,circle] (v7) at (10,0){$v_7$};
        \node[draw,circle, fill=gray, accepting] (v2) at (7,1){$v_2$};
        \node[draw,circle] (v3) at (7,-1){$v_3$};

        \draw[->,double] (v5) to [bend left]  (v4);
        \draw[->] (v4) to [bend left]  (v5);

        \draw[->,double] (v4) to [bend left]  (v0);
        \draw[->] (v0) to [bend left]  (v4);

        \draw[->,double] (v0) to  (v1);
        \draw[->] (v1) to node [below right]{} (v3);
        \draw[->,double] (v1) to (v6);
        \draw[->,double] (v6) to (v7);
        \draw[->,double] (v7) to [bend right] (v2);

        \draw[->,double] (v2) to [bend right](v0);
        \draw[->,double] (v3) to [bend left](v0);

        \draw[->] (v4) to [loop above] (v4);

        \end{tikzpicture}
    }
    \caption{A quantitative reachability game: player 1 (resp.\ player 2) owns  circle (resp.\ square) vertices  and $F_1=\{v_2\}$ and $F_2 = \{v_2,v_5\}$.}%
    \label{fig:fusee}
\end{figure}
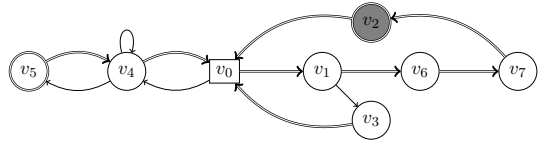

The positional strategy profile $\sigma = (\sigma_1, \sigma_2)$ is depicted by double arrows, its outcome in  $({\mathcal G},v_0)$ is equal to $\outcome{\sigma}{v_0} = {(v_0v_1v_6v_7v_2)}^\omega$ with cost $(4,4)$. Let us explain that $\sigma$ is an NE\@. Player~1 reaches his target set as soon as possible and has thus no incentive to deviate. Player~2 has no profitable deviation that allows him to reach $v_5$. For instance if he uses a deviating positional strategy $\sigma'_2$ such that $\sigma'_2(v_0) = v_4$, then the outcome of $(\sigma_1,\sigma'_2)$ is equal to ${(v_0v_4)}^\omega$ with cost $(+\infty,+\infty)$ which is not profitable for player~$2$.

One can verify that the strategy profile $\sigma$ is also an SPE\@. For instance in the subgame $(\rest{\mathcal{G}}{h},v_5)$ with $h = v_0v_4$, we have $\rho = \outcome{\rest{\sigma}{h}}{v_5} = v_5v_4{(v_0v_1v_6v_7v_2)}^\omega$ such that $\rest{\Cost}{h}(\rho) = \Cost(h\rho) = (8,2)$. In this subgame, with $\rho$, both players reach their target set as soon as possible and have thus no incentive to deviate.

Consider now the positional strategy profile $\sigma' = (\sigma'_1, \sigma'_2)$ such that $\sigma'_1(v_4) = v_0$, $\sigma'_1(v_1) = v_3$, and $\sigma'_2(v_0) = v_4$. Its outcome in $({\mathcal G},v_0)$ is equal to ${(v_0v_4)}^\omega$ with cost $(+\infty,+\infty)$. Let us explain that $\sigma'$ is an NE\@.
On one hand, since player~2 never goes to $v_1$, player~1 has no incentive to deviate, as his target is only accessible from $v_1$.
On the other hand, if player~2 deviates and chooses to go from $v_0$ to $v_1$, his cost is still $+\infty$ since player~1 goes to $v_3$.
However, the strategy profile $\sigma'$ is \emph{not} an SPE\@.
Indeed, it features a \emph{non-credible threat} by player~1: consider the history $v_0v_1$ and the corresponding subgame $(\rest{\mathcal{G}}{v_0},v_1)$.
In that case, player~1 has an incentive to deviate to reach $v_2$ to yield a cost of $4$.
Thus, the strategy profile $\rest{\sigma'}{v_0}$ is not an NE in the subgame $(\rest{\mathcal{G}}{v_0},v_1)$, and thus is not an SPE in $({\mathcal G},v_0)$.
\qed\end{exa}

\subsection{Weak SPE, very weak SPE and extended game}

In this section, we present two important tools that will be repeatedly used in the rest of this paper. First, we explain that in reachability games, the notion of SPE is equivalent to the simpler notion of very weak SPE\@. Second we present an extended version of a reachability game where the vertices are enriched with the set of players that have already visited their target sets along a history. Working with this extended game is essential to prove that the constraint problem for reachability games is in PSPACE\@.

We begin by recalling the concepts of weak and very weak SPE introduced in~\cite{BrihayeBMR15,Bruyere0PR17}. Let $(\mathcal{G},v_0)$ be an initialized game and $\sigma = {(\sigma_i)}_{i\in \Pi}$ be a strategy profile. Given $i \in \Pi$, we say that a strategy $\sigma'_i$ is \emph{finitely deviating} from $\sigma_i$ if  $\sigma'_i$ and $\sigma_i$ only differ on a finite number of histories, and that $\sigma'_i$ is \emph{one-shot deviating} from $\sigma_i$ if $\sigma'_i$ and $\sigma_i$ only differ on the initial vertex $v_0$. A strategy profile $\sigma$ is a \emph{weak NE} (resp.\ \emph{very weak NE}) in $(\mathcal{G},v_0)$ if, for each player $i\in \Pi$, for each finitely deviating (resp.\ one-shot) strategy $\sigma'_i$ of player~$i$ from $\sigma_i$, we have $\Cost_i(\outcome{\sigma}{v_0}) \leq \Cost_i(\outcome{\sigma'_i, \sigma_{-i}}{v_0})$. A strategy profile $\sigma$ is a \emph{weak SPE} (resp.\ \emph{very weak SPE}) in $(\mathcal{G},v_0)$ if, for all $hv \in \Hist(v_0)$, $\rest{\sigma}{h}$ is a weak (resp.\ very weak) NE in $(\rest{\mathcal{G}}{h},v)$.

From the given definitions, every SPE is a weak SPE, and every weak SPE is a very weak SPE\@. It is known that weak SPE and very weak SPE are equivalent notions and that there exist initialized games that have a weak SPE but no SPE;\@ nevertheless, all three concepts are equivalent for initialized reachability games~\cite{BrihayeBMR15,Bruyere0PR17}.

\begin{propC}[\cite{BrihayeBMR15,Bruyere0PR17}]%
\label{prop:SPE-vwSPE}
Let $({\mathcal G},v_0)$ be an initialized reachability game and $\sigma$ be a strategy profile in $({\mathcal G},v_0)$. Then $\sigma$ is an SPE if and only if $\sigma$ is a weak SPE if and only if $\sigma$ is a very weak SPE\@.
\end{propC}

Let us now recall the notion of extended game for a given reachability game $\mathcal{G}$ (see e.g.~\cite{DBLP:journals/corr/abs-1809-03888}).
The vertices $(v,I)$ of the extended game store a vertex $v \in V$ as well as a subset $I \subseteq \Pi$ of players that have already visited their target sets.

\begin{defi}[Extended game]%
\label{def:extGame}
Let $\mathcal{G} = (G, {(F_i)}_{i\in\Pi}, {(\Cost_i)}_{i\in \Pi})$ be a reachability game with an arena $G = (\Pi, V, {(V_i)}_{i\in \Pi}, E)$, and let $v_0$ be an initial vertex. The \emph{extended game} of $\mathcal{G}$ is equal to $\extGame = (\extG,  {(\extFi{i})}_{i\in\Pi}, {(\Cost^X_i)}_{i\in \Pi})$ with the arena $\extG = (\Pi, \extV, {(\extVi{i})}_{i\in\Pi}, \extE)$, such that:
\begin{itemize}
    \item $\extV = V \times 2^\Pi$
    \item $((v,I),(v',I')) \in \extE$ if and only if $(v,v')\in E$ and $I' = I \cup \{i \in \Pi \mid v' \in F_i \}$
    \item $(v,I) \in \extVi{i}$ if and only if $v\in V_i$
    \item $(v,I) \in \extFi{i}$ if and only if $i \in I$
    \item for each $\rho \in \Plays_X$, $\Cost^X_i(\rho)$ is equal to the least index $k$ such that $\rho_k \in \extFi{i}$, and to $+\infty$ if no such index exists.
\end{itemize}
The initialized extended game $(\extGame,x_0)$ associated with the initialized game $({\mathcal G}, v_0)$ is such that $x_0 = (v_0,I_0)$ with $I_0 = \{i \in \Pi \mid v_0 \in F_i \}$.
\end{defi}

Notice the way each target set $\extFi{i}$ is defined: if $v \in F_i$, then $(v,I) \in \extFi{i}$ but also $(v',I') \in \extFi{i}$ for all $(v',I') \in \Succ^*(v,I)$. In the remaining part, to avoid heavy notations, each cost function $\Cost^X_i$ will be simply written as $\Cost_i$.

The extended game of the reachability game of Figure~\ref{fig:fusee} is depicted in Figure~\ref{fig:fuseeExt}. We will come back to this example at the end of this section.

Let us state some properties of the extended game. First, notice that for each $\rho = (v_0, I_0)(v_1,I_1)\ldots \in \Plays_\extG(x_0)$, we have the next property called \emph{$I$-monotonicity}:
\begin{eqnarray} \label{eq:increasing}
I_{k}\subseteq I_{k+1} \quad\quad \mbox{for all } k \in \mathbb{N}.
\end{eqnarray}

Second, given an initialized game $({\mathcal G},v_0)$ and its extended game $({\extGame},x_0)$, there is a one-to-one correspondence between plays in $\Plays_G(v_0)$ and plays in $\Plays_\extG(x_0)$:
\begin{itemize}
\item from $\rho = \rho_0\rho_1 \ldots \in \Plays_G(v_0)$, we derive $\rho^\extended = (\rho_0,I_0)(\rho_1,I_1) \ldots \in \Plays_\extG(x_0)$ such that $I_k$ is the set of players~$i$ that have visited their target set $F_i$ along $\rho_{\leq k}$;
\item from $\rho = (v_0, I_0)(v_1,I_1)\ldots \in \Plays_\extG(x_0)$, we derive $\rho^G = v_0 v_1 \ldots \in \Plays_G(v_0)$ such that the second components $I_k$, $k \in \mathbb{N}$, are omitted.
\end{itemize}

\noindent
Third, given $\rho \in \Plays_G(v_0)$, we have that $\extCost(\rho^\extended) = \Cost(\rho)$, and conversely given $\rho  \in \Plays_\extG(x_0)$, we have that $\Cost(\rho^G) = \extCost(\rho)$.
It follows that
outcomes of SPE can be equivalently studied in $({\mathcal G},v_0)$ and in $({\extGame},x_0)$, as stated in the next lemma.

\begin{lem}%
\label{lem:equivSPE}
If $\rho$ is the outcome of an SPE in $({\mathcal G},v_0)$, then $\rho^\extended$ is the outcome of an SPE in $({\extGame},x_0)$ with the same cost. Conversely, if $\rho$ is the outcome of an SPE in $({\extGame},x_0)$, then $\rho^G$ is the outcome of an SPE in $({\mathcal G},v_0)$ with the same cost.
\end{lem}

Notice that in the same way there is a one-to-one correspondence between strategies in $({\mathcal G},v_0)$ and in its extended game $({\extGame},x_0)$. Thus Lemma~\ref{lem:equivSPE} could be rephrased in term of SPEs instead of SPE outcomes. In Theorem~\ref{thm:folkThm}, we see that outcome characterization is sufficient to deal with SPEs.

By construction, the arena $\extG$ of the initialized extended game is divided into different regions according to the players who have already visited their target set. Let us provide some useful notions with respect to this decomposition. (An illustrative example is given hereafter.)
We will often use them in the following sections. Let $\mathcal{I} = \{ I \subseteq \Pi \mid \text{there exists } v \in V \text{ such that } (v,I) \in \Succ^*(x_0) \}$ be the set of sets $I$ accessible from the initial state $x_0$, and let $N = |\mathcal{I}|$ be its size. For $I, I' \in \mathcal{I}$, if there exists $((v,I),(v',I')) \in \extE$, we say that $I'$ is a \emph{successor} of $I$ and we write $I'\in \Succ(I)$. Given $I \in \mathcal{I}$, $\regionGraphI{I}=(V^I,E^I)$ refers to the sub-arena of $\extG$ restricted to the vertices $\{(v,I) \in \extV \mid v \in V \}$. Hence $\regionGraphI{I}$ has all its vertices with the same second component $I$. We say that $\regionGraphI{I}$ is the \emph{region}\footnote{In the rest of this paper, we indifferently call region either $\regionGraphI{I}$, or $V^I$, or $I$.} associated with $I$. Such a region $\regionGraphI{I}$ is called a \emph{bottom region} whenever $\Succ(I) = I$.

There exists a partial order on $\mathcal{I}$ such that $I < I'$ if and only if $I' \in \Succ^*(I) \setminus \{I\}$.
We fix an arbitrary \emph{total order} on $\mathcal{I}$ that extends this partial order $<$ as follows:
\begin{eqnarray} \label{eq:order}
J_1 < J_2 < \cdots < J_N.
\end{eqnarray}
(with $\regionGraphI{J_N}$ a bottom region).\footnote{We use notation $J_n$, $n \in \{1, \ldots, N\}$, to avoid any confusion with the sets $I_k$ appearing in a play $\rho = (v_0, I_0)(v_1,I_1)\ldots$.} With respect to this total order, given $n \in \{1,\ldots,N\}$, we denote by $\regionGraphGeqI{J_n}= (V^{\geq J_n}, E^{\geq J_n})$ the sub-arena of $\extG$ restricted to the vertices $\{(v,I) \in \extV \mid I \geq J_n \}$.

The total order given in (\ref{eq:order}) together with the $I$-monotonicity (see (\ref{eq:increasing})) leads to the following lemma.

\begin{lem}[Region decomposition and section]\label{lem:region_decomposition}
Let $\pi$ be a path in the arena $X$ of the extended game $\extGame$.
Then there exists a \emph{region decomposition} of $\pi$ as
\[ \pi[\ell]\pi[\ell+1]\dots\pi[m]\]
with $1 \leq \ell \leq m \leq N$, such that for each $n$, $\ell \leq n \leq m$:
\begin{itemize}
    \item $\pi[n]$ is a (possibly empty) path in $X$,
    \item every vertex of $\pi[n]$ is of the form $(v, J_n)$ for some $v \in V$.
\end{itemize}
Each path $\pi[n]$ is called a \emph{section}. The last section $\pi[m]$ is infinite if and only if $\pi$ is infinite.
\end{lem}

\begin{exa}%
\label{example:extendedGame}
Let us come back to the initialized game $(\mathcal G,v_0)$ of Figure~\ref{fig:fusee}. Its extended game $(\extGame,x_0)$ is depicted in Figure~\ref{fig:fuseeExt} (only the part reachable from the initial vertex $x_0 = (v_0,\emptyset)$ is depicted; for the moment the reader should not consider the labeling indicated under the vertices).
As we can see, the extended game is divided into three different regions: one region associated with $I = \emptyset$ that contains the initial vertex $x_0$, a second region associated with $I = \{2\}$, and a third bottom region associated with $I = \{1,2\}=\Pi$. Hence the set ${\mathcal I} = \{\emptyset,\{2\},\Pi\}$ is totally ordered as $J_1 = \emptyset < J_2 = \{2\} < J_3 = \Pi$.

For all vertices $(v,I)$ of the region associated with $I = \{2\}$, we have $(v,I) \not\in \extFi{1}$ and $(v,I) \in \extFi{2}$, and for those of the region associated with $I = \Pi$, we have $(v,I) \in \extFi{1} \cap \extFi{2}$.
The sub-arena $\regionGraphGeqI{J_2}$ of $\extG$ is composed of all vertices $(v,I)$ such that $I=\lbrace 2 \rbrace$ or $I= \Pi$.

From the SPE $\sigma$ given in Example~\ref{ex:fusee} with outcome $\rho = {(v_0v_1v_6v_7v_2)}^\omega \in \Plays_G(v_0)$ and cost $(4,4)$, we derive the SPE outcome $\rho^X \in \Plays_X(x_0)$ equal to  \[ \rho^\extG = (v_0,\emptyset)(v_1,\emptyset)(v_6,\emptyset)(v_7,\emptyset){( (v_2,\Pi)(v_0,\Pi)(v_1,\Pi)(v_6,\Pi)(v_7,\Pi))}^\omega\]
with the same cost $(4,4)$.
The region decomposition of $\rho^X$ is equal to $\rho^X[1]\rho^X[2]\rho^X[3]$ such that its second section $\rho^X[2]$ is empty,
and its two other sections $\rho^X[1]$ and $\rho^X[3]$ are respectively equal to
 $(v_0,\emptyset)(v_1,\emptyset)(v_6,\emptyset)(v_7, \emptyset)$,  and  ${((v_2,\Pi)(v_0,\Pi)(v_1,\Pi)(v_6, \Pi)(v_7,\Pi))}^\omega$.
\qed\end{exa}

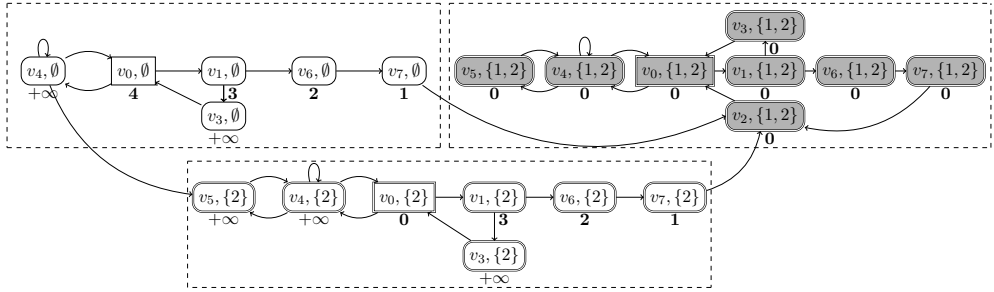
\begin{figure}[h!]
\centering
\scalebox{0.6}{
    \begin{tikzpicture}

    \node[draw] (v0e) at (-12,2.5){$v_0, \emptyset$};
    \node at (-12,2) {$\mathbf{4}$};
    \node[draw,rounded corners=6pt,minimum width=25pt, minimum height=15pt] (v1e) at (-10,2.5){$v_1,\emptyset$};
    \node at (-9.8,2) {$\mathbf{3}$};
    \node[draw,rounded corners=6pt,minimum width=25pt, minimum height=15pt] (v6e) at (-8,2.5){$v_6,\emptyset$};
    \node at (-8,2) {$\mathbf{2}$};
    \node[draw,rounded corners=6pt,minimum width=25pt, minimum height=15pt] (v7e) at (-6,2.5){$v_7,\emptyset$};
    \node at (-6,2) {$\mathbf{1}$};
    \node[draw,rounded corners=6pt,minimum width=25pt, minimum height=15pt] (v3e) at (-10,1.5){$v_3,\emptyset$};
    \node at (-10,1) {$\mathbf{+\infty}$};
    \node[draw,rounded corners=6pt,minimum width=25pt, minimum height=15pt] (v4e) at (-14,2.5){$v_4,\emptyset$};
    \node at (-14,2) {$\mathbf{+\infty}$};

    \node[draw,rounded corners=6pt,minimum width=25pt, minimum height=15pt, accepting] (v5o) at (-10,-0.3){$v_5, \{2\}$};
    \node at (-10,-0.8) {$\mathbf{+\infty}$};
    \node[draw,rounded corners=6pt,minimum width=25pt, minimum height=15pt, accepting] (v4o) at (-8,-0.3){$v_4, \{2\}$};
    \node at (-8,-0.8) {$\mathbf{+\infty}$};
    \node[draw,accepting] (v0o) at (-6,-0.3){$v_0, \{2\}$};
    \node at (-6,-0.8) {$\mathbf{0}$};
    \node[draw,rounded corners=6pt,minimum width=25pt, minimum height=15pt,accepting] (v1o) at (-4,-0.3){$v_1,\{2\}$};
    \node at (-3.8,-0.8) {$\mathbf{3}$};
    \node[draw,rounded corners=6pt,minimum width=25pt, minimum height=15pt,accepting] (v6o) at (-2,-0.3){$v_6,\{2\}$};
    \node at (-2,-0.8) {$\mathbf{2}$};
    \node[draw,rounded corners=6pt,minimum width=25pt, minimum height=15pt,accepting] (v7o) at (0,-0.3){$v_7,\{2\}$};
    \node at (0,-0.8) {$\mathbf{1}$};
    \node[draw,rounded corners=6pt,minimum width=25pt, minimum height=15pt,accepting] (v3o) at (-4,-1.6){$v_3,\{2\}$};
    \node at (-4,-2.1) {$\mathbf{+\infty}$};

    \node[draw,rounded corners=6pt,minimum width=25pt, minimum height=15pt,accepting,fill=gray!60] (v2f) at (2,1.5){$v_2,\{1,2\}$};
    \node at (2.1,1) {$\mathbf{0}$};
    \node[draw,accepting,fill=gray!60] (v0f) at (0,2.5){$v_0, \{1,2\}$};
    \node at (0,2) {$\mathbf{0}$};
    \node[draw,rounded corners=6pt,minimum width=25pt, minimum height=15pt,accepting,fill=gray!60] (v1f) at (2,2.5){$v_1,\{1,2\}$};
    \node at (2,2) {$\mathbf{0}$};
    \node[draw,rounded corners=6pt,minimum width=25pt, minimum height=15pt,accepting,fill=gray!60] (v6f) at (4,2.5){$v_6,\{1,2\}$};
    \node at (4,2) {$\mathbf{0}$};
    \node[draw,rounded corners=6pt,minimum width=25pt, minimum height=15pt,accepting,fill=gray!60] (v7f) at (6,2.5){$v_7,\{1,2\}$};
    \node at (6,2) {$\mathbf{0}$};
    \node[draw,rounded corners=6pt,minimum width=25pt, minimum height=15pt,accepting,fill=gray!60] (v3f) at (2,3.5){$v_3,\{1,2\}$};
    \node at (2.2,3) {$\mathbf{0}$};
    \node[draw,rounded corners=6pt,minimum width=25pt, minimum height=15pt,accepting,fill=gray!60] (v4f) at (-2,2.5){$v_4,\{1,2\}$};
    \node at (-2,2) {$\mathbf{0}$};
    \node[draw,rounded corners=6pt,minimum width=25pt, minimum height=15pt,accepting,fill=gray!60] (v5f) at (-4,2.5){$v_5,\{1,2\}$};
    \node at (-4,2) {$\mathbf{0}$};

    \draw[->] (v0e) -- (v1e);
    \draw[->] (v1e) -- (v3e);
    \draw[->] (v3e) -- (v0e);
    \draw[->] (v1e) -- (v6e);
    \draw[->] (v6e) -- (v7e);

    \draw[->] (v7e) to [bend right] node[right] {} (v2f.west);
    \draw[->] (v1e) -- (v3e);
    \draw[->] (v0e) to[bend left] (v4e);
    \draw[->] (v4e) to[bend left] (v0e);
    \draw[->] (v4e) to[bend right] (v5o);
    \draw[->] (v4e) to [loop above] (v4e);

    \draw[->] (v2f) -- (v0f);
    \draw[->] (v1f) to node [right] {}(v6f);
    \draw[->] (v6f) -- (v7f);
    \draw[->] (v3f) -- (v0f);
    \draw[->] (v1f) -- (v3f);
    \draw[->] (v0f) -- (v1f);
    \draw[->] (v0f) to[bend left] (v4f);
    \draw[->] (v4f) to[bend left] (v0f);
    \draw[->] (v4f) to[bend left] (v5f);
    \draw[->] (v5f) to[bend left] (v4f);
    \draw[->] (v4f) to [loop above] (v4f);
    \draw[->] (v7f) to [bend left] (v2f);

    \draw[->] (v1o) -- (v6o);
    \draw[->] (v6o) -- (v7o);
    \draw[->] (v3o) -- (v0o);
    \draw[->] (v1o) -- (v3o);
    \draw[->] (v0o) -- (v1o);
    \draw[->] (v0o) to[bend left] (v4o);
    \draw[->] (v4o) to[bend left] (v0o);
    \draw[->] (v4o) to[bend left] (v5o);
    \draw[->] (v5o) to[bend left] (v4o);
    \draw[->] (v4o) to [loop above] (v4o);
    \draw[->] (v7o) to [bend right] (v2f);

    \draw[-, dashed] (-10.8,0.5) -- (0.8,0.5);
    \draw[-, dashed] (0.8,0.5) -- (0.8,-2.3);
    \draw[-, dashed] (0.8,-2.3) -- (-10.8,-2.3);
    \draw[-, dashed] (-10.8,-2.3) -- (-10.8,0.5);

    \draw[-,dashed] (-5, 0.8) -- (-5, 4);
    \draw[-,dashed] (-5, 0.8) -- (7, 0.8);
    \draw[-,dashed] (7, 0.8) -- (7, 4);
    \draw[-,dashed] (-5, 4) -- (7, 4);

    \draw[-, dashed] (-5.2,0.8) -- (-5.2,4);
    \draw[-, dashed] (-5.2,4) -- (-14.8,4);
    \draw[-, dashed] (-14.8,4) -- (-14.8,0.8);
    \draw[-, dashed] (-14.8,0.8) -- (-5.2,0.8);
    \end{tikzpicture}
}
\caption{The extended game $(\extGame,x_0)$ for the initialized game $(G,v_0)$ of Figure~\ref{fig:fusee}. The values of a labeling function $\lambda$ are indicated under each vertex.}%
\label{fig:fuseeExt}
\end{figure}

\section{Characterization}%
\label{section:charac}

In this section, given an initialized reachability game $({\mathcal G}, v_0)$, we characterize the set of plays that are outcomes of SPEs, and we provide an algorithm to construct this set. For this characterization, by Lemma~\ref{lem:equivSPE}, we can work on the extended game $({\extGame}, x_0)$ instead of $({\mathcal G}, v_0)$. Moreover by Proposition~\ref{prop:SPE-vwSPE}, we can focus on very weak SPEs only since they are equivalent to SPEs. Such a characterization already appears in~\cite{BrihayeBMR15} however with a different algorithm that cannot be used to obtain good complexity upper bounds for the constraint problem.

All along this section, to avoid heavy notation when we refer to a vertex of $\extV$, we use notation $v$ (instead of $(u,I)$) and notation $I(v)$ means the second component $I$ of this vertex.

Our algorithm iteratively builds a set of \emph{constraints} imposed by a \emph{labeling function} $\lambda : \extV \rightarrow \mathbb{N} \cup \{+\infty \}$ such that the plays of the extended game satisfying those constraints are exactly the SPE outcomes. Let us provide a formal definition of such a function $\lambda$ with the constraints that it imposes on plays.

\begin{defi}[$\lambda$-consistent play]%
\label{def:l-consistant}
Let $\extGame$ be the extended game of a reachability game $\mathcal G$, and $\lambda : \extV \rightarrow \mathbb{N} \cup \{+\infty \}$ be a labeling function. Given $v \in \extV$, for all plays $\rho \in  \Plays_\extG(v)$, we say that $\rho = \rho_0\rho_1 \ldots$ is \emph{$\lambda$-consistent} if for all $n \in \mathbb{N}$ and $i \in \Pi$ such that $\rho_n \in V^\extG_i$: \begin{eqnarray} \label{eq:consistent}
\extCosti{i}(\rho_{\geq n}) \leq \lambda(\rho_n).
\end{eqnarray}
We denote by $\Lambda(v)$ the set of plays $\rho \in  \Plays_\extG(v)$ that are $\lambda$-consistent.
\end{defi}

Thus, a play $\rho$ is $\lambda$-consistent if for all its suffixes $\rho_{\geq n}$, if player~$i$ owns $\rho_n$ then the number of edges to reach his target set $\extFi{i}$ along $\rho_{\geq n}$ is upper bounded by $\lambda(\rho_n)$. Before going into the details of our algorithm, let us intuitively explain on an example how a well-chosen labeling function characterizes the set of SPE outcomes.

\begin{exa}%
\label{ex:lambda}
We consider the initialized extended game $(\extGame,x_0)$ of Figure~\ref{fig:fuseeExt}, and a labeling function $\lambda$ whose values are indicated under each vertex.

If $v \in \extV_i$ is labeled by $\lambda(v) = c$, then if $c \in \mathbb{N}_0$, this means that player~$i$ will only accept outcomes in $(\extGame,v)$ that reach his target set within $c$ steps, otherwise he would have a profitable deviation.
If $\lambda(v) = 0$, this means that player~$i$ has already reached his target set (that is, $i \in I(v)$), and if $\lambda(v) = +\infty$, player~$i$ has no profitable deviation whatever outcome is proposed to him.

In Example~\ref{example:extendedGame} was given the SPE outcome equal to \[ \rho = (v_0,\emptyset)(v_1,\emptyset)(v_6,\emptyset)(v_7,\emptyset){((v_2,\Pi)(v_0,\Pi)(v_1,\Pi)(v_6,\Pi)(v_7,\Pi))}^\omega\] and with cost $(4,4)$. We have $\lambda(v_0,\emptyset) = 4$ and player~$2$ reaches his target set from $(v_0,\emptyset)$ within exactly $4$ steps. For all vertices $v$ of the ending cycle of $\rho$, we have $\lambda(v)=0$ since both players belong to $I(v)= \Pi$.
The constraints imposed by $\lambda$ on the  other vertices of $\rho$ are respected too.

Recall now the strategy profile $\sigma'$ with outcome $\rho' = {((v_0, \emptyset)(v_4,\emptyset))}^\omega$ described in Example~\ref{ex:fusee}.
The outcome $\rho'$ is not $\lambda$-consistent since player~$2$ does not reach his target set, and so in particular does not reach it within $4$ steps.
We already know that $\sigma'$ is not an SPE from Example~\ref{ex:fusee}.
In fact all the profiles that yield $\rho'$ as an outcome are not SPEs.\qed\end{exa}

Our algorithm roughly works as follows: the labeling function $\lambda$ that characterizes the set of SPE outcomes is obtained from an initial labeling function that imposes no constraints, by iterating an operator that reinforces the constraints step after step, up to obtaining a fixpoint which is the required function $\lambda$. Thus, if $\lambda^k$ is the labeling function computed at step $k$ and $\Lambda^k(v)$, $v \in \extV$, the related sets of $\lambda^k$-consistent plays beginning in $v$, initially we have $\Lambda^0(v) = \Plays_X(v)$, and  step by step, the constraints imposed by $\lambda^k$ become stronger and the sets $\Lambda^k(v)$ become smaller, until a fixpoint is reached.

Initially, we want a labeling function $\lambda^0$ that imposes no constraint on the plays. We could define $\lambda^0$ as the constant function $+ \infty$. We proceed a little bit differently. Indeed recall the definition of the target sets ${(\extFi{i})}_{i \in \Pi}$ in an extended game (see Definition~\ref{def:extGame}): $v \in \extFi{i}$ if and only if $i \in I(v)$. Hence, given $\rho = \rho_0\rho_1 \ldots$, once $\rho_k \in \extFi{i}$ for some $k\in \mathbb{N}$ then $\rho_n \in \extFi{i}$ for all $n\geq k$. It follows that for all $n\geq k$, $\extCosti{i}(\rho_{\geq n}) = 0$ and Inequality (\ref{eq:consistent}) is trivially true. (See also Example~\ref{ex:lambda}.)
Therefore we define the labeling function $\lambda^0$ as follows.

\begin{defi}[Initial labeling]%
\label{def:init}
For all $v \in \extV$, let $i \in \Pi$ be such that $v \in \extV_i$,
\[\lambda^0(v) = \begin{cases}
0 & \text{if } i \in I(v) \\
+\infty & \text{otherwise}. \end{cases}\]
\end{defi}

\begin{lem}%
\label{lem:init}
$\rho \in \Lambda_0(v)$ if and only if $\rho \in \Plays_X(v)$.
\end{lem}

Let us now explain how our algorithm computes the labeling functions $\lambda^k$, $k \geq 1$, and the related sets $\Lambda^{k}(v)$, $v \in \extV$. It works in a \emph{bottom-up} manner, according to the total order $J_1 < J_2 < \cdots < J_N$ of $\mathcal I$ given in (\ref{eq:order}): it first iteratively updates the labeling function for all vertices $v$ of the arena $\regionGraphI{J_N}$ until reaching a fixpoint in this arena, it then repeats this procedure in $\regionGraphGeqI{J_{N-1}}$, $\regionGraphGeqI{J_{N-2}}$, $\ldots$, $\regionGraphGeqI{J_1} = \extG$. Hence, suppose that we currently treat the arena $\regionGraphGeqI{J_n}$ and that we want to compute $\lambda^{k+1}$ from $\lambda^{k}$. We define the updated function $\lambda^{k+1}$ as follows (we use the convention that $1 + (+\infty) = +\infty$).

\begin{defi}[Labeling update]%
\label{def:update}
Let $k \geq 0$ and suppose that we treat the arena $\regionGraphGeqI{J_n}$, with $n \in \{1,\ldots,N\}$. For all $v \in \extV$,
\begin{itemize}
    \item if $v \in V^{\geq J_n}$, let $i \in \Pi$ be such that $v \in \extV_i$, then
\[\lambda^{k+1}(v) = \begin{cases}
0 & \text{if } i \in I(v)  \\

1 + \displaystyle\min_{\substack{(v,v')\in \extE}}  \sup \{ \extCosti{i}(\rho) \mid \rho \in \Lambda^{k}(v') \} & \text{otherwise} \end{cases}\]
    \item if $v \not\in V^{\geq J_n}$, then
    \[\lambda^{k+1}(v) = \lambda^{k}(v).\]
    \end{itemize}
\end{defi}

\noindent
Let us provide some explanations. As this update concerns the arena $\regionGraphGeqI{J_n}$, we keep $\lambda^{k+1} = \lambda^{k}$ outside of this arena. Suppose now that $v$ belongs to $\regionGraphGeqI{J_n}$ and let $i$ be such that $v \in \extV_i$. We define $\lambda^{k+1}(v) = 0$ whenever $i \in I(v)$ (as already explained for the definition of $\lambda^0$). When it is updated, the value $\lambda^{k+1}(v)$ represents what is the best cost that player $i$ can ensure for himself from $v$ with a ``one-shot'' choice by only taking into account plays of $\Lambda^{k}(v')$ with $v' \in \Succ(v)$.

Notice that it makes sense to run the algorithm in a bottom-up fashion according to the total ordering $J_1 < \cdots < J_N$ since given a play $\rho = \rho_0\rho_1 \ldots$, if $\rho_0$ is a vertex of $V^{\geq J_n}$, then for all $k \in \mathbb{N}$, $\rho_k$ is a vertex of $V^{\geq J_n}$ (by $I$-monotonicity). Moreover running the algorithm in this way is essential to prove that the constraint problem for reachability games is in PSPACE\@.

\begin{exa}%
\label{ex:update}
We consider again the extended game $(\extGame,x_0)$ of Figure~\ref{fig:fuseeExt} with the total order $J_1  = \emptyset < J_2 = \{2\} < J_3 = J_N = \Pi$ of its set $\mathcal I$.

Let us illustrate Definition~\ref{def:update} on the arena $\regionGraphGeqI{J_2}$. Let $k \geq 0$ and suppose that the labeling function $\lambda^k$ has been computed such that $\lambda^k(v_0,J_2)= 0 $, $\lambda^k(v,J_2)= +\infty$ for every other vertex in region $X^{J_2}$ (notice that $\lambda^k$ is not the labeling function indicated in Figure~\ref{fig:fuseeExt}). Let us show how to compute $\lambda^{k+1}(v_1,J_2)$. We need to compute $\sup \{ \extCosti{1}(\rho) \mid \rho \in \Lambda^{k}(v') \}$ for the two successors $v'$ of $(v_1,J_2)$, that is, respectively $v' = (v_3,J_2)$ and $v' = (v_6,J_2)$.
Recall that $\Lambda^k(v')$ is the set of all plays $\lambda^k$-consistent from $v'$.
All the plays in $\Lambda^k(v_6,J_2)$ have cost $2$ for player~$1$ and player~$2$: indeed, they all first follow the history $(v_6,J_2),(v_7, J_2),(v_2,J_N)$. Thus, $\sup \{ \extCosti{1}(\rho) \mid \rho \in \Lambda^{k}(v_6,J_2) \}=2$.
On the other hand, as $\lambda^k(v_3,J_2)=+\infty$, the play $(v_3,J_2)(v_0,J_2){(v_4,J_2)}^\omega$ is $\lambda^k$-consistent and belongs to $\Lambda^k(v_3,J_2)$. Thus, $\sup \{ \extCosti{1}(\rho) \mid \rho \in \Lambda^{k}(v_3,J_2) \}=+\infty$.
Hence, the minimum is attained with successor $(v_6,J_2)$, and $\lambda^{k+1}(v_1,J_2)=3$.
\qed\end{exa}

We can now provide our algorithm (see Algorithm~\ref{algo:lambda}) that computes the sequence ${(\lambda^k(v))}_{k \in \mathbb{N}}$ until a fixpoint is reached (see Proposition~\ref{prop:fixpoint} below). Initially, the labeling function is $\lambda^0$ (see Definition~\ref{def:init}). For the next steps $k > 0$, we begin with the bottom region $\regionGraphI{J_N}$ of $\extGame$ and update $\lambda^{k-1}$ to $\lambda^k$ as described in Definition~\ref{def:update}. At some point, the values of $\lambda^k$ do not change anymore in $X^{J_N}$ and ${(\lambda^k)}_{k\in\mathbb{N}}$ reaches locally (on $\regionGraphI{J_N}$) a fixpoint (see again Proposition~\ref{prop:fixpoint}). Now, we consider the arena $\regionGraphGeqI{J_{N-1}}$ and in the same way, we continue to update locally the values of $\lambda^k$ in $\regionGraphGeqI{J_{N-1}}$. We repeat this procedure with arenas $\regionGraphGeqI{J_{N-2}}$, $\regionGraphGeqI{J_{N-3}}$, \ldots\ until the arena $\regionGraphGeqI{J_1} = \extG$ is completely processed. From the last computed $\lambda^k$, we derive the sets $\Lambda^k(v)$, $v \in \extV$, that we need for the characterization of outcomes of SPEs (see Theorem~\ref{thm:folkThm} below). An example illustrating the execution of this algorithm is given at the end of this section.

\RestyleAlgo{boxruled}

\begin{algorithm}[!ht]
\caption{Fixpoint}%
\label{algo:lambda}
\DontPrintSemicolon%
compute $\lambda^0$ \;
$k \leftarrow 0$  \;
$n \leftarrow N $  \;
\While{$n \neq 0$}{

\Repeat{$\lambda^k = \lambda^{k-1}$}{
     $k \leftarrow k + 1$ \;
     compute $\lambda^k$ from $\lambda^{k-1}$ with respect to $\regionGraphGeqI{J_{n}}$ \; }
    $n \leftarrow n - 1$ \;}
     return $\lambda^k$
    \end{algorithm}

We now state that the sequence ${(\lambda^k)}_{k\in \mathbb{N}}$ computed by this algorithm reaches a fixpoint --- locally on each arena $\regionGraphGeqI{J_{n}}$ and globally on $X$ --- in the following meaning:
\begin{prop}%
\label{prop:fixpoint}
There exists a sequence of integers $0 = k_N^* < k_{N-1}^* < \cdots < k_1^* = k^*$ such that
\begin{itemize}
    \item \textbf{Local fixpoint:} for all $J_n \in \mathcal{I}$, all $m\in \mathbb{N}$ and all $v \in V^{\geq J_n}$,
\begin{eqnarray} \label{eq:star}
\lambda^{k^*_n+m}(v) = \lambda^{k^*_n}(v),
\end{eqnarray}
\item \textbf{Global fixpoint}: in particular, with $k^* = k_1^*$, for all $m \in \mathbb{N}$ and all $v \in \extV$,
\begin{eqnarray*}
\lambda^{k^*+m}(v) = \lambda^{k^*}(v).
\end{eqnarray*}
The global fixpoint $\lambda^{k^*}$ is also simply denoted by $\lambda^*$, and each set $\Lambda^{k^*}(v), v\in V^\extG$, is denoted by $\Lambda^*(v)$.
\end{itemize}
\end{prop}

\noindent
Before proving Proposition~\ref{prop:fixpoint}, let us make some comments. This proposition indicates that Algorithm~\ref{algo:lambda} terminates. Indeed for each $J_n \in \mathcal{I}$, taking the \emph{least index} $k^*_n$ which makes Equality~(\ref{eq:star}) true shows that the repeat loop is broken and the variable $n$ decremented by $1$. The value $n=0$ is eventually reached and the algorithm stops with the global fixpoint $\lambda^*$. Notice that the first local fixpoint is reached with $k_N^* = 0$ because $X^{J_N}$ is a bottom region.
To prove that the algorithm stops in a finite number of steps, we show that each region requires a finite number of steps to be treated. To do so, we rely on the fact that constraint values on vertices cannot increase from one step to another and that the set of constraint values can be seen as a well-quasi ordering. Thus a decreasing sequence in this set is stationary, and this means our algorithm reaches a fixpoint and terminates.

Proposition~\ref{prop:fixpoint} also shows that when a local fixpoint is reached in the arena $\regionGraphGeqI{J_{n+1}}$ and the algorithm updates the labeling function $\lambda^k$ in the arena $\regionGraphGeqI{J_n}$, the values of $\lambda^k(v)$ do not change anymore for any $v \in V^{\geq J_{n+1}}$ but can still be modified for some $v \in V^{J_{n}}$. Recall also that outside of $\regionGraphGeqI{J_n}$, the values of $\lambda^k(v)$ are still equal to the initial values $\lambda^0(v)$. These properties will be useful when we will prove that the constraint problem for reachability games is in PSPACE\@. They are summarized in the next lemma.

\begin{lem}%
\label{lem:fixpoint-region}
Let $k \in \mathbb{N}$ be a step of the algorithm and let $J_n$ with $n \in \{1,\ldots, N\}$. For all $v \in V^{J_n}$:
\begin{itemize}
    \item if $k \leq k^*_{n+1}$, then $\lambda^{k+1}(v) = \lambda^k(v)=\lambda^0(v)$,
    \item if $k^*_n \leq k$, then $\lambda^{k+1}(v) = \lambda^k(v)=\lambda^{k^*_n}(v)$,
\end{itemize}
Hence the values of $\lambda^k(v)$ and $\lambda^{k+1}(v)$ may be different only when $k^*_{n+1} < k < k^*_n$. \qed%
\end{lem}

To prove Proposition~\ref{prop:fixpoint}, we have to prove that the sequences ${(\lambda^k(v))}_{k\in \mathbb{N}}$, with $v \in \extV$, are non increasing.

\begin{lem}%
\label{lemma:sequenceNonIncreasing}
For all $v\in \extV$, the sequences ${(\lambda^k(v))}_{k\in \mathbb{N}}$ and ${(\Lambda^k(v))}_{k\in \mathbb{N}}$ are non increasing.
\end{lem}

\begin{proof}
Let us prove by induction on $k$ that for all $v\in \extV$
\begin{eqnarray} \label{eq:assert}
\lambda^{k+1}(v) \leq \lambda^k(v).
\end{eqnarray}
We will get that $\Lambda^{k+1}(v) \subseteq \Lambda^k(v)$ by Definition~\ref{def:l-consistant}.

First, recall that $\lambda^k(v) = 0 $ if and only if $i \in I(v)$, where $i$ is the player owning $v$. So, in this case, $\lambda^{k+1}(v) = \lambda^k(v) = 0$ and assertion (\ref{eq:assert}) is proved. It remains to prove this assertion when $i \not\in I(v)$.

For $k=0$, let $v\in \extV_i$ such that $i \not\in I(v)$, then $\lambda^0(v) = + \infty$ and obviously $\lambda^1(v) \leq \lambda^0(v)$.

Suppose that assertion (\ref{eq:assert}) is true for $k$ and let us prove it for $k+1$. We know by induction hypothesis that for all $v'\in \extV$, $\lambda^{k+1}(v') \leq \lambda^k(v')$, and thus also
\begin{equation} \Lambda^{k+1}(v') \subseteq \Lambda^{k}(v'). \label{eq:recNoninc}\end{equation}

Let us prove that for all $v\in \extV_i$ such that $i \not\in I(v)$, $\lambda^{k+2}(v) \leq \lambda^{k+1}(v)$. If $\lambda^{k+2}(v) = \lambda^{k+1}(v)$ or $\lambda^{k+1}(v) = +\infty$, then the assertion is proved. Otherwise, \[\lambda^{k+1}(v) =  1 + \displaystyle\min_{\substack{(v,v')\in \extE}}  \sup \{ \extCosti{i}(\rho) \mid \rho \in \Lambda^{k}(v') \}.\]
By~\eqref{eq:recNoninc}, it follows that
\[\lambda^{k+1}(v) \geq  1 + \displaystyle\min_{\substack{(v,v')\in \extE}}  \sup \{ \extCosti{i}(\rho) \mid \rho \in \Lambda^{k+1}(v') \} = \lambda^{k+2}(v).\]
And so, the assertion again holds.
\end{proof}

We can now prove Proposition~\ref{prop:fixpoint}.

\begin{proof}[Proof of Proposition~\ref{prop:fixpoint}]
The base case is easily proved. Indeed, as $V^{J_N}$ is a bottom region, we have $\lambda_1(v) = \lambda_0(v)$ for all $v \in V^{J_N}$, and thus the local fixpoint is immediately reached on $V^{J_N}$. Hence with $k_N^* = 0$, for all $m \in \mathbb{N}$ and all $v \in V^{J_N}$, $\lambda^{k_N^*+m}(v) = \lambda^{k_N^*}(v)$.

Let $J_n$ be an element of $\mathcal I$, with $n \in \{1,\ldots,N-1\}$. Suppose that a fixpoint has been reached in the arena $\regionGraphGeqI{J_{n+1}}$ and that the labeling function $\lambda^k$ is udpated on the arena $\regionGraphGeqI{J_{n}}$ as described in Definition~\ref{def:update}. Recall (as already summarized in Lemma~\ref{lem:fixpoint-region}) that the previously computed values of $\lambda^k$ do no longer change on $\regionGraphGeqI{J_{n+1}}$ (a local fixpoint is reached) and they do not change outside of $V^{\geq J_n}$ (by construction). However they can be modified on $\regionGraphI{J_{n}}$. In this region $\regionGraphI{J_{n}}$, there are $|V|$ sequences ${(\lambda^k(v))}_{k\in \mathbb{N}}$, $v \in V^{J_n}$, since the vertices $v$ are of the form $v = (u,J_n)$ where $J_n$ is fixed. These sequences are non increasing by Lemma~\ref{lemma:sequenceNonIncreasing}. As the component-wise ordering over ${(\mathbb{N}\cup\{+\infty\})}^{|V|}$ is a well-quasi-ordering, there exists a natural number $k^*_n$ that we choose as small as possible such that for all $v \in V^{J_n}$, $\lambda^{k_n^*+1}(v) = \lambda^{k_n^*}(v)$. This equality also holds for all $v \in V^{\geq J_n}$ (and not only for $V^{J_n}$), and it follows that for all $m \in \mathbb{N}$ and all $v \in V^{\geq J_n}$, $\lambda^{k_n^*+m}(v) = \lambda^{k_n^*}(v)$.

Notice that when $n$ is decremented in Algorithm~\ref{algo:lambda}, $k$ is incremented at least once showing that the sequence $0 < k_N^* < k_{N-1}^* < \cdots < k_1^* = k^*$ is strictly increasing.

Finally, the last arena processed by the algorithm is $\regionGraphGeqI{J_1} = \extG$. So with $k^* = k^*_1$, we have that for all $v\in V$ and all $m \in \mathbb{N}$, $\lambda^{k^*+m}(v) = \lambda^{k^*}(v)$.
\end{proof}

We are ready to state how we characterize plays that are outcomes of SPEs. This is possible with the global fixpoint computed by  Algorithm~\ref{algo:lambda}.

\begin{thm}[Characterization]%
\label{thm:folkThm}
Let $(\mathcal G, v_0)$ be an initialized quantitative game and $(\extGame,x_0)$ be its initialized extended game. Let $\rho^0$ be a play in $\Plays_X(x_0)$. Then $\rho^0$ is the outcome of an SPE in $(\extGame,x_0)$ if and only if $\rho^0 \in \Lambda^*(x_0)$.
\end{thm}

Notice that this theorem also provides a characterization of the outcomes of SPEs in $(\mathcal G, v_0)$ by Lemma~\ref{lem:equivSPE}.

The two implications of Theorem~\ref{thm:folkThm} respectively follow from Proposition~\ref{prop:folkThm1} and Proposition~\ref{prop:folkThm2} given below. We first need to establish an important property satisfied by the sets $\Lambda^k(v)$: when for some player~$i$, the costs $\extCosti{i}(\rho)$ associated with the plays $\rho$ in $\Lambda^k(v)$ are unbounded, there actually exists some play in this set that has an infinite cost. In other terms, either $\Lambda^k(v)$ contains at least one play $\rho$ with an infinite cost $\extCosti{i}(\rho)$ or there exists a constant $\N{c}$ such that $\extCosti{i}(\rho) \leq c$ for all $\rho \in \Lambda^k(v)$.

\begin{prop}\label{prop:sup_max}
For every $\N{k}$, for every $v \in \extV$ and for every $i \in \Pi$, the following implication holds: if  $\sup \{\extCosti{i}(\rho) \mid \rho \in \Lambda^k(v)\} = +\infty$, then there exists a play $\rho \in \Lambda^k(v)$ such that $\extCosti{i}(\rho) = + \infty$.
\end{prop}

\begin{proof}
We assume that $\sup \{\extCosti{i}(\rho) \mid \rho \in \Lambda^k(v)\} = +\infty$, that is, for all $n\in \mathbb{N}$, there exists $\rho^n \in \Lambda^k(v)$ such that $\extCosti{i}(\rho^n) > n$. By K\"onig's lemma, there exist ${(\rho^{n_\ell})}_{\ell \in \mathbb{N}}$ a subsequence of ${(\rho^n)}_{n\in\mathbb{N}}$ and $\rho\in \Plays_X(v)$ such that $\rho = \lim_{\ell \rightarrow +\infty} \rho^{n_\ell}$. Moreover, $\extCosti{i}(\rho) = +\infty$ as $\extCosti{i}$ can be transformed into a continuous function (\cite{DBLP:journals/corr/abs-1205-6346}, see also~\eqref{eq:contFunct}, page~\pageref{eq:contFunct}). Let us prove that \[\rho \in \Lambda^k(v).\] This will establish Proposition~\ref{prop:sup_max}.

We prove by induction on $t$ with $0 \leq t \leq k$ that $\rho \in \Lambda^t(v)$.
If $t=0$, then $\Lambda^0(v) = \Plays_X(v)$ by Lemma~\ref{lem:init}. As $\rho = \lim_{\ell \rightarrow +\infty} \rho^{n_\ell}\in \Plays_X(v)$, it follows that $\rho \in \Lambda^0(v)$.

Let $t > 0$ and assume that the assertion is true for $t-1 < k$. Suppose by contradiction that $\rho \not \in \Lambda^t(v)$.  For all $\ell \in \mathbb{N}$,  as $\rho^{n_\ell} \in \Lambda^k(v)$, we have $\rho^{n_\ell} \in \Lambda^{t}(v)$ by Lemma~\ref{lemma:sequenceNonIncreasing}. Moreover by induction hypothesis we have $\rho \in \Lambda^{t-1}(v)$. From $\rho \in \Lambda^{t-1}(v) \setminus \Lambda^t(v)$, it follows that there exists $m \in \mathbb{N}$ and $j\in \Pi$ with $\rho_m \in V_j$ such that
\begin{equation*}
\begin{array}{ccc}
   \extCosti{j}(\rho_{\geq m}) > \lambda^{t}(\rho_m)  & \text{and} &  \extCosti{j}(\rho_{\geq m}) \leq \lambda^{t-1}(\rho_m).
\end{array}
\end{equation*}
In particular, $\lambda^{t}(\rho_m) < \lambda^{t-1}(\rho_m)$ and thus $\lambda^t(\rho_m) < +\infty$, and player~$j$ does not reach his target set along $\pi = \rho_0 \ldots \rho_m \ldots \rho_{m+\lambda^t(\rho_m)}$. We choose $n_\ell$ large enough such that $\rho$ and $\rho^{n_\ell}$ share a common prefix of length at least $|\pi|$. As $\extCosti{j}(\rho_{\geq m}) > \lambda^{t}(\rho_m)$, it follows that $\extCosti{j}(\rho^{\ell_n}_{\geq m}) > \lambda^{t}(\rho_m)$. We can conclude that $\rho^{n_\ell} \not \in \Lambda^t(v)$ which leads to a contradiction.
\end{proof}

The next corollary is a direct consequence of Proposition~\ref{prop:sup_max} with the convention that the max belongs to $\mathbb{N} \cup \{+ \infty\}$.

\begin{cor}%
\label{cor:sup_max}
$\sup \{\extCosti{i}(\rho) \mid \rho \in \Lambda^k(v)\} = \max \{\extCosti{i}(\rho) \mid \rho \in \Lambda^k(v)\}$.
\end{cor}

The two implications of Theorem~\ref{thm:folkThm} are proved in the following Propositions~\ref{prop:folkThm1} and~\ref{prop:folkThm2}.
Notice that in Proposition~\ref{prop:folkThm1}, we derive the additional property that $\Lambda^*(v)\neq \emptyset$, for all $v\in \Succ^*(x_0)$.
This is necessary to prove Proposition~\ref{prop:folkThm2}.

\begin{prop}%
\label{prop:folkThm1}
If $\sigma$ is an SPE in $(\extGame,x_0)$ then for all $v \in \Succ^*(x_0)$, $\Lambda^*(v) \neq \emptyset$ and $\outcome{\sigma}{x_0}\in \Lambda^*(x_0)$.
\end{prop}
\begin{proof}

Suppose that $\sigma$ is an SPE in $(\extGame,x_0)$ and let us prove by induction on $k$ that for all $k \in \mathbb{N}$ and all $hv \in \Hist_X(x_0)$,
\[\outcome{\rest{\sigma}{h}}{v}\in \Lambda^k(v).\]

For case $k=0$, this is true by definition of $\Lambda^0(v)$ and Lemma~\ref{lem:init}.

Suppose that this assertion is satisfied for $k\geq 0$ and by contradiction, assume that there exists $hv\in \Hist_X(x_0)$ such that $\outcome{\rest{\sigma}{h}}{v}\not\in \Lambda^{k+1}(v)$. Let $\rho = \outcome{\rest{\sigma}{h}}{v} $, by Definition~\ref{def:l-consistant}, it means that there exist $n \in \mathbb{N}$ and $i \in \Pi$ such that $\rho_n \in V_i$ and
\begin{equation}
    \extCosti{i}(\rho_{\geq n}) > \lambda^{k+1}(\rho_n) \label{eq:1}.
\end{equation}
But by induction hypothesis, we know that
\begin{equation}
    \extCosti{i}(\rho_{\geq n}) \leq \lambda^k(\rho_n). \label{eq:2}
\end{equation}
It follows by~\eqref{eq:1} and~\eqref{eq:2} that $\lambda^{k+1}(\rho_n) < \lambda^k(\rho_n)$ and that
\begin{equation*}
    \lambda^{k+1}(\rho_n) < +\infty \text{, } \lambda^k(\rho_n) \neq 0 \text{ and } \lambda^{k+1}(\rho_n) \neq 0.
\end{equation*}
In regards of Definition~\ref{def:update}, these three relations allow us to conclude that
\[\lambda^{k+1}(\rho_n) = 1 +  \min_{\substack{(\rho_n,v')\in \extE}} \sup\{\extCosti{i}(\rho') \mid \rho' \in \Lambda^{k}(v')\}. \]
Let $v'\in V$ be such that
\begin{equation} \label{eq:3}
\lambda^{k+1}(\rho_n) - 1 = \sup\{\extCosti{i}(\rho') \mid \rho' \in \Lambda^{k}(v')\}.
\end{equation}
Let $h' = h\rho_0\ldots \rho_{n-1} \in \Hist_X(x_0)$, and let us define the one-shot deviating strategy $\tau_i$ from $\sigma_{i\restriction h'}$ such that $\tau_i(\rho_n) = v'$. Let us prove that $\tau_i$ is a one-shot profitable deviation for player $i$ in $(\rest{\extGame}{h'},\rho_n)$.
\begin{align*}
    &\extCosti{i}(h'\outcome{\tau_i,\sigma_{-i|h'}}{\rho_n}) \\
    &= \extCosti{i}(h'\rho_n \outcome{\rest{\sigma}{h'\rho_n}}{v'}) \tag{as $\tau_i$ is a one-shot deviating strategy}\\
    &= |h'\rho_{n}v'| + \extCosti{i}(\outcome{\rest{\sigma}{h'\rho_n}}{v'}) \tag{as $\lambda^{k+1}(\rho_n) \neq 0$}\\ 
    &\leq |h'\rho_{n}v'| +\lambda^{k+1}(\rho_n)-1 \tag{by~\eqref{eq:3} and as $\outcome{\rest{\sigma}{h'\rho_n}}{v'} \in \Lambda^k(v')$ by ind.\ hyp.} \\ 
    &< |h'\rho_n|+\extCosti{i}(\rho_{\geq n}) & \tag{by~\eqref{eq:1}}\\
    &= \extCosti{i}(h'\outcome{\rest{\sigma}{h'}}{\rho_n})
\end{align*}
This proves that $\sigma$ is not a very weak SPE and so not an SPE by Proposition~\ref{prop:SPE-vwSPE}, which is a contradiction. This concludes the proof.
\end{proof}

\begin{prop}%
\label{prop:folkThm2}
Let $\rho^0 \in \Lambda^*(x_0)$, then $\rho^0$ is the outcome of an SPE in $(\extGame,x_0)$.
\end{prop}

\begin{proof}

By Theorem~\ref{thm:existence_SPE}, there exists an SPE in $(\extGame,x_0)$ and thus, by Proposition~\ref{prop:folkThm1}:
\begin{align}
    \Lambda^*(v) \neq \emptyset \text{ for all } v\in \Succ^*(x_0). \label{eq:Lambda_non_Vide}
\end{align}

Let $\rho^0 \in \Lambda^*(x_0)$ and let us show how to construct a very weak SPE $\sigma$ (and so an SPE by Proposition~\ref{prop:SPE-vwSPE}) with outcome $\rho^0$ in $(\extGame,x_0)$. We define $\sigma$ step by step by induction on the subgames of $(\extGame,x_0)$. We first partially build $\sigma$ such as it produces $\rho^0$, i.e., $\outcome{\sigma}{x_0} = \rho^0$.
Now, we define a set of plays which is useful to define $\sigma$ in the subgames.
For all $(i,v')$ such that $(v,v') \in \extE$ with $v,v' \in \Succ^*(x_0)$ and $v \in \extV_i$, we take $\rho_{i,v'} \in \Lambda^{*}(v')$ such that:
\begin{equation}
\extCosti{i}(\rho_{i,v'}) = \max \{\extCosti{i}(\rho') \mid \rho' \in \Lambda^{*}(v')\} \label{eq:argmax}
\end{equation}
Notice that such a play exists by~(\ref{eq:Lambda_non_Vide}) and Corollary~\ref{cor:sup_max}.
The construction of $\sigma$ is done by induction as follows. Consider $hvv' \in \Hist_X(x_0)$ with $v \in \extV_i$ such that $\outcome{\rest{\sigma}{h}}{v}$ is already defined but not yet $\outcome{\rest{\sigma}{hv}}{v'}$. We extend the definition of $\sigma$ as follows:
\begin{equation*}
\outcome{\rest{\sigma}{hv}}{v'} = \rho_{i,v'}
\end{equation*}

Let us prove that $\sigma$ is a very weak SPE\@. Consider the subgame $(\rest{\extGame}{h},v)$ for a given $hv \in \Hist_X(x_0)$ with $v \in \extV_i$ and the one-shot deviating strategy $\sigma'_i$ from $\sigma_{i|h}$ such that $\sigma'_i(v) = v'$. By construction, there exists $\rho_{j,u}$ and $\rho_{i,v'}$ as defined previously and $g \in \Hist_X(x_0)$ such that  $h\outcome{\rest{\sigma}{h}}{v} = g\rho_{j,u}$ and $\outcome{\rest{\sigma}{hv}}{v'} = \rho_{i,v'}$.

We prove that $\sigma'_i$ is not a one-shot profitable deviation by proving that $\extCosti{i}(h\outcome{\rest{\sigma}{h}}{v}) \leq \extCosti{i}(hv \outcome{\rest{\sigma}{hv}}{v'})$. If $i \in I(v)$, then obviously $\extCosti{i}(h\outcome{\rest{\sigma}{h}}{v}) = \extCosti{i}(hv \outcome{\rest{\sigma}{hv}}{v'})$. Otherwise $i \not\in I(v)$.
We have that $\rho = \outcome{\rest{\sigma}{h}}{v}$ is suffix of $\rho_{j,u}$, that is, $\rho = \rho_{j,u,\geq n}$ for some $n \in \mathbb{N}$ and that $\rho_{j,u}\in \Lambda^*(u) = \Lambda^{k^*}(u) = \Lambda^{k^*+1}(u)$ (by~\eqref{eq:argmax} and the fixpoint).
It follows that:
\begin{align}
\extCosti{i}(\rho)   &= \extCosti{i}(\rho_{j,u,\geq n}) \nonumber \\
                &\leq \lambda^{k^*+1}(v) \nonumber \\
               &= 1 + \min_{\substack{(v,w)\in \extE}} \sup\{\extCosti{i}(\rho') \mid \rho'\in \Lambda^{k^*}(w)\}  \nonumber\\
               &\leq 1 + \sup\{\extCosti{i}(\rho') \mid \rho'\in \Lambda^{k^*}(v')\}  \nonumber\\
               &= 1 + \extCosti{i}(\rho_{i,v'}). \label{eq:good}
\end{align}
Thus we have that
\begin{align*}
    \extCosti{i}(h\outcome{\rest{\sigma}{h}}{v})
    &= |hv|+\extCosti{i}(\rho) \\
    &\leq |hv|+1+ \extCosti{i}(\rho_{i,v'}) \tag{by~\eqref{eq:good}} \\
    &= \extCosti{i}(hv\rho_{i,v'})\\
    &= \extCosti{i}(hv\outcome{\rest{\sigma}{hv}}{v'}).
\end{align*}
This concludes the proof.
\end{proof}

\begin{exa}%
\label{ex:algo}
Let us come back to the running example of Figure~\ref{fig:fuseeExt} and illustrate Proposition~\ref{prop:fixpoint}. The different steps of Algorithm~\ref{algo:lambda} are given in Table~\ref{table:fusee}. The columns indicate the vertices according to their region, respectively $\Pi$, $\{2\}$, and $\emptyset$. Notice that for the region $\Pi$, we only write one column $v$ as for all vertices $(v,\Pi)$ the value of $\lambda$ is equal to 0 all along the algorithm.

Recall that $J_1  = \emptyset < J_2 = \{2\} < J_3 = \Pi =\{1,2\}$. The algorithm begins with the arena $\regionGraphI{J_{3}}$. A local fixpoint ($\lambda^1 = \lambda^0$) is immediately reached because all vertices belong to the target set of both players in $X^{J_3}$.
Thus the first local fixpoint is reached with $k_3^* = 0$.

The algorithm then treats the arena $\regionGraphGeqI{J_{2}}$. By Lemma~\ref{lem:fixpoint-region}, it is enough to consider the region $\regionGraphI{J_{2}}$. Let us explain how to compute $\lambda^2(v)$ from $\lambda^1(v)$ on this region.
For $v =(v_7, \{2 \})$, we have that $\lambda^2(v)= 1 + \min_{(v,v') \in E^X} \sup \{ \Cost_1(\rho) \mid \rho \in \Lambda^1(v') \}$. As the unique successor of $v$ is $(v_2,\{1,2\})$, all $\lambda^1$-consistent plays beginning in this successor have cost $0$ for player~$1$. So, we have that $\lambda^2(v) = 1$.
For the computation of $\lambda^2(v_6,\{2\})$, the same argument holds since $(v_6,\{2\})$ has the unique successor $(v_7,\{2\})$.
The vertex $(v_1,\{2\})$ has two successors: $(v_6,\{2\})$ and $(v_3,\{2\})$.\footnote{The computation of $\lambda^2(v_1,\{2\})$ was already explained in Example~\ref{ex:update}.}
Again, we know that all $\lambda^1$-consistent plays beginning in $(v_6,\{2\})$ have cost $2$ for player~$1$.
From $(v_3,\{2\})$ however, one can easily check that the play $(v_3,\{2\})(v_0,\{2\}){((v_4,\{2\}))}^\omega$ is $\lambda^1$-consistent and has cost $+\infty$ for player~1.
Thus, we obtain that $\lambda^2(v_1,\{2\}) = 3$.
For the other vertices of $\regionGraphI{J_{2}}$, one can see that $\lambda^2(v)=\lambda^1(v)$.

Finally, we can check that the local fixed point is reached in the arena $\regionGraphGeqI{J_{2}}$ (resp.\ $\regionGraphGeqI{J_{1}}$) with $\lambda^3 = \lambda^2$ (resp.\ $\lambda^6 = \lambda^5 = \lambda^*$). Therefore the respective fixpoints are reached with $k^*_2 = 2$ and $k^*_1 = 5$. The labeling function indicated in Figure~\ref{fig:fuseeExt} is the one of $\lambda^*$.
\qed\end{exa}

\begin{table}
\centering
\scalebox{0.8}{
\begin{tabular}{lllllllllllllll} \toprule
Region                & \multicolumn{1}{c}{$\{1,2\}$} & \multicolumn{7}{c}{$\{2\}$}                                                  & \multicolumn{6}{c}{$\emptyset$}                                      \\\cmidrule(lr){2-2}\cmidrule(lr){3-9}\cmidrule(lr){10-15}
                      & $v$                      & $v_0$ & $v_1$     & $v_6$     & $v_7$     & $v_3$     & $v_4$     & $v_5$     & $v_0$     & $v_1$     & $v_6$     & $v_7$     & $v_3$     & $v_4$     \\ \midrule
$\lambda^0 = \lambda^1$           & 0                              & 0     & $+\infty$ & $+\infty$ & $+\infty$ & $+\infty$ & $+\infty$ & $+\infty$ & $+\infty$ & $+\infty$ & $+\infty$ & $+\infty$ & $+\infty$ & $+\infty$ \\
$\lambda^2=\lambda^3$ & 0                              & 0     & $3$       & $2$       & 1         & $+\infty$ & $+\infty$ & $+\infty$ & $+\infty$ & $+\infty$ & $+\infty$ & $+\infty$ & $+\infty$ & $+\infty$ \\
$\lambda^4$           & 0                              & 0     & 3         & 2         & 1         & $+\infty$ & $+\infty$ & $+\infty$ & $+\infty$ & 3         & 2         & 1         & $+\infty$ & $+\infty$ \\
$\lambda^5=\lambda^*$ & 0                              & 0     & 3         & 2         & 1         & $+\infty$ & $+\infty$ & $+\infty$ & 4         & 3         & 2         & 1         & $+\infty$ & $+\infty$ \\\bottomrule
\end{tabular}
}
\caption{The different steps of the algorithm computing $\lambda^*$ for the extended game of Figure~\ref{fig:fuseeExt}}%
\label{table:fusee}
\end{table}

\section{Counter graph}%
\label{section:counterGraph}

In the previous section, we have introduced the concept of labeling function $\lambda$ and the constraints that it imposes on plays. We have also proposed an algorithm that computes a sequence of labeling functions ${(\lambda^k)}_{k \in \mathbb{N}}$ until reaching a fixpoint $\lambda^*$ such that the plays that are $\lambda^*$-consistent are exactly the SPE outcomes. In this section, given a labeling function $\lambda$, we introduce the concept of \emph{counter graph} such that its infinite paths coincide with the plays that are $\lambda$-consistent. We then show that the counter graph associated with the fixpoint function $\lambda^*$ has an exponential size, an essential step to prove PSPACE membership of the constraint problem.

For the entire section, we fix a reachability game $\mathcal{G} = (G, {(F_i)}_{i\in\Pi}, {(\Cost_i)}_{i\in \Pi})$ with an arena $G = (\Pi, V, {(V_i)}_{i\in \Pi}, E)$, and $v_0$ an initial vertex.
Let $\extGame = (\extG,  {(\extFi{i})}_{i\in\Pi}, {(\extCosti{i})}_{i\in \Pi})$ with the arena $\extG = (\Pi, \extV, {(\extVi{i})}_{i\in\Pi}, \extE)$ be its associated extended game.
Furthermore, when we speak about a play $\rho$ we always mean a play in the extended game $\extGame$.

A labeling function $\lambda$ give constraints on costs of plays from each vertex in $X$, albeit only for the player that owns this vertex. However, by the property of $\lambda$-consistence, constraints for a player carry over all the successive vertices, whether they belong to him or not. In order to check efficiently this property, we introduce the counter graph to keep track \emph{explicitly} of the accumulation of constraints for \emph{all} players at each step of a play. Let us first fix some notation.

\begin{defi}[Restriction and maximal finite range]\label{def:maxrange}
Let $\lambda: \extV \rightarrow \mathbb{N} \cup \{+\infty \}$ be a labeling function.
\begin{itemize}
    \item We consider \emph{restrictions} of $\lambda$ to sub-arenas of $\extV$ as follows. Let $n \in \{1,\ldots,N\}$, we denote by $\lambda_n: V^{J_n} \rightarrow \mathbb{N} \cup \{+\infty \}$ the restriction of $\lambda$ to $V^{J_n}$. Similarly we denote by $\lambda_{\geq n}$ (resp.\ $\lambda_{> n}$) the restriction of $\lambda$ to $V^{\geq J_n}$ (resp.\ $V^{> J_n}$).
    \item The \emph{maximal finite range} of $\lambda$, denoted by $\maxFR$, is equal to \[\maxFR = \max\{c \in \mathbb{N} \mid \lambda(v)=c \text{~for some~} v \in V^X \}\]
with the convention that $\maxFR = 0$ if $\lambda$ is the constant function $+\infty$. We also extend this notion to restrictions of $\lambda$ with the convention that $\mFR{}{> n}=0$ if $J_n$ is a bottom region.
\end{itemize}
\end{defi}

\noindent
Notice that in the definition of maximal finite range, we only consider the \emph{finite} values of $\lambda$ (and not the value $+\infty$).

\begin{defi}[Counter Graph]\label{def:counter_graph}
Let $\lambda: V^X \to \mathbb{N} \cup \lbrace +\infty \rbrace$ be a labeling function. Let $\mathcal{K}:= \lbrace 0,\dots,K \rbrace \cup \lbrace +\infty \rbrace$ with $K = \maxFR{}{}$. The \emph{counter graph} $\C{}$ for $\mathcal{G}$ and $\lambda$ is equal to $\C{} = (\Pi, V^C, {(V^C_{i})}_{i\in\Pi}, E^C)$, such that:
\begin{itemize}
    \item $V^C = V^X \times \mathcal{K}^{|\Pi|}$
    \item $(v,{(c_i)}_{i \in \Pi} ) \in V^C_j$ if and only if $v\in V^X_j$
    \item $((v,{(c_i)}_{i \in \Pi} ),(v',{(c'_i)}_{i \in \Pi})) \in E^C$ if and only if:
    \begin{itemize}
        \item $(v,v')\in \extE$, and
        \item for every $i \in \Pi$
        \[c'_i =
        \begin{cases}
        0 & \text{~if~}  i \in I(v')\\
         c_i - 1 & \text{~if~} i \notin I(v'),  v' \notin V^X_i \text{~and~} c_i > 1 \\
        \min (c_i - 1, \lambda(v')) &\text{~if~} i \notin I(v'), v' \in V^X_i \text{~and~} c_i > 1.
        \end{cases} \]

      \end{itemize}
\end{itemize}
\end{defi}

\noindent
Intuitively, the counter graph is constructed such that once a value $\lambda(v)$ is finite for a vertex $v \in \extV_i$ along a play in $\extGame$, the corresponding path in $\C{}$ keeps track of the induced constraint by \emph{(i)} decrementing the counter value $c_i$ for the concerned player~$i$ by $1$ at every step, \emph{(ii)} updating this counter if a stronger constraint for player~$i$ is encountered by visiting a vertex $v'$ with a smaller value $\lambda(v')$, and \emph{(iii)} setting the counter $c_i$ to $0$ if player~$i$ has reached his target set.

Note that in the counter graph, there may be some vertices with \emph{no outgoing edges}.
Indeed, consider a vertex $(v,{(c_i)}_{i \in \Pi} ) \in V^C$ such that $c_j = 1$ for some player~$j$.
By construction of $\C{}$, the only outgoing edges from $(v,{(c_i)}_{i \in \Pi} )$ must link to vertices $(v',{(c'_i)}_{i \in \Pi})$ such that $(v,v') \in E^X$,  $c'_j =0$ and $j \in I(v')$ (as in Definition~\ref{def:counter_graph} the two other cases require that $c_j > 1$).
However, it may be the case that no successor $v'$ of $v$ in $X$ is such that $j \in I(v')$.

Note as well that for each vertex $v \in V^X$, there exist many different vertices $(v,{(c_i)}_{i \in \Pi} )$ in $\C{}$, one for each counter values profile.
However, the intended goal of the counter graph is to monitor explicitly the constraints accumulated by each player along a play in $\extGame$ regarding the function $\lambda$.
Thus, we will only consider paths in $\C{}$ that start in vertices $(v,{(c_i)}_{i \in \Pi} )$ such that the counter values correspond indeed to the constraint at the \emph{beginning of a play} in $\extGame$ regarding $\lambda$:

\begin{defi}[Starting vertex in $\C{}$]\label{def:associated_vertex_counter_graph}
Let $v \in V^X$.
We distinguish one vertex $v^C = (v,   {(c_i)}_{i \in \Pi} )$ in $V^C$,
such that for every $i\in \Pi$:
\[c_i=
\begin{cases}
0 &\text{ if } i \in I(v) \\
\lambda(v) &\text{ if } i\notin I(v) \text{ and } v \in V^X_i \\
+\infty &\text{ otherwise.}
\end{cases}\]
We call $v^C$ the \emph{starting vertex} associated with $v$, and denote by $\mathrm{SV}(\lambda)$ the set of all starting vertices in $\C{}$.
\end{defi}

\begin{exa}%
\label{ex:compteurs}
Recall the extended game $(\extGame,x_0)$ of Figure~\ref{fig:fuseeExt}, and the labeling function $\lambda$ whose values are indicated under each vertex. In Figure~\ref{fig:counterGraph}, we illustrate a part of the counter graph $\C{}$.
In Example~\ref{ex:lambda}, we have shown that the play \[ \rho = (v_0,\emptyset)(v_1,\emptyset)(v_6,\emptyset)(v_7,\emptyset){((v_2,\Pi)(v_0,\Pi)(v_1,\Pi)(v_6,\Pi)(v_7,\Pi))}^\omega\] was $\lambda$-consistent.
Let us show that there is a corresponding infinite path $\pi$ that starts in ${(v_0,\emptyset)}^C = (v_0, \emptyset, (+\infty, 4))$ in $\C{}$.
Following Definition~\ref{def:counter_graph}, we see that in $\C{}$, there exists an edge between $(v_0, \emptyset, (+\infty, 4))$ and $(v_1, \emptyset, (3, 3))$ and that
\[\pi = {(v_0,\emptyset)}^C(v_1, \emptyset, (3, 3))(v_6, \emptyset, (2, 2))(v_7, \emptyset, (1, 1))\pi'^\omega\]
with
\[\pi' = (v_2, \Pi, (0, 0))(v_0, \Pi, (0, 0))(v_1, \Pi, (0, 0))(v_6, \Pi, (0, 0)) (v_7, \Pi, (0, 0)).\]
Come back now to the play $\rho' = {((v_0, \emptyset)(v_4,\emptyset))}^\omega$ described in Example~\ref{ex:lambda}, which is not $\lambda$-consistent.
From ${(v_0,\emptyset)}^C = (v_0, \emptyset, (+\infty, 4))$, there is an edge to $(v_4,\emptyset,(+\infty,3))$, then to $(v_0, \emptyset, (+\infty, 2))$ and to $(v_4,\emptyset,(+\infty,1))$. For the latter vertex, there is no outgoing edge back to $(v_0,\emptyset,(+\infty,0))$ because $2 \notin I(v_0,\emptyset)$. Therefore there is no infinite path starting in ${(v_0,\emptyset)}^C$ in $\C{}$ that corresponds to $\rho'$.\qed%

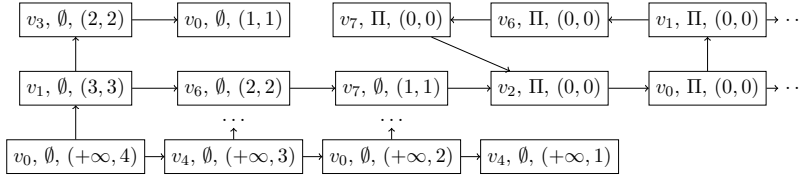
\begin{figure}
   \centering
   \scalebox{0.7}{
    \begin{tikzpicture}
       \node[draw,minimum width=25pt, minimum height=15pt] (v01) at (0,0){$v_0$, $\emptyset$, $(+\infty,4)$};

       \node[draw,minimum width=25pt, minimum height=15pt] (v41) at (3,0){$v_4$, $\emptyset$, $(+\infty,3)$};

       \node[draw, minimum width=25pt, minimum height=15pt] (v02) at (6,0){$v_0$, $\emptyset$, $(+\infty,2)$};

       \node[draw,minimum width=25pt, minimum height=15pt] (v42) at (9,0){$v_4$, $\emptyset$, $(+\infty,1)$};

       \node[draw,minimum width=25pt, minimum height=15pt] (v11) at (0,1.3){$v_1$, $\emptyset$, $(3,3)$};

       \node[draw,minimum width=25pt, minimum height=15pt] (v61) at (3,1.3){$v_6$, $\emptyset$, $(2,2)$};

       \node[draw,minimum width=25pt, minimum height=15pt] (v71) at (6,1.3){$v_7$, $\emptyset$, $(1,1)$};

       \node[draw,minimum width=25pt, minimum height=15pt] (v21) at (9,1.3){$v_2$, $\Pi$, $(0,0)$};

       \node[draw,  minimum height=15pt] (v03) at (12,1.3){$v_0$, $\Pi$, $(0,0)$};

       \node[draw,minimum width=25pt, minimum height=15pt] (v12) at (12,2.6){$v_1$, $\Pi$, $(0,0)$};

       \node[draw,minimum width=25pt, minimum height=15pt] (v62) at (9,2.6){$v_6$, $\Pi$, $(0,0)$};

       \node[draw,minimum width=25pt, minimum height=15pt] (v72) at (6,2.6){$v_7$, $\Pi$, $(0,0)$};

       \node[draw,minimum width=25pt, minimum height=15pt] (v31) at (0,2.6){$v_3$, $\emptyset$, $(2,2)$};

       \node[draw,minimum width=25pt, minimum height=15pt] (v04) at (3,2.6){$v_0$, $\emptyset$, $(1,1)$};

       \draw[->] (v01) -- (v41);

       \draw[->] (v41) -- (v02);

       \draw[->] (v02)--(v42);

       \draw[->] (v01) --(v11);

       \draw[->] (v11) -- (v61);

       \draw[->] (v61)--(v71);

       \draw[->] (v71) -- (v21);

       \draw[->] (v21) -- (v03);

       \draw[->] (v03) -- (v12);

       \draw[->] (v12) -- (v62);

       \draw[->] (v62) -- (v72);

       \draw[->] (v72) -- (v21);

       \draw[->] (v11) -- (v31);

       \draw[->] (v31)--(v04);

       \node (dot1) at (3,0.7){$\ldots$};

       \draw[->] (v41.north) -- (dot1);

       \node (dot2) at (6,0.7){$\ldots$};

       \draw[->] (v02) -- (dot2);

       \node (dot3) at (13.7,1.3){$\ldots$};

       \draw[->] (v03) -- (dot3);

       \node (dot4) at (13.7,2.6){$\ldots$};

       \draw[->] (v12) -- (dot4);

       \end{tikzpicture}
    }
    \caption{Part of the counter graph $\C{}$ associated with the game of Figure~\ref{fig:fusee}}%
    \label{fig:counterGraph}
\end{figure}
\end{exa}

There exists a correspondence between $\lambda$-consistent plays in $\extGame$ and \emph{infinite} paths from starting vertices in $\C{}$ in the following way. On one hand, every play $\rho$ in $\extGame$ that is not $\lambda$-consistent does not appear in the counter graph: the first constraint regarding $\lambda$ that is violated along $\rho$ is reflected by a vertex in $\C{}$ with a counter value getting to $1$ and no outgoing edges. On the other hand, $\lambda$-consistent plays in $\extGame$ have a corresponding infinite path in the counter graph $\C{}$. We call \emph{valid paths} the infinite paths of $\C{}$.  This correspondence is formalized in the two following lemmas:

\begin{lem}\label{lem:lambda_consistent_rho_X_to_rho_C}
Let $\rho = \rho_0\rho_1 \ldots$ be a $\lambda$-consistent play in $\Plays_X(v)$. Then there exists an associated infinite path $\pi = \pi_0\pi_1 \ldots$ in $\C{}$ such that:
\begin{itemize}
    \item $\pi_0 = v^C$,
    \item $\rho$ is the projection of $\pi$ on $V^X$, that is, $\pi_n$ is of the form $(v',{(c'_i)}_{i \in \Pi} )$ with $v'=\rho_n$, for every $\N{n}$.
\end{itemize}
\end{lem}

\begin{proof}
Let $v \in V^X$ and $\rho$ be a play in $\extGame$ such that $\rho \in \Lambda(v)$ and $\rho_0= v$.
We build the corresponding infinite path $\pi$ in $\C{}$ iteratively as follows:
Let $\pi_0:= v^C$.
Let $\N{n}$, $n\geq 1$.
Suppose that $\pi_{< n}$ has been already constructed, we show how to choose $\pi_n$.
Suppose $\pi_{n-1} = (v',{(c'_i)}_{i \in \Pi} )$.
Then $\pi_n:= (v'', {(c''_i)}_{i \in \Pi} )$, where:
\begin{itemize}
    \item $v'' = \rho_n$,
    \item for every $i \in \Pi$:
    \[ c''_i = \begin{cases}
    0 & \text{~if~}  i \in I(v''),\\
         c'_i - 1 & \text{~if~} i \notin I(v'') \text{~and~}  v'' \notin V^X_i, \\
        \min (c'_i - 1, \lambda(v'')) &\text{~if~} i \notin I(v'') \text{~and~} v'' \in V^X_i
    \end{cases}\]
\end{itemize}

\noindent
Let us show that $((v',{(c'_i)}_{i \in \Pi} ),(v'', {(c''_i)}_{i \in \Pi} )) \in E^C$. As $\rho$ is a play in $\extGame$, we clearly have $(v',v'') \in E^X$. Assume now, towards contradiction, that $((v',{(c'_i)}_{i \in \Pi} ),(v'', {(c''_i)}_{i \in \Pi} )) \notin E^C$, that is, there is no outgoing edge from vertex $(v',{(c'_i)}_{i \in \Pi} )$ in the counter graph. By Definition~\ref{def:counter_graph}, this means that there exists a player~$i$ such that $c'_i = 1$ and $i \notin I(v'')$.

As $c'_i = 1$, by Definition~\ref{def:counter_graph} and construction of $\pi_0 \ldots \pi_{n-1}$, there must exist a largest index $m \leq n-1$ such that $\rho_m \in V^X_i$, $\lambda(\rho_m) = d$ with $d > 0$ and such that the counter value for player~$i$ decreases by exactly $1$ at each step from vertex $\pi_m$ until reaching value $c'_{i} = 1$ at vertex $\pi_{n-1} = v'$. Since $\rho$ is $\lambda$-consistent, player~$i$ visits his target set along $\rho$ in at most $d$ steps, thus at most at vertex $v''$, and thus $i \in I(v'')$, which is a contradiction.
\end{proof}

\begin{lem}\label{lem:chemin_infini_compteurs}
Let $v^C = (v,{(c_i)}_{i \in \Pi} )$ be a starting vertex in $\mathrm{SV}(\lambda)$.
Let $\pi = \pi_0\pi_1\ldots$ be an infinite path in $\C{}$ such that $\pi_0 = v^C $. Then there exists a corresponding play $\rho = \rho_0\rho_1 \ldots$ in $\extGame$ such that:
\begin{itemize}
    \item $\rho$ is $\lambda$-consistent,
    \item $\rho$ is the projection of $\pi$ on $V^X$, that is, $\rho_n = v'$ with $\pi_n=(v',{(c'_i)}_{i \in \Pi} )$, for every $\N{n}$.
\end{itemize}

\end{lem}
\begin{proof}
Let $v^C = (v,{(c_i)}_{i \in \Pi} )$ be a starting vertex in $\mathrm{SV}(\lambda)$.
Let $\pi$ be an infinite path in $\C{}$ such that $\pi_0 = v^C $.
Let $\rho$ be the projection of $\pi$ on $V^X$, that is, $\rho_n = v'$ with $\pi_n=(v',{(c'_i)}_{i \in \Pi} )$, for every $\N{n}$.
\begin{itemize}
    \item Clearly, $\rho$ is a play in $\extGame$: by construction of $\C{}$, there exists an edge between two vertices $(v,{(c_i)}_{i \in \Pi} )$ and $(v',{(c'_i)}_{i \in \Pi} )$ in $\C{}$ only if $(v,v')$ is an edge in $E^X$.
    \item Furthermore, the play $\rho$ is $\lambda$-consistent:
    Assume, towards contradiction, that it is not.
    Thus, there exists $\N{n}$ and $i \in \Pi$ such that $\rho_n \in V^X_i$ and
    \begin{align}
    \Cost_i(\rho_{\geq n}) > \lambda(\rho_n). \label{eq:cost_more_than_lambda}
    \end{align}
    Consider now $\pi_n$ and the value $c_i$ of the counter for player~$i$ at this vertex.
    Since $\rho_n \in V^X_i$, we know that $c_i \leq \lambda(\rho_n)$.
    From this vertex, the counter value for player~$i$ decreases at least by $1$ at each step along $\pi_{\geq n}$.
    and hits the value $0$ before $\lambda(\rho_n)$ steps.
    However, this means that in $\rho_{\geq n}$, player~$i$ visits his target set sooner than expected, which is a contradiction with~(\ref{eq:cost_more_than_lambda}). \qedhere
\end{itemize}
\end{proof}

\noindent
Since the edge relation $E^C$ in the counter graph respects the edge relation $E^X$ in the extended game, the region decomposition of path in $\extGame$ given in Lemma~\ref{lem:region_decomposition} can also be applied to a path in $\C{}$. We will often use such path region decompositions in the proofs of this section.

In order to prove the PSPACE membership for the constraint problem, we need to show that the counter graph $\C{*}$, with $\lambda^*$ the fixpoint function computed by Algorithm~\ref{algo:lambda}, has an exponential size. To this end, as the size of $|\C{*}|$ of $\C{*}$ is equal to $|V|\cdot 2^{|\Pi|} \cdot {(K + 2)}^{|\Pi|}$ with $K = \mFR{*}{}$ defined in Definition~\ref{def:maxrange}, it is enough to show an exponential upper bound on $K$. We proceed in two steps:
First, with the next proposition, given a labeling function $\lambda$ and its restriction $\lambda_{\geq \ell}$ to $V^{\geq \ell}$, we exhibit a bound on the supremum of the cost of $\lambda$-consistent plays for each player, in terms of the maximal finite range $\mFR{}{\geq \ell}$.
Second, we consider the actual sequence of functions ${(\lambda^k)}_{k \in \mathbb{N}}$ defined in Definitions~\ref{def:init} and~\ref{def:update}, as implemented by Algorithm~\ref{algo:lambda}. With Theorem~\ref{thm:bound_on_MR}, we show that $\mFR{*}{}$ is bounded by an exponential in the size of the original game $\mathcal{G}$. The proof is by induction on $k$ and uses Proposition~\ref{prop:sup_cost_bound}.

\begin{prop}[Bound on finite supremum]{\label{prop:sup_cost_bound}}
Let $\lambda$ be a labeling function. Let $v\in V^X$ such that $I(v) = J_\ell$ with $\ell \in \{1,\ldots,N\}$. Let $c \in \mathbb{N} \cup \{+\infty\}$ be such that $\sup~ \{ \Cost_i(\rho) \mid \rho \in \Lambda^{}(v)\} = c$. If $c < + \infty$, then
\[c \leq |V| + 2 \cdot \mFR{}{\ell} + \sum^{|\Pi|}_{r=|J_\ell|+1} |V| + 2 \cdot \max_{\stackrel{J_j > J_\ell}{ |J_j|=r}} \mFR{}{j}.\]
Moreover, in both cases $c = +\infty$ and $c < +\infty$, there exists a valid path $\pi$ in $\C{}$ starting in $v^C$ that is a lasso $hg^\omega$ with the length of $hg$ bounded by $2 \cdot |\C{}|$ and such that its corresponding play $\rho$ in $\extGame$ belongs to $\Lambda^{}(v)$ and has its cost $\Cost_i(\rho)$ equal to $c$.
\end{prop}

Before proving Proposition~\ref{prop:sup_cost_bound}, we need the following technical lemma:

\begin{lem}\label{lem:length_of_rho_1_in_C}
Let $v^C$ be a starting vertex in $\mathrm{SV}(\lambda$) associated with $v \in V^X$ such that $I(v) = J_\ell$.
Let $\pi$ be a finite prefix of a valid path in $\C{}$ such that:
\begin{itemize}
    \item $\pi_0 = v^C$,
    \item $\pi$ does not contain any cycle.
\end{itemize}
Then,
\[|\pi| \leq |V| + 2 \cdot \mFR{}{\ell} + \sum^{|\Pi|}_{r=|J_\ell|+1} |V| + 2 \cdot \max_{\stackrel{J_j > J_\ell}{ |J_j|=r}} \mFR{}{j}\]
\end{lem}

\begin{proof}[Proof sketch]\label{proof_sk:length} The proof of this lemma is quite technical, so we give here the main ideas and refer to the full proof hereafter for more details.

Let $\pi$ be a finite prefix of a valid path in $\C{}$ as in the statement.
Let $\pi[\ell]\dots\pi[m]$ be its region decomposition according to Lemma~\ref{lem:region_decomposition}, graphically represented in Figure~\ref{fig:decomposition}. Let $\rho$ be the corresponding path in $\extGame$ and $\rho[\ell]\dots\rho[m]$ be its region decomposition. Let us consider a fixed non-empty section $\pi[n]$.
\begin{figure}[h!]
    \begin{tikzpicture}

\draw[very thick] (0,1) -- (2,1) ;
\node[draw,circle, thick, fill=black] (pi_l_0) at (0,1) {};
\draw (0,1.5) node{$\pi[\ell]_0$};
\draw[thick, decorate,decoration={brace, mirror}] (0,0) -- (2,0) node[midway, yshift=-1em] {$J_\ell$} ;



\draw[ultra thick, dotted] (2.5,1) -- (4.5,1) ;
\draw[thick, dotted] (2.5,0) -- (4.5,0) ;


\draw[very thick] (5,1) -- (7,1) ;
\draw[thick, double] (5.75,0.5) -- (5.75,1.5) ;
\draw (5.75, 0.3) node{\tiny{$\infty \to c$}} ;
\node[draw,circle, thick, fill=black] (pi_l1_0) at (5,1) {};
\draw (5, 1.5) node{${\pi[n]}_0$} ;
\draw[thick, decorate,decoration={brace, mirror}] (5,0) -- (7,0) node[midway, yshift=-1em] {$J_{n}$} ;


\draw[ultra thick, dotted] (7.5,1) -- (9.5,1) ;
\draw[thick, dotted] (7.5,0) -- (9.5,0) ;

\draw[very thick] (10,1) -- (12,1) ;

\node[draw,circle, thick, fill=black] (pi_m_0) at (10,1) {};
\draw (10, 1.5) node{$\pi[m]_0$} ;
\draw[thick, decorate,decoration={brace, mirror}] (10,0) -- (12,0) node[midway, yshift=-1em] {$J_{m}$} ;

\end{tikzpicture}
    \caption{Region decomposition of $\pi$}%
    \label{fig:decomposition}
\end{figure}
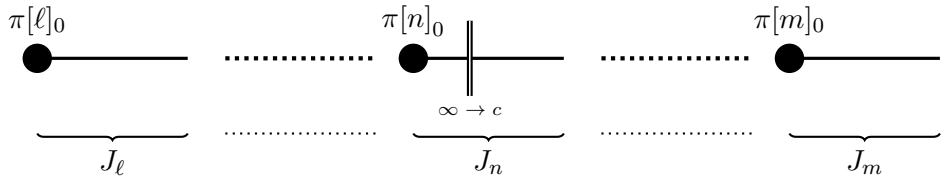

Suppose first that the counter values at ${\pi[n]}_0$ are either $0$ or $+\infty$. Let us prove that along $\pi[n]$, there can be at most $|V|$ steps before reaching a vertex with a finite positive value of $\lambda$:
\begin{itemize}
    \item assume there is a cycle in the corresponding section $\rho[n]$ in $\extGame$ such that from ${\rho[n]}_0$ and along the cycle, all the values of $\lambda$ are either $0$ or $+\infty$,
    \item by construction of $\C{}$, the counter values in the corresponding prefix of $\pi[n]$ remain fixed for each vertex of this prefix: as no value of $\lambda$ is positive and finite, no counter value can be decremented,
    \item thus, the cycle in $\rho[n]$ is also a cycle in $\pi[n]$ which is impossible by hypothesis,
    \item thus there is no such cycle in $\rho[n]$, and as there are at most $|V|$ vertices in region $X^{J_n}$, $\rho[n]$ can have a prefix of length at most $|V|$ with only values $0$ or $+\infty$ for $\lambda$, implying that this is also the case for $\pi[n]$.
\end{itemize}
Therefore, we can decompose $\pi[n]$ into a (possibly empty) prefix of length at most $|V|$, and a (possibly empty) suffix where at least one counter value $c'_i$, for some $i$, is a positive finite value in its first vertex $v'$.
This frontier between prefix and suffix of $\pi[n]$ is represented by a vertical double bar $\|$ with caption $\infty \to c$ in Figure~\ref{fig:decomposition}.
This value $c'_i$ is bounded by $\mFR{}{n}$, the maximal finite range of $\lambda_n$.
From there, as the corresponding $\rho$ is $\lambda$-consistent, player~$i$ reaches his target set in at most $c'_i$ steps, and $\rho$ enters a new region, which means that the section $\pi[n]$ is over.
So, in that case, the length of $\pi[n]$ can be bounded by $|V| + \mFR{}{n}$.

Suppose now that at vertex ${\pi[n]}_0$, there exists a counter value $c_i$ for some player~$i$ that is neither $0$ nor $+\infty$.
This means that there was a constraint for player~$i$ initialized in a previous section $\pi[n']$, with $n' < n$, that has carried over to ${\pi[n]}_0$, via decrements of at least $1$ per step.
We know that the initial finite counter value is bounded by $\mFR{}{n'}$, and appeared before the end of section $\pi[n']$. Thus the length from the end of section $\pi[n']$ to the end of section $\pi[n]$ is bounded by $\mFR{}{n'}$, as again, once the counter value attains $0$ for player~$i$, the path $\pi$ has entered the next section.

Therefore, considering the possible cases for each section, we can bound the total length of $\pi$ as follows:
\[|\pi| \leq \sum^{m}_{j=\ell} |V| + 2 \cdot \mFR{}{j}.\]

Finally, remark that by $I$-monotonicity, it is actually the case that only (and at most) $|\Pi|$ different non-empty sections can appear in the decomposition of $\pi$.
Furthermore, for each $n \in \{\ell +1, \ldots, N\}$, we have \[\mFR{}{n} \leq \max_{\stackrel{J_j > J_\ell}{ |J_j|=|J_n|}} \mFR{}{j}\] by Definition~\ref{def:maxrange}.
Thus, we have the following bound:
\[ |\pi| \leq |V| + 2\cdot \mFR{}{\ell} + \sum^{|\Pi|}_{r=|J_\ell|+1} |V| + 2 \cdot \max_{\stackrel{J_j > J_\ell}{ |J_j|=r}} \mFR{}{j}\]
which is the bound stated in Lemma~\ref{lem:length_of_rho_1_in_C}.
\end{proof}

\begin{proof}[Proof of Lemma~\ref{lem:length_of_rho_1_in_C}]
Let $v^C$ be a starting vertex in $\mathrm{SV}(\lambda$).
Let $\pi$ be a finite prefix of a valid path in $\C{}$ such that:
\begin{itemize}
    \item $\pi_0 = v^C$,
    \item $\pi$ does not contain any cycle.
\end{itemize}

\noindent
The following proof is quite technical, thus we alleviate some of the difficulties by proving a first upper bound on the length of finite paths without cycles in $\C{}$, then by showing how to obtain the desired bound.
The main idea is to bound iteratively the length of prefixes of $\pi$, adding at each step of the reasoning the section for the next region traversed by $\pi$.

Let $\pi'$ be a valid path in $\C{}$ such that $\pi$ is a prefix of $\pi'$.
By Lemma~\ref{lem:region_decomposition}, we know that there exist two natural numbers $\ell, m' \leq N$ and $m'-\ell$ (possibly empty) paths $\pi'^\ell, \dots, \pi'^{m'}$ in $\C{}$ such that:
\begin{itemize}
    \item $\pi'=\pi'[\ell] \dots \pi'[m']$,
    \item for each $ \ell \leq j \leq m'$, each vertex in $\pi'[n]$ is of the form $(w,{(c_i)}_{i \in \Pi})$ with $I(w)=J_j$.
\end{itemize}

\noindent
The finite path $\pi$ is a finite prefix of $\pi'$.
Thus there exists a natural number $m \leq m'$ and $m-\ell$ (possibly empty) paths $\pi[\ell], \dots, \pi[{m'}]$ in $\C{}$ such that:
\begin{itemize}
    \item $\pi=\pi[\ell] \dots \pi[{m}]$,
    \item for each $ \ell \leq j < m$, $\pi[j]=\pi'[j]$,
    \item $\pi[m] $ if a finite prefix of $\pi'[{m}]$.
\end{itemize}

\noindent
By Lemma~\ref{lem:chemin_infini_compteurs}, there exist a corresponding $\lambda$-consistent play $\rho'$ in $\extGame$ and a corresponding history $\rho$. Furthermore, since $\pi$ contains no cycle, the sections $\pi[n]$ do not either.

\bigskip
We first treat the case where $m < m'$.

We first bound the length of the first section $\pi[\ell]$.
Recall that since ${\pi[\ell]}_0=v^C \in \mathrm{SV}(\lambda)$, we have that $v^C= (v,{(c_i)}_{i \in \Pi} )$ with $c_i = 0$ if $ i\in J_\ell$, $c_i = \lambda(v)$ if $v \in V_i$ and $c_i = +\infty$ otherwise.

Suppose $0 < \lambda(v) < +\infty$.
In that case, we know that along $\rho'$, which is $\lambda$-consistent, player~$i$ reaches his target after at most $\lambda(v)$ steps.
Thus, there exists $n \leq m'$, $n\neq \ell$ such that $I(\rho_{\lambda(v)}) = J_n$, $J_\ell \subsetneq J_n$, and $\pi_{\lambda(v)}$ belongs to section $\pi[n]$.
This means that section $\pi[\ell]$ is shorter than $\lambda(v)$.
Since $I(v)=J_\ell$, we have that $\lambda(v) \leq \mFR{}{\ell}$, and thus $|\pi[\ell]| \leq \mFR{}{\ell}$.

\medskip

Suppose now that $\lambda(v)=0$ or $\lambda(v)=+\infty$.
The counter values at $v^C$ are thus either $0$ or $+\infty$.
Along $\pi'$, they can stay stable for at most $|V|$ steps (see Proof sketch of Lemma~\ref{lem:length_of_rho_1_in_C}).
Otherwise, the cycle induced in $\rho'$ is also a cycle in $\pi'$ as the counter values are fixed.
If $\pi'_{|V| +1}$ is in section $\pi'[n]$ with $n > \ell$, then we immediately get that $|\pi[\ell]| \leq |V|$.
If $\pi'_{|V| +1}$ is in section $\pi'[\ell]$, then it means that a counter value for some player~$i$ has become finite along the first $|V|+1$ vertices of $\pi'[\ell]$.
Let $t$ be the first index it does so along $\pi'[\ell]$.
By the same argument as in the previous case, we know that from ${\pi'[\ell]}_t$, there is at most $\mFR{}{\ell}$ vertices before entering the next section of $\pi'$.
Thus, we obtain that:
\begin{align*}
|\pi[\ell]|
& \leq |\pi'[\ell]| \\
&\leq |V| + \mFR{}{\ell} \\
& \leq |V| + 2\cdot\mFR{}{\ell}.
\end{align*}

If $\ell=m$, we can already conclude that:
\begin{align*}
|\pi|
& \leq |\pi[\ell]| \\
& \leq |V| + 2\cdot\mFR{}{\ell} \\
& \leq \sum^{m}_{j=\ell} |V| + 2 \cdot \mFR{}{j}.
\end{align*}

\bigskip

Suppose now that $\ell < m$.
We show that for each $n > \ell$, we have:
\[|\pi[\ell]\dots\pi[n]| \leq |V| + 2\cdot\mFR{}{\ell} + \left( \sum^{n-1}_{j=\ell+1} |V| + 2\cdot\mFR{}{j}  \right) + |V| + \mFR{}{n}.\]

\bigskip

Let $n=\ell+1$.
We assume $\pi[n]$ is not empty, otherwise its length is $0$.
Consider now ${\pi[n]}_0$.

\medskip

Suppose there exists a player~$i$ such that his counter value $c_i$ at ${\pi[n]}_0 = (w,J_n,{(c_i)}_{i \in \Pi} )$ is a finite non-zero value.
If $w \not \in V_i$, it means that the counter value has decreased since a vertex $(w',{(c'_i)}_{i \in \Pi} )$ in the previous section $\pi[{\ell}]$ such that $w' \in V_i$ and $c'_i = \lambda(w )$.
Thus, the length of the path from this vertex $(w',{(c'_i)}_{i \in \Pi} )$ to the next section $\pi[{n+1}]$ is smaller or equal than $\lambda(w' )$.
In particular, the whole section $\pi[n]$ is ``covered'' by this path.
Therefore, we can conclude that $|\pi[n]| \leq \lambda(w' ) \leq \mFR{}{\ell}$.
Since we already know that $|\pi[\ell]| \leq |V| + \mFR{}{\ell}$, we obtain that:

\begin{align*}
|\pi[\ell] \pi[n]|& \leq |V| + \mFR{}{\ell} + \mFR{}{\ell} \\
& \leq  |V| + 2\cdot\mFR{}{\ell}.
\end{align*}

If $w \in V_i$, it means that either the counter value $c_i$ is equal to $\lambda(w )$ or has decreased since a vertex $(w',{(c'_i)}_{i \in \Pi} )$ in the previous section $\pi[{\ell}]$ such that $w' \in V_i$ and $c'_i = \lambda(w' )$.
Thus, we have that $c_i \leq \mFR{}{n} $ or $c_i \leq \mFR{}{\ell}$, and, in turn, $|\pi[n]| \leq \mFR{}{n} $ or $|\pi[n]| \leq \mFR{}{n}$.
Therefore, we can conclude that:
\begin{align*}
 |\pi[\ell] \pi[n]|    & \leq  |V| + \mFR{}{\ell} + \mFR{}{\ell} + \mFR{}{n} \\
     & \leq  |V| + 2 \cdot \mFR{}{\ell} + \mFR{}{n}.
\end{align*}

\medskip

Suppose now that for every player~$i$, his counter value $c_i$ at ${\pi[n]}_0 = (w, {(c_i)}_{i \in \Pi} )$ is either $0$ or $+\infty$.
In that case, we are in a similar case than for section $\pi[\ell]$, thus we can conclude that $|\pi[n]| \leq |v| + \mFR{}{n}$.
Thus, we have indeed:
\[|\pi[\ell] \pi[n]| \leq  |V| + \mFR{}{\ell} + |V| + \mFR{}{n},\]
and also
\[|\pi[\ell] \pi[n]| \leq  |V| + 2\cdot\mFR{}{\ell} + |V| + \mFR{}{n}.\]

If $m=\ell+1$, we are done.
\bigskip

Suppose now that $m > \ell+1$.
Let $n$ be such that $\ell+1 <n \leq m$.
Assume that for each $n'<n$, it holds that
\[|\pi[\ell]\dots\pi[{n'}]| \leq |V| + 2\cdot\mFR{}{\ell}
+ \left( \sum^{n'-1}_{j=\ell+1} |V| + 2\cdot\mFR{}{j} \right)
+ |V| + \mFR{}{n'}.\]
We assume $\pi[n]$ is not empty, otherwise its length is $0$.
Consider now ${\pi[n]}_0$.

\medskip

Suppose there exists a player~$i$ such that his counter value $c_i$ at ${\pi[n]}_0 = (w, {(c_i)}_{i \in \Pi} )$ is a finite non-zero value.
If $w \not \in V_i$, it means that the counter value has decreased since a vertex $(w', {(c'_i)}_{i \in \Pi} )$ in a previous section $\pi[{n'}]$ such that $w' \in V_i$ and $c'_i = \lambda(w' )$.
Thus, the length of the path from this vertex $(w', {(c'_i)}_{i \in \Pi} )$ to the next section $\pi[{n+1}]$ is smaller or equal than $\lambda(w' )$.
In particular, the whole sections from $\pi[{n'+1}]$ to $\pi[n]$ are ``covered'' by this path.
Therefore, we can conclude that $|\pi[{n'+1}]\dots\pi[n]| \leq \lambda(w' ) \leq \mFR{}{n'}$.
Since we already know that:
\[|\pi[\ell] \dots \pi[{n'}]| \leq |V| + 2\cdot\mFR{}{\ell}
+ \left( \sum^{n'-1}_{j=\ell+1} |V| + 2\cdot\mFR{}{j} \right)
+ |V| + \mFR{}{n'},\]
we obtain that:

\begin{align*}
|\pi[\ell] \dots \pi[n]|
&\leq |V| + 2\cdot\mFR{}{\ell}
+ \left( \sum^{n'-1}_{j=\ell+1} |V| + 2\cdot\mFR{}{j} \right)
+ |V| + \mFR{}{n'} +\mFR{}{n'} \\
& \leq |V| + 2\cdot\mFR{}{\ell}
+ \left( \sum^{n'-1}_{j=\ell+1} |V| + 2\cdot\mFR{}{j} \right)
+ |V| + 2 \cdot \mFR{}{n'}  \\
& \leq |V| + 2\cdot\mFR{}{\ell}
+ \left( \sum^{n-1}_{j=\ell+1} |V| + 2\cdot\mFR{}{j} \right)
+ |V| + \mFR{}{n}.
\end{align*}

If $w \in V_i$, it means that either the counter value $c_i$ is equal to $\lambda(w )$ or has decreased since a vertex $(w', {(c'_i)}_{i \in \Pi} )$ in a previous section $\pi[{n'}]$ such that $w' \in V_i$ and $c'_i = \lambda(w' )$.
Thus, we have that $c_i \leq \mFR{}{n} $ or $c_i \leq \mFR{}{n'}$, and, in turn, $|\pi[n]| \leq \mFR{}{n} $ or $|\pi[n]| \leq \mFR{}{n'}$.

If $|\pi[n]| \leq \mFR{}{n} $, we can conclude that:
\begin{align*}
|\pi[\ell] \dots \pi[n]| &\leq
|V| + 2\cdot\mFR{}{\ell}
+ \left( \sum^{n-1}_{j=\ell+1} |V| + 2\cdot\mFR{}{j} \right)
+ \mFR{}{n}  \\
&\leq |V| + 2\cdot\mFR{}{\ell}
+ \left( \sum^{n-1}_{j=\ell+1} |V| + 2\cdot\mFR{}{j} \right)
+ |V| + \mFR{}{n}.
\end{align*}

If $|\pi[n]| \leq \mFR{}{n'}$, we can conclude that:
\begin{align*}
|\pi[\ell] \dots \pi[n]| &\leq
|V| + 2\cdot\mFR{}{\ell}
+ \left( \sum^{n'-1}_{j=\ell+1} |V| + 2\cdot\mFR{}{j} \right)
+ |V| + \mFR{}{n'} +\mFR{}{n'} \\
&\leq |V| + 2\cdot\mFR{}{\ell}
+ \left( \sum^{n-1}_{j=\ell+1} |V| + 2\cdot\mFR{}{j} \right)
+ |V| + \mFR{}{n}.
\end{align*}

\medskip

Suppose now that for every player~$i$, his counter value $c_i$ at ${\pi[n]}_0 = (w, {(c_i)}_{i \in \Pi} )$ is either $0$ or $+\infty$.
In that case, we are in a similar case than for section $\pi[\ell]$, thus we can conclude that $|\pi[n]| \leq |v| + \mFR{}{n}$.
Thus, we have indeed:
\[|\pi[\ell] \dots \pi[n]|
\leq |V| + 2\cdot\mFR{}{\ell}
+ \left( \sum^{n-1}_{j=\ell+1} |V| + 2\cdot\mFR{}{j} \right)
+ |V| + \mFR{}{n}\]
and thus finally:
\[
|\pi|
\leq |\pi[\ell] \dots \pi[m]|
\leq \sum^{m}_{j=\ell} |V| + 2 \cdot \mFR{}{j}.
\]

\bigskip

Assume now that $m=m'$.
In that case, we cannot rely on the section $\pi'^m$ to be finite, as the play $\rho'$ never reaches another region.
However, in this situation, it is guaranteed that the counter values in $\pi'^m$ are fixed and are equal to either $0$ or $+\infty$:
indeed, a finite counter value would imply that some player~$i$ such that $i \not \in J_m$ reaches his target in $\rho'$ exactly when his counter value becomes $0$ in $\pi'$.
But $\pi'^m$ is the last section of $\rho'$, thus no new region is reached after $J_m$ and no new player can visit his target set than the players $i \in J_m$.
Therefore, finite prefix of $\pi'^m$ that contains no cycle has its length bounded by $|V|$.
Thus, we can conclude:
\begin{align*}
    |\pi|
    & =  |\pi[\ell]\dots\pi[m]| \\
    & \leq |V| + 2\cdot\mFR{}{\ell}
+ \left( \sum^{m-1}_{j=\ell+1} |V| + 2\cdot\mFR{}{j} \right)
+ |V|  \\
& \leq \sum^{m}_{j=\ell} |V| + 2 \cdot \mFR{}{j}.
\end{align*}

\bigskip

In fact, the bound given above can be slightly changed to give the desired bound.
Indeed, the bound above relies on the fact that $m \leq N$ and covers the case where a path traverses every region from $J_\ell$ to $J_N$.
In all generality, the number $N$ of different regions can be exponential in the number $|\Pi|$ of players.
However, by the $I$-monotonicity property, we know that a path can actually traverse at most $|\Pi|$ regions.
Thus, in the region decomposition of a path, only (and at most) $|\Pi|$ sections are relevant and of length greater than $0$.
Therefore, we can define a subsequence of indices ${(n_r)}_{r \leq |\pi|}$, with $n_1=\ell$, such that in fact $\pi= \pi[{n_1}]\dots\pi[{n_{|\Pi|}}]$.
Hence, we obtain the following bound on the length $t$ of $\pi$:

\[ t \leq \sum^{|\Pi|}_{r=1} |V| + 2 \cdot \mFR{}{n_r}\]

Finally, as for every $r \leq |\pi|$, $r>1$, we have $\mFR{}{n_r} \leq \max \{\mFR{}{j} \mid J_j > J_\ell,  |J_j|=|J_{n_r}|\}$, we obtain the desired bound:

\[ t \leq |V| + 2 \cdot \mFR{}{\ell} + \sum^{|\Pi|}_{r=|J_\ell|+1} |V| + 2 \cdot \max_{\stackrel{J_j > J_\ell}{ |J_j|=r}} \mFR{}{j}\]
\end{proof}

We are now ready to prove Proposition~\ref{prop:sup_cost_bound}.
\begin{proof}[Proof of Proposition~\ref{prop:sup_cost_bound}]\label{proof:sup_cost_bound} Let $v \in V^X$ with $I(v)=J_\ell$ and $i \in \Pi$. Let $c\in \mathbb{N} \cup \{+\infty\}$ be such that $\sup~ \{ \Cost_i(\rho) \mid \rho \in \Lambda^{}(v)\} = c$. Consider $\rho \in \Lambda^{k}(v )$ such that $\Cost_i(\rho) = c$. Notice that such a play always exists by Corollary~\ref{cor:sup_max}. Consider also $\pi$ the valid path in $\C{}$ that starts in $v^C$ and corresponds to $\rho$.

\medskip
Suppose first that $c < +\infty$ and let us prove that
\begin{eqnarray}
c \leq |V| + 2 \cdot \mFR{}{\ell} + \sum^{|\Pi|}_{r=|J_\ell|+1} |V| + 2 \cdot \max_{\stackrel{J_j > J_\ell}{ |J_j|=r}} \mFR{}{j}. \label{eq:prop32}
\end{eqnarray}

If $i \in J_\ell$, then $c = 0$, as every play starting in $v$ has a cost $0$ for player~$i$. Hence (\ref{eq:prop32}) trivially holds.

Suppose now that $i \notin J_\ell$. As player~$i$ eventually reaches his target set along $\rho$, this means that $\rho$ eventually leaves region $J_\ell$ and eventually reaches another region $J_n$ such that $i \in J_n$.
Consider the prefix $\rho_{\leq c}$ of $\rho$ of length $c$.
Let $\pi_{\leq c}:= \pi_{\leq c}[\ell]\dots\pi_{\leq c}[m]$ be the prefix of $\pi$ and its region decomposition such that $\pi_{\leq c}$ is associated with $\rho_{\leq c}$.
Notice that $\pi_{\leq c}[m]$ consists only of one vertex corresponding to $\rho_c$, and that for every $n < m$, we have $i \notin J_n$.

Suppose that $\pi_{\leq c}$ contains a cycle.
By construction of $\C{}$, this cycle is included in one single section $\pi_{\leq c}[n]$, where $n<m$ (as $\pi_{\leq c}[m]$ contains only one vertex), and thus $i\notin J_n$.
Consider the infinite path $\pi'$ in $\C{}$ that follows $\pi_{\leq c}$ until the cycle and then repeats the cycle forever.
By Lemma~\ref{lem:chemin_infini_compteurs}, there exists a $\lambda$-consistent play $\rho'$ in $\extGame$ corresponding to $\pi'$.
We have $\Cost_i(\rho') = +\infty$ for player~$i$, as $\rho'$ never reaches a region where player~$i$ visits his target set.
This is a contradiction with the fact that $\sup~ \{ \Cost_i(\rho) \mid \rho \in \Lambda^{}(v )\} = c$ is finite.

Therefore $\pi_{\leq c}$ contains no cycle, and by Lemma~\ref{lem:length_of_rho_1_in_C},
\[|\pi_{\leq c}| \leq |V| + 2 \cdot \mFR{}{\ell} + \sum^{|\Pi|}_{r=|J_\ell|+1} |V| + 2 \cdot \max_{\stackrel{J_j > J_\ell}{ |J_j|=r}} \mFR{}{j}.\]
Since $|\pi_{\leq c}| = c = \Cost_i(\rho)$, we obtain Inequality (\ref{eq:prop32}).

\medskip Let us now prove the second part of Proposition~\ref{prop:sup_cost_bound} for both cases $c < +\infty$ and $c = +\infty$. Given the valid path $\pi$, we consider the first occurrence of a cycle in $\pi$. We then construct the infinite path $\pi'$ in $\C{}$ that follows $\pi$ until this cycle and then repeats it forever. Then $\pi'$ is a lasso $h g^\omega$ with the length $|h g|$ bounded by $2 \cdot |\C{}|$. Clearly if $c = +\infty$, then the corresponding play $\rho'$ in $\extGame$ belongs to $\Lambda^{}(v)$ and has a cost $\Cost_i(\rho')$ equal to $\Cost_i(\rho) = + \infty$. If $c < +\infty$, we know by the first part of the proof that $\pi_{\leq c}$ contains no cycle and thus is prefix of $h$. Therefore we also have that the corresponding play $\rho'$ has cost $\Cost_i(\rho') = \Cost_i(\rho) = c$.
\end{proof}

\begin{rem}%
\label{rem:countergraph}
In the proof of Proposition~\ref{prop:sup_cost_bound} and Lemma~\ref{lem:length_of_rho_1_in_C}, we consider paths $\pi$ in the counter graph $\C{}$ that starts in a vertex $v$ such that $I(v) = J_\ell$. This means that such paths only visit vertices of regions $V^{J_j}$ such that $j \in \{\ell, \ell +1, \ldots, N\}$. There are therefore paths in the counter graph restricted $V^{\geq J_\ell}$ that we denote by $\mathbb{C}(\lambda_{\geq \ell})$.
\end{rem}

We now come back to the labeling functions ${(\lambda^k)}_{k \in \mathbb{N}}$ computed by Algorithm~\ref{algo:lambda}. Recall that this algorithm works in a bottom-up manner (see Proposition~\ref{prop:fixpoint}): it first computes the local fixpoint $\lambda^{k_N^*}$ on region $V^{J_N}$, then the local fixpoint $\lambda^{k_{N-1}^*}$ on $V^{\geq J_{N-1}}$, \dots, until finally computing the global fixpoint $\lambda^*$ on $V^{\geq J_{1}} = V^X$. Recall also that when the algorithm computes the local fixpoint in the arena $\regionGraphGeqI{J_n}$, the values of $\lambda^k(v)$ may only change in the region $V^{J_{n}}$ (Lemma~\ref{lem:fixpoint-region}). We are now ready to show an exponential bound (in the size of $\mathcal{G}$) on the maximal finite ranges $\mFR{k^*_\ell}{\ell}$ for each region $X^{J_\ell}$ and $\mFR{k^*_\ell}{\geq \ell}$ for each arena $X^{\geq J_\ell}$. With $\ell = 1$, we get that $\mFR{*}{}$ is of exponential size, and thus also the size of the counter graph.

\begin{thm}[Bound on maximal finite range]\label{thm:bound_on_MR}
For each $\ell \in \lbrace 1,\dots,N \rbrace$, we have
\[ \mFR{k^*_\ell}{\ell} \leq \mathcal{O}(|V|^{(|V|+3) \cdot  (|\Pi \setminus J_\ell|+2)})
\]
and also:
\[ \mFR{k^*_\ell}{\geq \ell} \leq \mathcal{O}(|V|^{(|V|+3) \cdot ( |\Pi|+2)}).
\]
In particular for the global fixpoint $\lambda^*$ we have
\[\mFR{*}{} \leq \mathcal{O}(|V|^{(|V|+3) \cdot ( |\Pi|+2)}).
\]
\end{thm}

\begin{proof}
The proof is done by a \emph{double} induction:
First, we exploit the fact that Algorithm~\ref{algo:lambda} treats every region one after the other, following the total order on regions in reverse.
That is, to compute the values of the fixpoint function $\lambda^{*}$ over $V^X$, Algorithm~\ref{algo:lambda} computes first the values of $\lambda^{*}$ on region $V^{J_N}$, then on region $V^{J_{N-1}}$ etc\dots until finally on region $V^{J_1}$.
Thus, we follow this order to prove the local bounds on $\mFR{k^*_\ell}{\ell}$, starting by $\mFR{k^*_N}{N}$ and making our way up to $\mFR{k^*_1}{1}$, assuming for each region $V^{J_\ell}$ such that $\ell < N$ that the bound is true for every region already treated by Algorithm~\ref{algo:lambda}.
Second, given a non-bottom region $V^{J_\ell}$ and assuming the bound is true for every already treated region, we proceed to show the local bound for $V^{J_\ell}$ by induction on the number $k$ of steps in the computation, which corresponds to the values of function $\lambda^k$ in the sequence of functions leading to the fixpoint, up to step $k^*_\ell$, where the values stabilize on region $V^{\ell}$.

Let us now detail the proof. It is structured into several parts.

\paragraph{Part 1.}
We begin with some notations and basic properties. We introduce a useful notation: for each $\ell <N$ and each $k \in \mathbb{N}$, we define $\alpha(\lambda^k,\ell)$ as follows\footnote{This sum appears in the statement of Proposition~\ref{prop:sup_cost_bound}.}:

\[\alpha(\lambda^k,\ell) := \sum^{|\Pi|}_{r=|J_\ell|+1} |V| + 2 \cdot \max_{\stackrel{J_j > J_\ell}{ |J_j|=r}} \mFR{k}{j}\]

\noindent
Let $\ell \leq N$ and $v \in V^{J_\ell}$. By Lemma~\ref{lem:fixpoint-region}, we know the following:
\begin{enumerate}
    \item the values of $\lambda^{k+1}_\ell(v)$ and $\lambda^{k}_\ell(v)$ may differ only when $ k^*_{\ell+1} < k < k^*_\ell$;
    \item for $k \leq k^*_{\ell+1}$, we have $\lambda^{k}_\ell(v) = \lambda^{k+1}_\ell(v) = \lambda^{0}_\ell(v)$;
    \item\label{item:iii} for $k \geq k^*_{\ell}$, we have $\lambda^{k}_\ell(v) = \lambda^{k+1}_\ell(v) = \lambda^{k^*_\ell}(v)$
\end{enumerate}

\noindent
By~\ref{item:iii}., we have that for each $\ell<N$, and for $k \geq k^*_{\ell+1}$, we have $\mFR{k}{\ell+1} = \mFR{k^*_{\ell+1}}{\ell+1}$.
Thus, for $k \geq k^*_{\ell+1}$, we also have
\begin{align}
\alpha(\lambda^k,\ell)= \alpha(\lambda^{k^*_{\ell+1}}, \ell)   \label{eq:alpha_l+1}.
\end{align}

\paragraph{Part 2.}
Let $\ell < N$ and assume $J_\ell$ is not a bottom region. We start by proving, this time by induction on the number $k$ of algorithm steps, a bound on $\mFR{k}{\ell}$. Note, by Lemma~\ref{lem:fixpoint-region} recalled above, that the only relevant steps for region $J_\ell$ are the steps $k$, where $k^*_{\ell+1} \leq k \leq k^*_\ell$.
Let us show, for each such $k$, that:

\begin{align}
   \mFR{k}{\ell} \leq \left(\sum^{|\dom{k}{\ell}|}_{i=0} 2^i \right) \cdot (1 + |V| + \alpha(\lambda^{k^*_{\ell+1}},\ell) )  \label{eq:hyp_ind_k}
\end{align}
where $\dom{k}{\ell} = \lbrace v \in V^{J_\ell} ~|~ \lambda^k_\ell(v)\neq 0, \lambda^k_\ell(v) \neq +\infty \rbrace$.

\begin{itemize}
    \item \textit{Base case:} Assume $k=k^*_{\ell+1}$.
By Lemma~\ref{lem:fixpoint-region}, we know that $\lambda^k_\ell=\lambda^0_\ell$.
Thus, $\mFR{k}{\ell}=\mFR{0}{\ell}=0$.
Furthermore, $\dom{k}{\ell}=\dom{0}{\ell}=\emptyset$, thus $\sum^{|\dom{k}{\ell}|}_{i=0} 2^i =1$.
Clearly, Equation~\eqref{eq:hyp_ind_k} is satisfied.

\item \textit{General case:} Let now $k$ be such that $k^*_{\ell+1} \leq k < k^*_\ell$.
Assume that Equation~\eqref{eq:hyp_ind_k} holds for $k$.
Let us show that
\begin{align}
   \mFR{k+1}{\ell} \leq \left(\sum^{|\dom{k+1}{\ell}|}_{i=0} 2^i \right) \cdot (1 + |V| + \alpha(\lambda^{k^*_{\ell+1}},\ell) )  \label{eq:cas_k+1}
\end{align}
We distinguish the two following subcases: $(a)$ when $\dom{k+1}{\ell}=\dom{k}{\ell}$, and $(b)$ when $\dom{k+1}{\ell}\neq\dom{k}{\ell}$.

\begin{itemize}[align=left]
    \item[$(a)$] Assume $\dom{k+1}{\ell}=\dom{k}{\ell}$. We know by Lemma~\ref{lemma:sequenceNonIncreasing} that for each $v\in V^{J_\ell}$, $\lambda^{k+1}_\ell(v) \leq \lambda^k_\ell(v)$.
    Thus in the considered case we have:
    \begin{align*}
        \mFR{k+1}{\ell} &\leq \mFR{k}{\ell} & \\
                        &\leq \left(\sum^{|\dom{k}{\ell}|}_{i=0} 2^i \right) \cdot (1 + |V| + \alpha(\lambda^{k^*_{\ell+1}},\ell) )  \tag{by induction hyp.~\eqref{eq:hyp_ind_k}} \\
                        &\leq \left(\sum^{|\dom{k+1}{\ell}|}_{i=0} 2^i \right) \cdot (1 + |V| + \alpha(\lambda^{k^*_{\ell+1}},\ell) )  &\tag{by ($a$)}
    \end{align*}
    That is, Equation~\eqref{eq:cas_k+1} holds.

    \item[$(b)$] Assume $\dom{k+1}{\ell}\neq \dom{k}{\ell}$.
    Either we have $\mFR{k+1}{\ell} \leq \mFR{k}{\ell}$, and we can proceed as in subcase $(a)$, or
    we know that there exists $v \in V^{J_\ell}$ such that $\lambda^k_\ell(v)=+\infty$ and $\lambda^{k+1}_\ell(v) < +\infty$.
    Recall that by Definition~\ref{def:update} we have that
    \begin{align}
        \lambda^{k+1}_\ell(v) = 1 + \min_{\substack{(v,v')\in \extE}}  \sup \{ \extCosti{i}(\rho) \mid \rho \in \Lambda^{k}(v') \}
        \end{align}

Since $\lambda^{k+1}_\ell(v) \neq +\infty$, we know that $ \sup \{ \extCosti{i}(\rho) \mid \rho \in \Lambda^{k}(v') \} $ is finite, for at least one successor $v'$ of $v$.
Thus, by Proposition~\ref{prop:sup_cost_bound}, we obtain:

\begin{align*}
    \lambda^{k+1}_\ell(v) &\leq 1 + |V| + 2\cdot \mFR{k}{\ell} + \alpha(\lambda^k, \ell) & \\
                          &\leq 1 + |V| + 2\cdot \mFR{k}{\ell} + \alpha(\lambda^{k^*_{\ell+1}}, \ell) \tag{by~\eqref{eq:alpha_l+1}} \\
                          &\leq 1 + |V| + \alpha(\lambda^{k^*_{\ell+1}}, \ell) &\\
                          & \quad + 2\cdot \left(\sum^{|\dom{k}{\ell}|}_{i=0} 2^i \right) \cdot (1 + |V| + \alpha(\lambda^{k^*_{\ell+1}},\ell) ) &  \tag{by ind.\ hyp.~\eqref{eq:hyp_ind_k}} \\
                          &\leq \left \lbrack 2\cdot \left(\sum^{|\dom{k}{\ell}|}_{i=0} 2^i \right) +1 \right \rbrack \cdot (1 + |V| + \alpha(\lambda^{k^*_{\ell+1}},\ell) ) &
\end{align*}

Since $\dom{k+1}{\ell}\neq \dom{k}{\ell}$ by $(b)$ (thus indeed $|\dom{k+1}{\ell}| > |\dom{k}{\ell}$|), we have \[ \left \lbrack 2\cdot \left(\sum^{|\dom{k}{\ell}|}_{i=0} 2^i \right) +1 \right \rbrack \leq \left(\sum^{|\dom{k+1}{\ell}|}_{i=0} 2^i \right)\]

Hence, we have
\begin{align*}
    \lambda^{k+1}_\ell(v) &\leq \left(\sum^{|\dom{k+1}{\ell}|}_{i=0} 2^i \right) \cdot (1 + |V| + \alpha(\lambda^{k^*_{\ell+1}}, \ell) )
\end{align*}

Finally, as this holds for any such $v\in V^{J_\ell}$, we obtain
\begin{align*}
    \mFR{k+1}{\ell} &\leq \left(\sum^{|\dom{k+1}{\ell}|}_{i=0} 2^i \right) \cdot (1 + |V| + \alpha(\lambda^{k^*_{\ell+1}}, \ell) )
\end{align*}

That is, Equation~\eqref{eq:cas_k+1} holds for region $J_\ell$.
\end{itemize}
\end{itemize}

\noindent
We proved that Equation~\eqref{eq:hyp_ind_k} holds for each non-bottom region $J_\ell$ and each step $k$ such that $k^*_{\ell+1} \leq k \leq k^*_\ell$.

\paragraph{Part 3.}
We can now come back to the induction on the regions $J_\ell$, following the order provided by Algorithm~\ref{algo:lambda}. Let us show that each $\ell \leq N$, we have
\begin{align}
    \mFR{k^*_\ell}{\ell} \leq 2^{(|V|+1)} \cdot (|\Pi||V| + |V| +1) \cdot \sum^{|\Pi \setminus J_\ell|}_{i=0} 2^{(|V|+1)} \cdot 2|\Pi| \label{eq:hyp_ind_ell}
\end{align}

\begin{itemize}

\item \textit{Base case:} If $J_\ell$ is a bottom region, we have that $\mFR{k^*_\ell}{\ell} = 0$ showing that (\ref{eq:hyp_ind_ell}) holds in this case and thus in particular when $\ell =N$.

\item \textit{General case:} Let $\ell < N$ and suppose $J_\ell$ is not a bottom region.
Assume that, for each $j> \ell$, Inequality (\ref{eq:hyp_ind_ell}) holds. By~\eqref{eq:hyp_ind_k}, we know that
\[ \mFR{k}{\ell} \leq \left(\sum^{|\dom{k}{\ell}|}_{i=0} 2^i \right) \cdot (1 + |V| + \alpha(\lambda^{k^*_{\ell+1}},\ell) ). \]

Furthermore, as $|\dom{k^*_\ell}{\ell}| \leq |V|$, we have
\begin{align}
  \mFR{k}{\ell} \leq 2^{(|V|+1)} \cdot (1 + |V| + \alpha(\lambda^{k^*_{\ell+1}},\ell) )  \label{eq:ind_sans_somme_puissance}
\end{align}

We turn now our attention towards the term $\alpha(\lambda^{k^*_{\ell+1}},\ell)$:
\begin{align*}
   \alpha(\lambda^{k^*_{\ell+1}},\ell) &= \sum^{|\Pi|}_{r=|J_\ell|+1} |V| + 2 \cdot \max_{\stackrel{J_j > J_\ell}{ |J_j|=r}} \mFR{k^*_{\ell+1}}{j} \tag{by Definition}\\
                                           &= \sum^{|\Pi|}_{r=|J_\ell|+1} |V| + 2 \cdot \max_{\stackrel{J_j > J_\ell}{ |J_j|=r}} \mFR{k^*_j}{j}
    \tag{by Lemma~\ref{lem:fixpoint-region}}
\end{align*}
By the induction hypothesis~(\ref{eq:hyp_ind_ell}), we can bound each term $\mFR{k^*_j}{j}$ appearing in the above sum:
\begin{align*}
    \mFR{k^*_j}{j} &\leq 2^{(|V|+1)} \cdot (|\Pi||V| + |V| +1) \cdot \sum^{|\Pi \setminus J_j|}_{i=0} 2^{(|V|+1)} \cdot 2|\Pi| &  \\
                   &\leq 2^{(|V|+1)} \cdot (|\Pi||V| + |V| +1) \cdot \sum^{|\Pi \setminus J_\ell|-1}_{i=0} 2^{(|V|+1)} \cdot 2|\Pi| \tag{as $|J_j|>|J_\ell|$}
\end{align*}
Coming back to $\alpha(\lambda^{k^*_{\ell+1}},\ell)$ whose the number of terms in the sum can be bounded by $|\Pi|$, we get
\begin{align}
   \alpha(\lambda^{k^*_{\ell+1}},\ell) &= \sum^{|\Pi|}_{r=|J_\ell|+1} |V| + 2 \cdot \max_{\stackrel{J_j > J_\ell}{ |J_j|=r}} \mFR{k^*_{\ell+1}}{j} & \nonumber \\
   &\leq |\Pi||V| + 2 |\Pi| \left( 2^{(|V|+1)} \cdot (|\Pi||V| + |V| +1) \cdot \sum^{|\Pi \setminus J_\ell|-1}_{i=0} 2^{(|V|+1)} \cdot 2|\Pi|\right). & \label{eq:termes_somme}
\end{align}
Hence we have, by combining~\eqref{eq:termes_somme} and~\eqref{eq:ind_sans_somme_puissance}:
\begin{align*}
     \mFR{k^*_\ell}{\ell} &\leq 2^{(|V|+1)} \cdot (1 + |V| + \alpha(\lambda^{k^*_{\ell+1}},\ell) ) & \nonumber \\
                   &\leq 2^{(|V|+1)} \cdot ( 1 + |V| + |\Pi||V| )& \nonumber \\
                   & + 2^{(|V|+1)} \cdot 2 |\Pi| \left(2^{(|V|+1)} \cdot (|\Pi||V| + |V| +1) \cdot \sum^{|\Pi \setminus J_\ell|-1}_{i=0} 2^{(|V|+1)} \cdot 2|\Pi|\right)  & \nonumber \\
                   &\leq 2^{(|V|+1)} \cdot ( 1 + |V| + |\Pi||V| ) \cdot \left \lbrack 1 + 2^{(|V|+1)} \cdot 2 |\Pi| \left( \sum^{|\Pi \setminus J_\ell|-1}_{i=0} 2^{(|V|+1)} \cdot 2|\Pi|\right) \right \rbrack & \nonumber \\
                   &\leq 2^{(|V|+1)} \cdot ( 1 + |V| + |\Pi||V| ) \cdot \left( \sum^{|\Pi \setminus J_\ell|}_{i=0} 2^{(|V|+1)} \cdot 2|\Pi|\right). & \nonumber \\
\end{align*}

Hence, Equation~\eqref{eq:hyp_ind_ell} holds for region $J_\ell$.
\end{itemize}

\paragraph{Part 4.}
We can now prove the three statements of Theorem~\ref{thm:bound_on_MR}. We obtain the first one from Inequality~(\ref{eq:hyp_ind_ell}) by recalling that $|\Pi| \leq |V|$ and $|V|\geq 2$:

\begin{align}
    \mFR{k^*_\ell}{\ell} &\leq 2^{(|V|+1)} \cdot ( 1 + |V| + |\Pi||V| ) \cdot \left( \sum^{|\Pi \setminus J_\ell|}_{i=0} 2^{(|V|+1)} \cdot 2|\Pi|\right) \nonumber \\
                  &\leq |V|^{(|V|+1)} \cdot ( 1 + |V| + |V|^2 ) \cdot \left( \sum^{|\Pi \setminus J_\ell|}_{i=0} |V|^{(|V|+3)} \right) \nonumber \\
                  &\leq |V|^{(|V|+1)} \cdot ( 1 + |V| + |V|^2 ) \cdot {(|V|^{(|V|+3)})}^{(|\Pi \setminus J_\ell|+1)} \nonumber \\
                  &\leq \mathcal{O}(|V|^{(|V| + 3)}) \cdot {(|V|^{(|V|+3)})}^{(|\Pi \setminus J_\ell|+1)} \nonumber \\
   &\leq \mathcal{O}\left( |V|^{(|V|+3|)(|\Pi \setminus J_\ell| +2)}\right) \label{eq:borne_region_*}
\end{align}
This proves the first statement of Theorem~\ref{thm:bound_on_MR}.

For the second one, remark that $\mFR{k^*_\ell}{\geq \ell} = \max \lbrace \mFR{k^*_j}{j} ~|~ J_j \geq J_\ell \rbrace$ (since only regions already treated are considered).
For each such $j\geq \ell$, we have $\mFR{k^*_j}{j} \leq \mathcal{O}\left( |V|^{(|V|+3|)(|\Pi \setminus J_j| +2)}\right) $ by~\eqref{eq:borne_region_*} and thus $\mFR{k^*_j}{j} \leq \mathcal{O}\left( |V|^{(|V|+3|)(|\Pi| +2)}\right)$. It follows that
$\mFR{k^*_\ell}{\geq \ell} \leq \mathcal{O}\left( |V|^{(|V|+3|)(|\Pi| +2)}\right)$.

For the last statement, recall that $\mFR{*}{}= \mFR{k^*_1}{\geq 1}$. Hence in particular we obtain $\mFR{*}{} \leq \mathcal{O}\left( |V|^{(|V|+3|)(|\Pi| +2)}\right)$.
\end{proof}

The next corollary is a direct consequence of Theorem~\ref{thm:bound_on_MR}.

\begin{cor}%
\label{cor:exponential}
The counter graph $\C{*}$ has a size $|\C{*}| = |V|\cdot 2^{|\Pi|} \cdot {(\mFR{*}{} + 2)}^{|\Pi|}$ that is exponential in the size of the game $\mathcal G$.
\end{cor}

We conclude this section with two other corollaries that will be useful in the next section.

\begin{cor}\label{cor:bound_every_k}
For every $\N{k}$ and region $V^{J_\ell}$, we have
\[\mFR{k}{\ell} \leq \mathcal{O}\left( |V|^{(|V|+3)(|\Pi| +2)}\right)  \]
and also
\[\mFR{k}{\geq \ell} \leq \mathcal{O}\left( |V|^{(|V|+3)(|\Pi| +2)}\right). \]
\end{cor}

\begin{proof}
Assume, without loss of generality, that $J_\ell$ is not a bottom region (otherwise we immediately have $\mFR{k}{\ell}=0$ for every $\N{k}$).
Let $\N{k}$.
Again, by Lemma~\ref{lem:fixpoint-region}, if $k \leq k^*_{\ell+1}$, we already have $\mFR{k}{\ell}=0$, and if $k \geq k^*_{\ell}$, we have $\lambda^{k}_\ell(v) = \lambda^{k^*_\ell}(v)$. Thus we assume that $k^*_{\ell+1} < k \leq k^*_{\ell}$.
Using the terminology of the proof of Theorem~\ref{thm:bound_on_MR}, we know, by~\eqref{eq:hyp_ind_k} in its proof, that

\begin{align}
    \mFR{k}{\ell} \leq \left(\sum^{|\dom{k}{\ell}|}_{i=0} 2^i \right) \cdot (1 + |V| + \alpha(\lambda^{k^*_{\ell+1}},\ell) )
\end{align}

As $|\dom{k}{\ell}|\leq |V|$, we have

\begin{align}
    \mFR{k}{\ell} \leq 2^{|V|+1} \cdot (1 + |V| + \alpha(\lambda^{k^*_{\ell+1}},\ell) ) \label{eq:alpha_k_ell}
\end{align}

From there, following the same steps (from Equation~\eqref{eq:ind_sans_somme_puissance} onwards) as in the proof of Theorem~\ref{thm:bound_on_MR} leads to the desired bound: $\mFR{k}{\ell} \leq \mathcal{O}\left( |V|^{(|V|+3)(|\Pi| +2)}\right)$.

Similarly, as $\mFR{k}{\geq \ell} = \max\lbrace \mFR{k}{j} ~|~ j\geq \ell \rbrace $, and as we just showed that $\mFR{k}{j} \leq \mathcal{O}\left( |V|^{(|V|+3|)(|\Pi| +2)}\right)$ for every $j$ and $k$, we immediately get
\[
    \mFR{k}{\geq \ell}\leq \mathcal{O}\left( |V|^{(|V|+3|)(|\Pi| +2)}\right)
    \qedhere
\]
\end{proof}

\begin{cor}\label{cor:bound_sup}
Let $v\in V^X$ with $I(v) = J_\ell$ with $\ell \in \{1, \ldots N-1\}$.
Let $\N{k}$. Suppose there exists $\N{c}$ such that $\sup~ \{ \Cost_i(\rho) \mid \rho \in \Lambda^{k}(v)\} = c$.
Then, the following holds:

\[ c \leq \mathcal{O}\left( |V|^{(|V|+3)(|\Pi| +2)}\right).\]
\end{cor}
\begin{proof}
Let $v\in V^X$ with $I(v) = J_\ell$ and $\N{k}$.
Suppose there exists $\N{c}$ such that $\sup~ \{ \Cost_i(\rho) \mid \rho \in \Lambda^{k}(v)\} = c$.
If $J_\ell$ is a bottom region, then $c=0$, thus we assume from now on that $J_\ell$ is not a bottom region.
Then, by Proposition~\ref{prop:sup_cost_bound}, we know that:

\[ c \leq |V| + 2 \cdot \mFR{k}{\ell} + \sum^{|\Pi|}_{r=|J_\ell|+1} |V| + 2 \cdot \max_{\stackrel{J_j > J_\ell}{ |J_j|=r}} \mFR{k^*_j}{j} .\]

Using the terminology of the proof of Theorem~\ref{thm:bound_on_MR}, we have:

\begin{align*}
    c &\leq |V| + 2 \cdot \mFR{k}{\ell} + \alpha(\lambda^{k^*_{\ell+1}},\ell)& \nonumber\\
      &\leq 1 + |V| + 2 \cdot \mFR{k}{\ell} + \alpha(\lambda^{k^*_{\ell+1}},\ell)&  \nonumber\\
      &\leq 1 + |V| + \alpha(\lambda^{k^*_{\ell+1}}, \ell) + 2\cdot \left(\sum^{|\dom{k}{\ell}|}_{i=0} 2^i \right) \cdot (1 + |V| + \alpha(\lambda^{k^*_{\ell+1}},\ell) ) \tag{by~\eqref{eq:hyp_ind_k}} \\
      &\leq \left(\sum^{|\dom{k}{\ell}|+1}_{i=0} 2^i \right) \cdot (1 + |V| + \alpha(\lambda^{k^*_{\ell+1}}, \ell))
      \end{align*}

  As $|\dom{k}{\ell}| \leq |V| $, we obtain:
 \begin{align*}
     c &\leq  2 \cdot \left[2^{(|V|+1)} \cdot (1 + |V| + \alpha(\lambda^{k^*_{\ell+1}},\ell)) \right].
     \end{align*}

From there, following the same steps (from Equation~\eqref{eq:ind_sans_somme_puissance} onwards) as in the proof of Theorem~\ref{thm:bound_on_MR} leads to the desired bound: \[c \leq 2 \cdot \mathcal{O}\left( |V|^{(|V|+3)(|\Pi| +2)}\right) =  \mathcal{O}\left( |V|^{(|V|+3)(|\Pi| +2)}\right). \qedhere\]
\end{proof}

\section{PSPACE completeness}%
\label{section:PSPACEc}

In this section, we prove Theorem~\ref{thm:main}. We first prove that the constraint problem is in PSPACE and then that it is PSPACE-hard.

\subsection{PSPACE easiness}

The purpose of this section is to prove that determining if, given a reachability game ($\mathcal{G}$,$v_0$) and two thresholds $x,y \in {(\mathbb{N} \cup \{+\infty \})}^{|\Pi|}$, there exists an SPE $\sigma$ in this game such that for all $i \in \Pi$, $x_i \leq \Cost_i(\outcome{\sigma}{v_0}) \leq y_i$ can be done in PSPACE\@.

 \begin{prop}%
 \label{prop:PSPACE-easiness}
    The constraint problem for initialized reachability games is in PSPACE\@.
 \end{prop}

Let us provide a high level sketch of the proof of our PSPACE procedure for this constraint problem. Thanks to Theorem~\ref{thm:folkThm}, solving the constraint problem for a given game $({\mathcal G},v_0)$ reduces in finding a $\lambda^*$-consistent play $\rho$ in $(\extGame,x_0)$ satisfying the constraints. By Lemmas~\ref{lem:lambda_consistent_rho_X_to_rho_C} and~\ref{lem:chemin_infini_compteurs}, the latter problem reduces in finding a valid path $\pi$ in the counter graph $\mathbb{C}(\lambda^*)$ that satisfies the constraints. We will see that it suffices to decide the existence of such a valid path that is a lasso $hg^\omega$. As $\mathbb{C}(\lambda^*)$ is exponential in the size of the input $\mathcal G$ (by Corollary~\ref{cor:exponential}), classical arguments using Savitch's Theorem can thus be used to prove the PSPACE membership. Nevertheless, the detailed proof is more intricate for two reasons. The first reason is that the counter graph is constructed from the labeling function $\lambda^*$. We thus also have to prove that $\lambda^*$ can be computed in PSPACE\@. The second reason is that, a priori, although we know that the counter graph is of exponential size, we do not know explicitly its size. This is problematic when using classical NPSPACE algorithms that guess, vertex by vertex, some finite path in a graph of exponential size, where a counter bounded by the size of the graph is needed to guarantee the termination of the procedure. In order to overcome this, we also need a PSPACE procedure to obtain the actual size of $\mathbb{C}(\lambda^*)$. Recall that the size $|\mathbb{C}(\lambda^*)|$ is equal to $|V| \cdot 2^{|\Pi|} \cdot {(K+2)}^{|\Pi|}$ where $K = \mFR{*}{}$. Hence to compute the size of the counter graph, we have to compute the actual value of $\mFR{*}{}$.

The PSPACE procedure to compute $\lambda^*$ and $\mFR{*}{}$ works by induction on $k$, the steps in the computation of the labeling function $\lambda^*$. Moreover, it exploits the structural evolution of the local fixpoints formalized in Proposition~\ref{prop:fixpoint} and Lemma~\ref{lem:fixpoint-region}. Indeed recall that these local fixpoints are computed region by region, from $X^{J_N}$ to $X^{J_1}$, and that if $X^{J_\ell}$ is the currently treated region, then the values of $\lambda^{k+1}(v)$ are computed from $\lambda^{k}(v)$ for all $v \in V^{J_{\ell}}$ (those values remain unchanged for all $v$ outside of $V^{J_{\ell}}$).
For the moment, suppose that we have at our disposal a PSPACE procedure to compute $\{\lambda^k(v) \mid v \in V^{J_{\ell}} \}$ and the maximal finite range $\mFR{k}{\geq \ell}$:

\begin{prop}%
\label{prop:regionalComputation} Given an initialized reachability game $(\mathcal{G},v_0)$, for all $k \in \mathbb{N}$ and for all $J_{\ell}$, $\ell \in \{1, \ldots, N \}$, the set $\{\lambda^k(v) \mid v \in V^{J_{\ell}} \}$ and the maximal finite range $\mFR{k}{\geq \ell}$ can be computed in PSPACE\@.
\end{prop}

Let us prove Proposition~\ref{prop:PSPACE-easiness}. The proof of Proposition~\ref{prop:regionalComputation} will be given just after.

\begin{proof}[Proof of Proposition~\ref{prop:PSPACE-easiness}]

Let $(\mathcal{G},v_0)$ be an initialized reachability game and let $x,y \in {(\mathbb{N} \cup \{+\infty\})}^{|\Pi|}$ be two thresholds. Let $(\extGame,x_0)$ be the extended game of $(\mathcal{G},v_0)$, $\mathbb{C}(\lambda^*)$ the counter graph constructed from the labeling function $\lambda^*$ and its maximal finite range $\mFR{*}{}$.

We first prove that there exists an SPE in $(\mathcal G, v_0)$ such that its outcome $\rho$ satisfies the constraints $x_i \leq \Cost_i(\rho) \leq y_i$ for all $i \in \Pi$, if and only if, there exists a valid path in $\mathbb{C}(\lambda^*)$ starting from the starting vertex $x_0^C$ associated with $x_0$ such that it is a lasso $hg^\omega$ with the length $|hg|$ bounded by $d + 2 \cdot |\mathbb{C}(\lambda^*)|$ with
\[d := \max \{x_i \mid x_i < +\infty\}\]
and such that it also satisfies these constraints.

We already know by Theorem~\ref{thm:folkThm} and Lemmas~\ref{lem:lambda_consistent_rho_X_to_rho_C}-\ref{lem:chemin_infini_compteurs} that the existence of an SPE outcome $\rho$ in $\mathcal G$ satisfying the constraints is equivalent to the existence of a valid path $\pi$ in $\mathbb{C}(\lambda^*)$ satisfying these constraints. It remains to show that the latter path can be chosen as a lasso $hg^\omega$ with the announced length of $hg$. This lasso is constructed as follows. Consider the suffix $\pi_{\geq d}$ of $\pi$ and its region decomposition $\pi_{\geq d}[m]\pi_{\geq d}[m+1]\ldots \pi_{\geq d}[n]$. For all $\ell \in \{m,\ldots,n-1\}$, we remove all the cycles in section $\pi_{\geq d}[\ell]$ to get a simple path $\pi'_\ell$, and from the last section $\pi_{\geq d}[n]$, we derive the infinite path $\pi'_n$ formed of the first vertices of $\pi_{\geq d}[n]$ until a cycle is reached and then repeated forever. Notice that there is no other cycle in $\pi'_{m}\pi'_{m+1}\ldots\pi'_{n}$ by the $I$-monotonicity property (\ref{eq:increasing}). The required lasso $hg^\omega$ is equal to the concatenation of the prefix $\pi_{\leq d}$ with the modified suffix $\pi'_{m}\pi'_{m+1}\ldots\pi'_{n}$. By construction $\pi'_{m}\ldots\pi'_{n}$ is itself a lasso $h'g'^\omega$ with $|h'g'|$ bounded by $2 \cdot |\mathbb{C}(\lambda^*)|$ and thus $|hg|$ is bounded by $d + 2 \cdot |\mathbb{C}(\lambda^*)|$. Moreover, for all $i$, $\pi$ visits the target set of player $i$ if and only if $hg^\omega$ visits this set (maybe earlier if the visit is inside $\pi'_{m}\ldots\pi'_{n}$). By definition of $d$, the constraints imposed by $x_i, y_i$, $i \in \Pi$, are simultaneously satisfied by $\pi$ and $hg^\omega$.

Let us now show how to get a PSPACE procedure for the constraint problem. As just explained, we have to guess a lasso $\pi = hg^\omega$ in $\mathbb{C}(\lambda^*)$ that starts in $x_0^C$, satisfies the constraints, and such that $|hg|$ is bounded by
\[L = d + 2 \cdot |\mathbb{C}(\lambda^*)|.\]
We cannot guess $\pi$ entirely and we have to proceed region by region. Suppose that $I(x_0) = J_m$ for some $m \in \{ 1, \ldots, N \}$, and consider the region decomposition $\pi[m]\pi[m+1]\ldots \pi[n]$ of $\pi$, where some sections $\pi[\ell]$ may be empty.

We guess successively the sections $\pi[m]$, $\pi[m+1]$ and so on. To guess $\pi[\ell]$ with $\ell \in \{m,\ldots, n\}$, assuming it is not empty, we guess one by one its vertices that all belong to the same region $V^{J_\ell}$. To guess such a vertex $(v',{(c'_i)}_{i\in \Pi})$, we only have to keep its predecessor $(v,{(c_i)}_{i\in \Pi})$ in memory and to know the value $\lambda^*(v')$ (in a way to compute each $c'_i$ from $c_i$). So, we need to know $\{\lambda^*(v) \mid v \in V^{J_\ell} \}$. By Proposition~\ref{prop:regionalComputation}, we can compute this set in PSPACE\@. Once we move to another region in a way to guess the next section of $\pi$, we can forget this set and compute the new one. We also need to guess which vertex will be the first vertex of $g$.

Notice that any vertex of the counter graph can be encoded in polynomial size memory. Indeed it is composed of a vertex of $V$, a subset $I$ of $\Pi$, and $|\Pi|$ counter values that belong to $\{0,\ldots, \mFR{*}{}\} \cup \{+ \infty\}$  ($\mFR{*}{}$ is at most exponential in the input by Theorem~\ref{thm:bound_on_MR}). Moreover the set $\{\lambda^*(v) \mid v \in V^{J_\ell} \}$ can also be encoded in polynomial size memory since it is composed of $|V|$ values that belong to $\{0,\ldots, \mFR{*}{}\} \cup \{+ \infty\}$.

Recall that the length  $|hg|$ for the guessed lasso $\pi = hg^\omega$ cannot exceed constant $L$. We thus have to compute and store $L$. The computation in PSPACE of $L$ requires the computation of $|\mathbb{C}(\lambda^*)|$, and thus in particular the computation of $\mFR{*}{}$. This is possible thanks to Proposition~\ref{prop:regionalComputation}. And it follows by Corollary~\ref{cor:exponential} that $L$ can be stored in polynomial size memory. In addition to $L$, during the guessing of $\pi$, we also have a counter $C_L$ to count the current length of $\pi$, and for each player $i \in \Pi$ a counter $C_i$ keeping track of the current cost of player $i$ along $\pi$. As $C_L \leq L$ and $C_i \leq L$ for all $i$, all these counters can be also encoded in polynomial size memory.

Finally, we stop guessing $\pi$ when either its length exceeds $L$ or when its currently guessed vertex is equal to the first vertex of cycle $g$ for the second time. In the latter case, we check whether $\pi$ satisfies the constraints, that is, $x_i \leq C_i \leq y_i$ for all $i$. This completes the proof that the given procedure works in PSPACE\@.  \end{proof}

Let us now prove Proposition~\ref{prop:regionalComputation}. We proceed by induction on the steps in the computation of the labeling function $\lambda^*$ and, once a step $k$ is fixed, we proceed region by region, beginning with the bottom region $J_N$ and then proceeding bottom-up by following the total order $J_1 < \cdots < J_N$. Let $X^{J_\ell}$ be a region, we aim at proving that the set $\{ \lambda^{k+1}(v) \mid v \in V^{J_\ell} \}$ and the value $\mFR{k+1}{\geq \ell}$ are both computable in PSPACE (from the previous step $k$). To this aim consider Proposition~\ref{prop:fixpoint} and especially Lemma~\ref{lem:fixpoint-region}. Let $k^*_\ell$ (resp.\ $k^*_{\ell+1}$) be the step where the local fixpoint is reached for region $X^{J_\ell}$ (resp.\ $X^{J_{\ell+1}}$). Recall that $k^*_{\ell+1} < k^*_{\ell}$ and that when $k \leq k^*_{\ell+1}$ (resp.\ $k \geq k^*_{\ell}$), we have that $\lambda^{k+1}(v)=\lambda^{k}(v)$, for each $v \in V^{J_\ell}$. The tricky case in when $k^*_{\ell+1} < k < k^*_{\ell}$. In the latter case, the computation of $\lambda^{k+1}(v)$ from $\lambda^{k}(v)$, for all $v \in V^{J_{n}}$, relies on the computation of the maximal cost, for the player who owns vertex $v$, of the plays of $\Lambda^k(v')$, with $v' \in \Succ(v)$ (see Definition~\ref{def:update}). This will be possible with the same approach as in the proof of Proposition~\ref{prop:PSPACE-easiness}: to guess a lasso in the counter graph $\mathbb{C}(\lambda^k)$ that realizes this maximal cost.

\begin{proof}[Proof of Proposition~\ref{prop:regionalComputation}]
We proceed by induction on $k \in \mathbb{N}$.

\paragraph{Base case.}
 For $k = 0$ we have to prove that for all $J_{\ell}$, $\ell \in \{1, \ldots, N\}$, the set $\{ \lambda^0(v) \mid v \in V^{J_\ell}\}$ and the value $\mFR{0}{\geq \ell}$ can be both computed in PSPACE\@. Given $J_{\ell}$, thanks to Definition~\ref{def:init}, we have that either $\lambda^0(v) = 0$ if $v \in V_i^X$ and $i \in I(v) = J_\ell$ or $\lambda^0(v) = +\infty$ otherwise. Thus $\mFR{0}{\geq l} = 0$. So, we clearly have a PSPACE procedure in this case.

\paragraph{General case.}
Now, assume that for all $n$, $n \in \{0, \ldots, k\}$, and for all $J_{\ell}$, $\ell \in \{1,\ldots, N\}$, the set $\{ \lambda^n(v) \mid v \in V^{J_{\ell}}\}$ and the value $\mFR{n}{\geq \ell}$ can be computed in PSPACE\@. Let us prove that it remains true for $n = k+1$, that is:
\begin{equation} \{\lambda^{k+1}(v) \mid v \in V^{J_{\ell}} \} \text{ and } \mFR{k+1}{\geq \ell} \text{ can be computed in PSPACE } \label{eq:hypInduction} \end{equation}
We proceed by induction on the region $X^{J_\ell}$ (we thus use a double induction, one on the computation steps and the other one on the regions).

\medskip
If $J_{\ell}= J_N$, then it is a bottom region. Then for all $v \in V^{J_N}$, $\lambda^{k+1}(v) = \lambda^0(v)$ and so $\{\lambda^{k+1}(v) \mid v \in V^{J_N}\} = \{\lambda^0(v) \mid v \in V^{J_N} \}$ (local fixpoint $k_N^* = 0$ in Proposition~\ref{prop:fixpoint}). By induction hypothesis we know that this latter set can be computed in PSPACE\@. Moreover, we have that $\mFR{k+1}{\geq N} = \mFR{k+1}{N} = 0$. Therefore Assertion~(\ref{eq:hypInduction})  holds.

\medskip
Now, let $J_{\ell}$ be a region different from $J_N$. If it is a bottom region, we have $\{\lambda^{k+1}(v) \mid v \in V^{J_\ell}\} = \{\lambda^0(v) \mid v \in V^{J_\ell} \}$ as for $J_N$. Moreover $\mFR{k+1}{\ell} = \mFR{0}{\ell} = 0$ and then $\mFR{k+1}{\geq \ell} = \max\{\mFR{k+1}{\ell}, \mFR{k+1}{\geq \ell + 1} \} = \mFR{k+1}{\geq \ell + 1}$. Thus by induction hypothesis, $\{\lambda^{k+1}(v) \mid v \in V^{J_\ell}\}$ and $\mFR{k+1}{\geq \ell}$ can be computed in PSPACE and Assertion~(\ref{eq:hypInduction}) holds.

Let us now suppose that $J_{\ell}$ is not a bottom region. We first recall Proposition~\ref{prop:fixpoint} that states that Algorithm~\ref{algo:lambda} reaches a local fixpoint for each region. Let $k^*_\ell$ (resp.\ $k^*_{\ell+1}$), the step after which the region $X^{J_{ \ell}}$ (resp.\ $X^{J_{ \ell +1}}$) has reached its local fixpoint. Recall that $k^*_{\ell +1} < k^*_{\ell}$. Let us now consider the three cases of Lemma~\ref{lem:fixpoint-region}.

\begin{itemize}
\item If $k \leq k^*_{\ell+1}$, then by Lemma~\ref{lem:fixpoint-region}, the region $X^{J_{\ell+1}}$ has not reached its local fixpoint yet. So, it implies that the labeling of the vertices of $V^{J_\ell}$ has not change since initialization. More formally, for all $v \in V^{J_{\ell}}$, $\lambda^{k+1}(v) = \lambda^0(v)$. Thus $\mFR{k+1}{\ell} = 0$ and then $\mFR{k+1}{\geq \ell} = \mFR{k+1}{\geq \ell + 1}$. So by induction hypothesis, both $\{\lambda^{k+1}(v) \mid v \in V^{J_{\ell}}\}$ and $\mFR{k+1}{\geq \ell}$ can be computed in PSPACE showing (\ref{eq:hypInduction}).

\item If $k > k^*_{\ell}$, then by Lemma~\ref{lem:fixpoint-region}, the local fixpoint of region $X^{J_\ell}$ is reached, that is, $\{\lambda^{k+1}(v) \mid v \in V^{J_{\ell}} \} = \{\lambda^{k}(v) \mid v \in V^{J_{\ell}} \} $ and $\mFR{k+1}{\geq \ell} =\mFR{k}{\geq \ell}$ and~\eqref{eq:hypInduction} holds by induction hypothesis.

(Notice the little difference with the inequalities given in Lemma~\ref{lem:fixpoint-region}: we here consider case $k > k^*_{\ell}$ instead of case $k \geq k^*_{\ell}$ of Lemma~\ref{lem:fixpoint-region}. Indeed when $k = k^*_\ell$, we still need to compute $\lambda^{k+1}$ to realize that the fixpoint is effectively reached. This is thus postponed in the next case.)

\item It remains to consider the case $k^*_{\ell+1} < k \leq k^*_{\ell}$, which is the most difficult one. In this case, either the values of $\lambda^k(v)$ and $\lambda^{k+1}(v)$ differ for some $v \in V^{J_\ell}$, or $k=k^*_\ell$ and we realize that the local fixpoint is effectively reached on $X^{J_\ell}$.

Let us first show that the set $\{\lambda^{k+1}(v) \mid v \in V^{J_\ell}\}$ can be computed in PSPACE\@. Given $v \in V^{J_\ell}$, if $v \in V_i^X$, then by Definition~\ref{def:update}, $\lambda^{k+1}(v)$ is either equal to 0 (if $i\in J_\ell$) or it is computed from the values $\sup\{\Cost_i(\rho) \mid \rho \in \Lambda^k(v') \}$, $v' \in \Succ(v)$. Thus we have to show that each value $\sup\{\Cost_i(\rho) \mid \rho \in \Lambda^k(v') \}$ can be computed in PSPACE\@.

To this aim, we use Proposition~\ref{prop:sup_cost_bound} stating that if \[\sup\{\Cost_i(\rho) \mid \rho \in \Lambda^k(v') \} = c,\] then there exists a valid path $\pi=hg^\omega$ in $\C{k}$ starting in $v'^C$ that is a lasso with $|hg|$ bounded by $2 \cdot |\C{k}|$ and such that its corresponding play $\rho$ in $\extGame$ belongs to $\Lambda^{k}(v)$ and has its cost $\Cost_i(\rho)$ equal to $c$. Notice that we can restrict $\lambda^k$ to $\lambda^k_{\geq \ell}$ since any path beginning in $V^{J_\ell}$ only visits vertices of $V^{\geq J_{\ell}}$. We thus work in the counter graph $\mathbb{C}(\lambda^k_{\geq \ell})$ restricted to $V^{\geq \ell}$ (see Remark~\ref{rem:countergraph}). We will guess such a lasso $\pi$ as done in the proof of Proposition~\ref{prop:PSPACE-easiness}. More precisely, as the value of $c$ is unknown, we will first test in PSPACE whether there exists such a lasso $\pi$ with cost $c = +\infty$. If yes, we are done, otherwise by considering increasing values $d \in \mathbb{N}$, we will test in PSPACE whether there exists a lasso $\pi$ with cost $\geq d$. The last value $d$ for which the answer is yes is the required cost $c$.
\begin{itemize}
        \item Let us detail the case $c = +\infty$. Similarly to the proof of Proposition~\ref{prop:PSPACE-easiness}, we guess the lasso $\pi = hg^\omega$ by on one hand guessing the first vertex of its cycle $g$ and on the other hand guessing the sections $\pi[\ell], \pi[\ell+1], \dots$ of $\pi$, one by one, from the starting vertex $v'^C \in V^{J_\ell}$. Given $m \geq \ell$, section $\pi[m]$ is guessed vertex by vertex, where each vertex belongs to $V^{J_m}$. By induction hypothesis, we can compute the set $\{\lambda^k(u) \mid u \in V^{J_m}\}$ and the value $\mFR{k}{\geq m}$ in PSPACE, and by Corollary~\ref{cor:bound_every_k} the previous set and each vertex of the counter graph $\mathbb{C}(\lambda^k_{\geq \ell})$ can be encoded in polynomial size memory. Once $\pi[m]$ is guessed we can forget the set $\{\lambda^k(u) \mid u \in V^{J_m} \}$ and guess the next section of $\pi$.

        Additionally, to check that the length of $hg$
        does not exceed the constant $L = 2 \cdot |\mathbb{C}(\lambda^k_{\geq \ell})|$, we have to compute $\mFR{k}{\geq \ell}$ that can be done in PSPACE by induction hypothesis. Thus we can also compute $L$ in PSPACE and by Corollary~\ref{cor:bound_every_k} we can encode it in polynomial size memory. We also store a counter $C_L \leq L$ (which is initialized to $0$ and incremented by 1 each time we guess a new vertex of $\pi$) and a boolean $C_i$ (which is equal to $0$ as long as player $i$ does not visit his target set and is equal to $1$ after this visit).

        We stop if either if $C_L > L$ or if we have found the lasso. In the latter case we check whether $C_i = 0$ or not.

        \item Let us now proceed to the case $c < +\infty$.  We check whether there exists a lasso $\pi=hg^\omega$ with $|hg|$ bounded by $L = 2 \cdot |\mathbb{C}(\lambda^k_{\geq \ell})|$ with cost (for player~$i$) $\geq d$, for values $d=0, d=1, \ldots$, until the answer is no. The last value $d$ for which the answer is yes is equal to $c$.

        We know by Corollary~\ref{cor:bound_sup}, that the size of $c$ cannot exceed an exponential in the size of the input. Thus each tested value $d$ can be encoded in polynomial size memory.

        To check in PSPACE the existence of such a lasso with length of $hg$ bounded by $L$ with cost $\geq d$, we proceed exactly as for the case $c = +\infty$ except that instead of the boolean $C_i$ we keep a counter $C_i \leq L$ that keeps track of the cost of player $i$ for the lasso that we are guessing.

     \end{itemize}

     Notice that the depth of the recursion of the procedure is at most $|\Pi|$ by the $I$-monotonicity property (\ref{eq:increasing}) (any path crosses at most $|\Pi|$ regions). And at each recursion level, only a polynomial size information is stored. This concludes the proof that the set $\{\lambda^{k+1}(v) \mid v \in V^{J_\ell} \}$ can be computed in PSPACE\@.

     To conclude the case $k^*_{\ell+1} < k \leq k^*_{\ell}$, it remains to prove that the value $\mFR{k+1}{\geq \ell}$ can also be computed in PSPACE\@. Clearly we can compute $\mFR{k+1}{\ell}$ in PSPACE as we now have $\{\lambda^{k+1}(v) \mid v \in V^{J_{\ell}} \}$ in memory and we can compute $\mFR{k+1}{\geq \ell+1}$ by induction hypothesis. Notice that both values can be encoded in polynomial memory size by Corollary~\ref{cor:bound_every_k}. So, as $\mFR{k+1}{\geq \ell} = \max \{ \mFR{k+1}{ \ell}, \mFR{k+1}{\geq \ell +1}\}$, we can compute $\mFR{k+1}{\geq \ell}$ in PSPACE and Assertion~\eqref{eq:hypInduction} is proved. \qedhere
\end{itemize}
\end{proof}

\subsection{PSPACE hardness}

We now prove that the constraint problem is PSPACE-hard for quantitative reachability games.

\begin{prop}%
	\label{prop:reachPSPACE-hard}
	The constraint problem for quantitative reachability games is PSPACE-hard.
\end{prop}

The proof of this proposition is based on a polynomial reduction from the QBF problem which is PSPACE-complete. It is close to the proof given in~\cite{DBLP:journals/corr/abs-1806-05544} for the PSPACE-hardness of the constraint problem for Boolean games with reachability objectives --- either a player reaches his objective or not. The main difference is here to manipulate costs instead of considering qualitative reachability.

 The QBF problem is to decide whether a fully quantified Boolean formula $\psi$ is true. The formula $\psi$ can be assumed to be in prenex normal form $Q_1x_1Q_2x_2 \ldots Q_{m}x_m \,\phi(X)$ such that the quantifiers are alternating existential and universal quantifiers ($Q_1 = \exists$, $Q_2 = \forall$, $Q_3 = \exists, \ldots$), $X = \{ x_1, x_2, \ldots, x_m\}$ is the set of quantified variables, and $\phi(X) = C_1 \wedge \ldots \wedge C_n$ is an unquantified Boolean formula over $X$ equal to the conjunction of the clauses $C_1, \ldots, C_n$.

Such a formula $\psi$ is true if there exists a value of $x_1$ such that for all values of $x_2$, there exists a value of $x_3$ $\ldots$, such that the resulting valuation $\nu$ of all variables of $X$ evaluates $\phi(X)$ to true. Formally, for each odd (resp.\ even) $k$, $1 \leq k \leq m$, let us denote by $f_k: {\{0,1\}}^{k-1} \rightarrow \{0,1\}$ (resp.\ $g_k: {\{0,1\}}^{k-1} \rightarrow \{0,1\}$) a valuation of variable $x_k$ given a valuation of previous variables $x_1, \ldots, x_{k-1}$\footnote{Notice that $f_1: \emptyset \rightarrow \{0,1\}$.}. Given theses sequences $f= f_1, f_3, \ldots$ and $g = g_2, g_4, \ldots$, let us denote by $\nu = \nu_{(f,g)}$ the valuation of all variables of $X$ such that $\nu(x_1) = f_1$, $\nu(x_2) = g_2(\nu(x_1))$, $\nu(x_3) = f_3(\nu(x_1)\nu(x_2))$, $\ldots$. Then

\begin{center}
$\psi = Q_1x_1Q_2x_2 \ldots Q_{m}x_m \,\phi(X)$ is true \\
if and only if \\
there exist $f= f_1, f_3, \ldots$ such that for all $g = g_2,g_4, \ldots$, the valuation $\nu_{f,g}$ evaluates $\phi(X)$ to true.
\end{center}

\begin{proof}[Proof of Proposition~\ref{prop:reachPSPACE-hard}]
Let us detail a polynomial reduction from the QBF problem to the constraint problem for quantitative reachability games. Let $\psi = Q_1x_1Q_2x_2 \ldots Q_{m}x_m \,\phi(X)$ with $\phi(X) = C_1 \wedge \ldots \wedge C_n$ be a fully quantified Boolean formula in prenex normal form. We build the following quantitative reachability game $\mathcal{G}_{\psi} = (\Pi,V,{(V_i)}_{i\in \Pi}, E, {(\Cost_i)}_{i\in \Pi}, {(F_i)}_{i\in \Pi})$ (see Figure~\ref{figure:reachPSPACEh}):

\begin{itemize}
\item	the set $V$ of vertices:
	\begin{itemize}
	\item for each variable $x_k \in X$ under quantifier $Q_k$, there exist vertices $x_k$, $\neg x_k$ and $q_k$;
	\item for each clause $C_k$, there exist vertices $c_k$ and $t_k$;
	\item there exists an additional vertex $t_{n+1}$;
	\end{itemize}
\item the set $E$ of edges:
	\begin{itemize}
	\item from each vertex $q_k$ there exist an edge to $x_k$ and an edge to $\neg x_k$;
	\item from each vertex $x_k$ and $\neg x_k$, there exists an edge to $q_{k+1}$, except for $k=m$ where this edge is to $c_1$;
	\item from each vertex $c_k$, there exist an edge to $t_k$ and an edge to $c_{k+1}$, except for $k=n$ where there exist an edge to $t_n$ and an edge to $t_{n+1}$;
	\item there exists a loop on each $t_k$;
	\end{itemize}
\item the set $\Pi$ of $n+2$ players:
	\begin{itemize}
		\item each player $i$, $1 \leq i \leq  n$, owns vertex $c_i$;
		\item player $n+1$ (resp.\ $n+2$) is the player who owns the vertices $q_i$ for each existential (resp.\ universal) quantifier $Q_i$;
		\item as all other vertices have only one outgoing edge, it does not matter which player owns them;
	\end{itemize}
\item each function $\Cost_i$ is associated with the target set $F_i$ defined as follows:
	\begin{itemize}
	\item for all $i$, $1 \leq i \leq n$, $F_i = \{ \ell \in V \mid \ell \text{ is a literal of clause } C_i \}  \cup \{t_i\}$;
	\item $F_{n+1} = \{t_{n+1}\}$;
	\item $F_{n+2}=\{t_1,\ldots, t_n\}$.
	\end{itemize}
\end{itemize}

\begin{figure}[h!]
\centering
\scalebox{0.8}{
    \begin{tikzpicture}
    \node[draw] (Q1) at (0,0){$q_1$};
    \node[draw] (Q2) at (3,0){$q_2$};
    \node[draw] (Q3) at (6,0){$q_3$};
    \node (empty) at (7,0){$\ldots$};
    \node[draw] (Qm) at (8,0){$q_m$};
    \node[draw] (C1) at (11,0){$c_1$};
    \node (empty2) at (12.25,0){$\ldots$};
    \node[draw](Cn) at (13.5,0){$c_n$};
    \node[draw](T0) at (15,0){$t_{n+1}$};

    \node[draw] (x1) at (1.5,1.5){$x_1$};
    \node[draw] (nx1) at (1.5,-1.5){$\neg x_1$};

    \node[draw] (x2) at (4.5,1.5){$x_2$};
    \node[draw] (nx2) at (4.5,-1.5){$\neg x_2$};

    \node[draw] (xm) at (9.5, 1.5){$x_m$};
    \node[draw] (nxm) at (9.5, -1.5){$\neg x_m$};

    \node[draw] (T1) at (11,2){$t_1$};
    \node[draw] (Tn) at (13.5,2){$t_n$};

    \draw[->] (Q1) to (x1);
    \draw[->] (Q1) to (nx1);

    \draw[->] (x1) to (Q2);
    \draw[->] (nx1) to (Q2);

    \draw[->] (Q2) to (x2);
    \draw[->] (Q2) to (nx2);

    \draw[->] (x2) to (Q3);
    \draw[->] (nx2) to (Q3);

    \draw[->](Qm) to (xm);
    \draw[->](Qm) to (nxm);

    \draw[->](xm) to (C1);
    \draw[->](nxm) to (C1);

    \draw[->](C1) to (empty2);
    \draw[->](empty2) to (Cn);

    \draw[->](Cn) to (T0);

    \draw[->](C1) to (T1);
    \draw[->](Cn) to (Tn);

    \draw[->] (T1) edge [loop above] (T1);
    \draw[->] (Tn) edge [loop above] (Tn);
    \draw[->] (T0) edge [loop right] (T0);
    \end{tikzpicture}
}
\caption{Reduction from the formula $\psi$ to the quantitative reachability game $\mathcal{G}_{\psi}$}%
\label{figure:reachPSPACEh}
\end{figure}
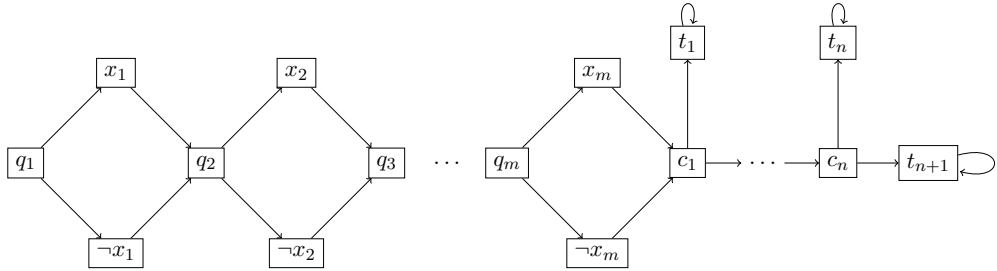

\begin{rem}%
\label{rem:winningStrat}
(1) Notice that a sequence $f$ of functions $f_k: {\{0,1\}}^{k-1} \rightarrow \{0,1\}$, with $k$ odd, $1 \leq k \leq m$, as presented above, can be translated into a strategy $\sigma_{n+1}$ of player~$n+1$ in the initialized game $(\mathcal{G}_{\psi},q_1)$, and conversely. Similarly, a sequence $g$ of functions $g_k: {\{0,1\}}^{k-1} \rightarrow \{0,1\}$, with $k$ even, $1 \leq k \leq m$ is nothing else than a strategy $\sigma_{n+2}$ of player~$n+2$. (2) Notice also that if $\rho$ is a play in $(\mathcal{G}_{\psi},q_1)$, then $\Cost_{n+1}(\rho) < +\infty$ if and only if $\Cost_{n+2}(\rho) = +\infty$. Moreover, suppose that $\rho$ visits $t_{n+1}$, then for all $i$, $1 \leq i \leq n$, {$\Cost_{i}(\rho) \leq 2\cdot m$} if and only if for all $i$, $1 \leq i \leq n$, $\rho$ visits a vertex that is a literal of $C_i$, and that is the case if and only if there is a valuation of all variables of $X$ that evaluates $\phi(X)$ to true.
\end{rem}

 Consider the game $\mathcal{G}_{\psi}$ and the bound $x = (2\cdot m, \ldots, 2\cdot m, 2\cdot m +n, +\infty)$. Both  can be constructed from $\psi$ in polynomial time. Let us now show that $\psi$ is true if and only if there exists an SPE in $(\mathcal{G}_{\psi},q_1)$ with cost $\leq x$.

\medskip

 $(\Rightarrow)$ Suppose that $\psi$ is true. Then there exists a sequence $f$ of functions $f_k: {\{0,1\}}^{k-1} \rightarrow \{0,1\}$, with $k$ odd, $1 \leq k \leq m$, such that for all sequences $g$ of functions $g_k: {\{0,1\}}^{k-1} \rightarrow \{0,1\}$, with $k$ even, $1 \leq k \leq m$, the valuation $\nu_{f,g}$ evaluates $\phi(X)$ to true. We define a strategy profile $\sigma$ as follows:
	\begin{itemize}
		\item for player~$n+1$, his strategy $\sigma_{n+1}$ is the strategy corresponding to the sequence $f$ (by Remark~\ref{rem:winningStrat});
		\item for player~$n+2$, his strategy is an arbitrary strategy $\sigma_{n+2}$; we denote by $g$ the corresponding sequence $g_k: {\{0,1\}}^{k-1} \rightarrow \{0,1\}$, with $k$ even, $1 \leq k \leq m$ 	(by Remark~\ref{rem:winningStrat});
		\item for each player~$i$, $1 \leq i \leq n$,
		\begin{itemize}
			\item if $hv \in \Hist_{i}(q_1)$ with $v = c_i$, is consistent with $\sigma_{n+1}$, then $\sigma_i(hv)= c_{i+1}$ if $i \neq n$ and $t_{n+1}$ otherwise
			\item else $\sigma_i(hv) = t_i$.
		\end{itemize}
	\end{itemize}

\noindent
Let us first prove that the play $\rho = \outcome{\sigma}{q_1}$ has a cost $\leq x = (2\cdot m,\ldots, 2 \cdot m, 2\cdot m +n,+\infty)$. By hypothesis, the valuation $\nu_{f,g}$ evaluates $\phi(X)$ to true, that is, it evaluates all clauses $C_i$ to true. Hence by Remark~\ref{rem:winningStrat}, $\rho$ visits a vertex of $F_i$ for all $i$, $1 \leq i \leq n$, and by definition of $\sigma$, $\rho$ eventually loops on $t_{n+1}$. It follows that $\Cost_i(\rho) \leq 2 \cdot m$ for all $i$, $1 \leq i \leq n$, $\Cost_{n+1}(\rho) \leq 2 \cdot m + n$, and $\Cost_{n+2}(\rho) = + \infty$. Hence $\Cost(\rho) \leq x$.

Let us now prove that $\sigma$ is an SPE, that is, for each history $hv \in \Hist(q_1)$, there is no one-shot deviating strategy in the subgame $(\mathcal G_{\psi\restriction h},v)$ that is profitable to the player who owns vertex $v$ (by Proposition~\ref{prop:SPE-vwSPE}). This is clearly true for all $v = t_i$, $1 \leq i \leq n+1$, since $t_i$ has only one outgoing edge. For the other vertices $v$, we study two cases:

\begin{itemize}

	\item $hv$ is consistent with $\sigma_{n+1}$:  Notice that $hv$ is maybe not consistent with $\sigma_{n+2}$, but with another arbitrary strategy $\sigma'_{n+2}$. Let $g'$ be the sequence corresponding to $\sigma'_{n+2}$ by Remark~\ref{rem:winningStrat}. By hypothesis, the valuation $\nu_{f,g'}$ evaluates $\phi(X)$ to true. Hence as explained previously for $\outcome{\sigma}{q_1}$, the cost of play $\rho = h \outcome{\rest{\sigma}{h}}{v}$ is such that $\Cost_i(\rho) \leq 2 \cdot m$ for all $i$, $1 \leq i \leq n$, $\Cost_{n+1}(\rho) \leq 2 \cdot m + n$, and $\Cost_{n+2}(\rho) = +\infty$. If $v$ belongs to player $i$, $1 \leq i \leq n$, this player has no incentive to deviate since he has already visited his target set along $h$ and thus cannot decrease his cost. If $v$ belongs to player $n+1$, a one-shot deviation will lead to a play eventually looping on $t_1$  by definition of $\sigma$, thus leading to a cost $+ \infty$ which is not profitable for player~$n+1$. Finally if $v$ belongs to player~$n+2$, a one-shot deviation will not decrease his cost by definition of $\sigma$ ($\sigma_{n+2}$ is arbitrary).

	\item $hv$ is not consistent with $\sigma_{n+1}$: Suppose that $v = c_k$. Then by definition of $\sigma$, the play $h \outcome{\rest{\sigma}{h}}{v}$ eventually loops on $t_k$ leading to a cost $\leq 2 \cdot m + k$ for player~$k$. In fact, if player $k$ has already seen his target set along $hv$, using a one-shot deviation in the subgame $(\mathcal G_{\psi\restriction h},v)$ leads to the same cost for him. Otherwise, it leads to a cost equal to $+\infty$:
	indeed, deviating here means going to the state $c_{k+1}$ (or $t_{n+1}$ if $k=n$, which leads to a cost of $+ \infty$ for player $n$), and since $hv$ is not consistent with $\sigma_{n+1}$, by definition of $\sigma_{k+1}$, player $k+1$ will choose to go to $t_{k+1}$.
	This player has thus no incentive to deviate.

Suppose that $v = q_k$. Then by definition of $\sigma$, the play $\rho = h \outcome{\rest{\sigma}{h}}{v}$ eventually loops on $t_1$. It follows that $\Cost_{n+1}(\rho) = +\infty$ and $\Cost_{n+2}(\rho) = 2\cdot m +1$. Due to the structure of the game graph, $2 \cdot m + 1$ is the smallest cost that player $n+2$ is able to obtain. So if $q_k \in V_{n+2}$, player~$n+2$ has no incentive to deviate. And if $q_k \in V_{n+1}$, player~$n+1$ could try to use a one-shot deviating strategy, however the resulting  play still eventually loops on $t_1$.
	\end{itemize}

\noindent
This proves that $\sigma$ is an SPE and we already showed that its cost was bounded by $x$.

\medskip

$(\Leftarrow)$ Suppose that there exists an SPE $\sigma$ in $(\mathcal{G}_{\psi},q_1)$ with outcome $\rho$ such that $\Cost(\rho) \leq x$. In particular $\Cost_{n+1}(\rho) < +\infty$. By Remark~\ref{rem:winningStrat}, it follows that $\Cost_{n+2}(\rho) = +\infty$. We have to prove that $\psi$ is true. To this end, consider the sequence $f$ of functions  $f_k: {\{0,1\}}^{k-1} \rightarrow \{0,1\}$, with $k$ odd, $1 \leq k \leq m$, that corresponds to strategy $\sigma_{n+1}$ of player~$n+1$ by Remark~\ref{rem:winningStrat}. Let us show that for all sequences $g$ of functions $g_k: {\{0,1\}}^{k-1} \rightarrow \{0,1\}$, with $k$ even, $1 \leq k \leq m$, the valuation $\nu_{f,g}$ evaluates $\phi(X)$ to true.

By contradiction assume that it is not the case for some sequence $g'$ and consider the related strategy $\sigma'_{n+2}$ of player~$n+2$ by Remark~\ref{rem:winningStrat}. Notice that $\sigma'_{n+2}$  is a finitely deviating strategy. Let us consider the outcome $\rho'$ of the strategy profile $(\sigma'_{n+2},\sigma_{-(n+2)})$ from $q_1$. As $\Cost_{n+2}(\rho) = +\infty$, we must have $\Cost_{n+2}(\rho') = +\infty$, otherwise $\sigma'_{n+2}$ is a profitable deviation for player~$n+2$ whereas $\sigma$ is an SPE\@. It follows that $\Cost_{n+1}(\rho') < +\infty$ by Remark~\ref{rem:winningStrat}, that is, $\rho'$ eventually loops on $t_{n+1}$.

Now recall that the valuation $\nu_{f,g'}$ evaluates $\phi(X)$ to false, which means that it evaluates some clause $C_k$ of $\phi(X)$ to false. Consider the history $hc_k < \rho'$. As strategy $\sigma'_{n+2}$ only acts on the left part of the underlying graph of  $\mathcal G_{\psi}$, we have $\rho' = \outcome{\sigma'_{n+2},\sigma_{-(n+2)}}{q_1} = h \outcome{\rest{\sigma}{h}}{c_k}$. In the subgame $(\mathcal{G}_{\psi\restriction h},c_k)$, the outcome of $\rest{\sigma}{h}$ gives a cost of $+\infty$ to player~$k$ because $\rho' = h \outcome{\rest{\sigma}{h}}{c_k}$ does not visit $t_k$ and $\nu_{f,g'}$ evaluates $C_k$ to false. In this subgame, player~$k$ has thus a profitable one-shot deviation that consists to move to $t_k$. It follows that $\sigma$ is not an SPE which is impossible. Then $\psi$ is true.\end{proof}

\section{Conclusion}

In this paper, we study multiplayer quantitative reachability games played on finite directed graphs, where the objective of each player is to reach a specified target set of vertices as quickly as possible. The players want to minimize their cost computed as the number of edges followed to reach their target set. For this class of games, we focus on the constraint existence problem for SPEs that is to decide whether there exists an SPE in which the cost of each player lies between two given bounds. It was known that there always exists an SPE in quantitative reachability games and that the constraint existence problem is decidable. We here prove that this problem is PSPACE-complete.

We summarize the main steps to obtain the PSPACE-membership:
\begin{enumerate}
\item Very weak SPEs are equivalent to SPEs for quantitative reachability games. This allows us to manipulate a simpler concept in our proofs and algorithm.
\item Instead of working with the original game, we work with its extended version such that the vertices of the game are enriched with the set of players that have already visited their target set. The original game and its extented game are equivalent with respect to SPE outcomes.
\item We provide a characterization of the set of plays that are SPE outcomes in the extended game as well as an algorithm to construct this set. This characterization makes use of a labeling function assigning an integer bound to each vertex. It imposes constraints on plays of the extended game such that the plays satisfying those constraints are exactly the SPE outcomes. Our algorithm builds this labeling as a fixpoint obtained iteratively from an initial labeling. We expect that this new concept of labeling will be helpful for studying SPEs for other classes of games.
\item Given the extended game and a labeling function, we construct the counter graph such that its infinite paths coincide with the plays that satisfy the constraints imposed by the labeling. Therefore when the labeling is the fixpoint computed by our algorithm, the infinite paths of the counter graph are exactly the SPE outcomes of the extended game. We establish the crucial property that this counter graph has an exponential size.
\item To establish the PSPACE membership of the constraint problem, a careful inspection of the computation shows that one can decide the existence of an infinite lasso of polynomial size in the counter graph that satisfies the bounds provided in the constraint problem.
\end{enumerate}

\noindent
The PSPACE-hardness of the constraint problem is easily obtained from the proof proposed for the constraint problem for qualitative reachability games.

\bigskip

\section*{Acknowledgments} We thank the anonymous reviewers for their helpful comments and feedback.

\bibliographystyle{alpha}
\bibliography{biblio}

\end{document}